\documentclass[11pt,onecolumn]{article}%
\usepackage{amsmath}
\usepackage{amsfonts}
\usepackage{amssymb}
\usepackage{graphicx}
\usepackage{fullpage}%
\providecommand{\U}[1]{\protect\rule{.1in}{.1in}}
\newtheorem{theorem}{Theorem}

\newtheorem{corollary}[theorem]{Corollary}

\newtheorem{lemma}[theorem]{Lemma}

\newtheorem{proposition}[theorem]{Proposition}

\newenvironment{proof}[1][Proof]{\noindent\textbf{#1.} }{\ \rule{0.5em}{0.5em}}

\begin{document}

\title{Forrelation: A Problem that Optimally Separates Quantum from Classical Computing}
\author{Scott Aaronson\thanks{Email: aaronson@csail.mit.edu. \ Supported by an NSF
Waterman Award.}\\MIT
\and Andris Ambainis\thanks{The research leading to these results has received
funding from the European Union Seventh Framework Programme (FP7/2007-2013)
under grant agreement n%
${{}^\circ}$
600700 (QALGO) and ERC Advanced Grant MQC.}\\University of Latvia}
\date{}
\maketitle

\begin{abstract}
We achieve essentially the largest possible separation between quantum and
classical query complexities. \ We do so using a property-testing problem
called \textsc{Forrelation}, where one needs to decide whether one Boolean
function is highly correlated with the Fourier transform of a second function.
\ This problem can be solved using $1$ quantum query, yet we show that any
randomized algorithm needs $\Omega(\sqrt{N}/\log N)$ queries (improving an
$\Omega(N^{1/4})$\ lower bound of Aaronson). \ Conversely, we show that this
$1$ versus $\widetilde{\Omega}(\sqrt{N})$ separation is optimal: indeed, any
$t$-query\ quantum algorithm whatsoever can be simulated by an $O\left(
N^{1-1/2t}\right)  $-query randomized algorithm. \ Thus, resolving an open
question of Buhrman et al.\ from 2002, there is no partial Boolean function
whose quantum query complexity is constant and whose randomized query
complexity is linear. \ We conjecture that a natural generalization of
\textsc{Forrelation}\ achieves the optimal $t$ versus $\Omega\left(
N^{1-1/2t}\right)  $ separation for all $t$. \ As a bonus, we show that this
generalization is $\mathsf{BQP}$-complete. \ This yields what's arguably the
simplest $\mathsf{BQP}$-complete problem yet known, and gives a second sense
in which \textsc{Forrelation}\ \textquotedblleft captures the maximum power of
quantum computation.\textquotedblright

\end{abstract}

\section{Introduction\label{INTRO}}

Since the work of Simon \cite{simon}\ and Shor \cite{shor}\ two decades ago,
we have had powerful evidence that quantum computers can achieve exponential
speedups over classical computers. \ Of course, for problems like
\textsc{Factoring}, these speedups are conjectural at present: we cannot rule
out that a fast classical factoring algorithm might exist. \ But in the
\textit{black-box model}, which captures most known quantum algorithms,
exponential and even larger speedups can be \textit{proved}. \ We know, for
example, that \textsc{Period-Finding}\ (a natural abstraction of the problem
solved by Shor's algorithm) is solvable\ with only $O\left(  1\right)
$\ quantum queries, but requires $N^{\Omega\left(  1\right)  }$\ classical
randomized queries, where $N$ is the number of input elements
\cite{cleve:lb,cfmw,mdw}. \ We also know that \textsc{Simon}'\textsc{s
Problem}\ is solvable with $O\left(  \log N\right)  $\ quantum queries, but
requires $\Omega(\sqrt{N})$\ classical queries; and that a similar separation
holds for the \textsc{Glued-Trees}\ problem introduced by Childs et
al.\ \cite{ccdfgs,fennerzhang}.\footnote{However, in all these cases the
queries are non-Boolean. \ If we insist on Boolean queries, then the quantum
query complexities get multiplied by an $O(\log N)$ factor.}

To us, these results raise an extremely interesting question:

\begin{itemize}
\item \textbf{\textquotedblleft The Speedup Question.\textquotedblright}
\ \textit{Within the black-box model, just how large of a quantum speedup is
possible? \ For example, could there be a function of }$N$\textit{ bits with a
quantum query complexity of }$1$\textit{, but a classical randomized query
complexity of }$\Omega\left(  N\right)  $\textit{?}
\end{itemize}

One may object: once we know that exponential and even larger quantum
speedups are possible in the black-box model, who cares about the exact limit?
\ In our view, the central reason to study the Speedup Question is that doing
so can help us better understand the nature of quantum speedups themselves.
\ For example, can all exponential quantum speedups be seen as originating
from a common cause? \ Is there a single problem or technique that captures
the advantages of quantum over classical query complexity, in much the same
way that random sampling could be said to capture the advantages of randomized
over deterministic query complexity?

As far as we know, the Speedup Question was first posed by Buhrman et
al.\ \cite{bfnr}\ around 2002, in their study of quantum property-testing.
\ Specifically, Buhrman et al.\ asked whether there is any property of $N$-bit
strings that exhibits a \textquotedblleft maximal\textquotedblright%
\ separation: that is, one that requires $\Omega(N)$\ queries to test
classically, but only $O\left(  1\right)  $\ quantumly. \ The best separation
they could find, based on Simon's problem, was \textquotedblleft
deficient\textquotedblright\ on both ends: it required $\Omega(\sqrt{N}%
)$\ queries to test classically, and $O(\log N\log\log N)$ quantumly.

Since then, there has been only sporadic progress on the Speedup Question.
\ In 2009, Aaronson \cite{aar:ph} introduced the \textsc{Forrelation}%
\ problem---a problem that we will revisit in this paper---and showed that it
was solvable with only $1$ quantum query, but required $\Omega(N^{1/4}%
)$\ classical randomized queries. \ In 2010, Chakraborty et al.\ \cite{cfmw}%
\ argued that \textsc{Period-Finding}\ gives a different example of an
$O(1)$\ versus $\widetilde{\Omega}(N^{1/4})$\ quantum/classical gap; there,
however, we only get an $O(1)$-query quantum algorithm if we allow non-Boolean queries.

Earlier, in 2001, de Beaudrap, Cleve, and Watrous \cite{beaudrap}\ had given
what they described as a black-box problem that was solvable with $1$ quantum
query, but that required $\Omega(N^{1/4})$\ or $\Omega(\sqrt{N})$ classical
randomized queries (depending on how one defines the \textquotedblleft input
size\textquotedblright\ $N$). \ However, de Beaudrap et al.\ were not working
within the usual model of quantum query complexity. \ Normally, one provides
\textquotedblleft black-box access\textquotedblright\ to a function $f$,
meaning that the quantum algorithm can apply a unitary transformation that
maps basis states of the form $\left\vert x,y\right\rangle $ to basis states
of the form $\left\vert x,y\oplus f(x)\right\rangle $ (or $\left\vert
x\right\rangle $\ to $\left(  -1\right)  ^{f\left(  x\right)  }\left\vert
x\right\rangle $, if $f$\ is Boolean). \ By contrast, for their separation, de
Beaudrap et al.\ had to assume the ability to map basis states of the form
$\left\vert x,y\right\rangle $\ to basis states of the form $\left\vert
x,\pi(y+sx)\right\rangle $, for some unknown permutation $\pi$\ and hidden
shift $s$.

\subsection{Our Results\label{RESULTS}}

This paper has two main contributions---the largest quantum black-box speedup
yet known, and a proof that that speedup is essentially optimal---as well as
many smaller related contributions.

\subsubsection{Maximal Quantum/Classical Separation\label{RESULT1}}

In Section \ref{GAP}, we undertake a detailed study of the
\textsc{Forrelation} problem, which Aaronson \cite{aar:ph} introduced for a
different purpose than the one that concerns us here (he was interested in an
oracle separation between $\mathsf{BQP}$\ and the polynomial
hierarchy).\footnote{Also, in \cite{aar:ph}, the problem was called
\textquotedblleft Fourier Checking.\textquotedblright} \ In
\textsc{Forrelation}, we are given access to two Boolean functions
$f,g:\left\{  0,1\right\}  ^{n}\rightarrow\left\{  -1,1\right\}  $. \ We want
to estimate the amount of correlation between $f$\ and the Fourier transform
of $g$---that is, the quantity%
\[
\Phi_{f,g}:=\frac{1}{2^{3n/2}}\sum_{x,y\in\left\{  0,1\right\}  ^{n}}f\left(
x\right)  \left(  -1\right)  ^{x\cdot y}g\left(  y\right)  .
\]
It is not hard to see that $\left\vert \Phi_{f,g}\right\vert \leq1$\ for all
$f,g$. \ The problem is to decide, say, whether $\left\vert \Phi
_{f,g}\right\vert \leq\frac{1}{100}$\ or $\Phi_{f,g}\geq\frac{3}{5}$, promised
that one of these is the case.\footnote{The reason for the asymmetry---i.e.,
for promising that $\Phi_{f,g}$\ is positive if its absolute value is large,
but not if its absolute value is small---is a bit technical. \ On the one
hand, we want the \textquotedblleft unforrelated\textquotedblright\ case to
encompass almost all randomly-chosen functions $f,g$. \ On the other hand, we
also want \textsc{Forrelation}\ to be solvable using only $1$ quantum query.
\ If we had promised $\left\vert \Phi_{f,g}\right\vert \geq\frac{3}{5}$,
rather than $\Phi_{f,g}\geq\frac{3}{5}$, then we would only know a $2$-query
quantum algorithm. \ In any case, none of these choices make a big difference
to our results.} \ Here and throughout this paper, the \textquotedblleft input
size\textquotedblright\ is taken to be $N:=2^{n}$.

One can give (see Section \ref{PRELIM}) a quantum algorithm that solves
\textsc{Forrelation}, with bounded probability of error, using only $1$
quantum query. \ Intuitively, however, the property of $f$\ and $g$\ being
\textquotedblleft forrelated\textquotedblright\ (that is, having large
$\Phi_{f,g}$\ value) is an extremely global property, which should not be
apparent to a classical algorithm until it has queried a significant fraction
of the entire truth tables of $f$\ and $g$. \ And indeed, improving an
$\Omega\left(  N^{1/4}\right)  $\ lower bound of Aaronson \cite{aar:ph}, in
Section \ref{GAP}\ we show the following:

\begin{theorem}
Any classical randomized algorithm for \textsc{Forrelation} must make
$\Omega(\frac{\sqrt{N}}{\log N})$ queries.\label{thm1}
\end{theorem}

Theorem \ref{thm1} yields \textit{the largest quantum versus classical
separation yet known} in the black-box model. \ As we will show in Appendix
\ref{PROP}, Theorem \ref{thm1} also implies the largest
\textit{property-testing} separation yet known---for with some work, one can
recast \textsc{Forrelation}\ (or rather, its negation) as a property that is
testable with only $1$ query quantumly, but that requires $\Omega(\frac
{\sqrt{N}}{\log N})$\ queries to test classically.

We deduce Theorem \ref{thm1} as a consequence of a more general result:
namely, a lower bound on the classical query complexity of a problem called
\textsc{Gaussian Distinguishing}. \ Here we are given oracle access to a
collection of $\mathcal{N}\left(  0,1\right)  $\ real Gaussian\ random
variables, $x_{1},\ldots,x_{M}$. \ We are asked to decide whether the
variables are all independent, or alternatively, whether they lie in a known
low-dimensional subspace of $\mathbb{R}^{M}$: one that induces a covariance of
at most $\varepsilon$\ between each pair of variables, while keeping each
variable an $\mathcal{N}\left(  0,1\right)  $\ Gaussian individually. \ We
show the following:

\begin{theorem}
\textsc{Gaussian Distinguishing} requires $\Omega\left(  \frac{1/\varepsilon
}{\log(M/\varepsilon)}\right)  $\ classical randomized queries.\label{thm1.5}
\end{theorem}

Theorem \ref{thm1}\ is then simply a (discretized) special case of Theorem
\ref{thm1.5},\ with $M=2N$\ and $\varepsilon=1/\sqrt{N}$. \ Beyond that, it
seems to us that Theorem \ref{thm1.5}\ could have independent applications in
statistics and machine learning.

\subsubsection{Proof of Optimality\label{RESULT2}}

In Section \ref{SIM}, we show that the quantum/classical query complexity
separation achieved by the \textsc{Forrelation}\ problem is close to the best
possible. \ More generally:

\begin{theorem}
\label{thm2}Let $Q$ be any quantum algorithm that makes $t=O\left(  1\right)
$\ queries to an $N$-bit string $X\in\left\{  0,1\right\}  ^{N}$. \ Then we
can estimate $\Pr\left[  Q\text{ accepts }X\right]  $, to constant additive
error and with high probability, by making only $O(N^{1-1/2t})$ classical
randomized queries to $X$.\footnote{The reason for the condition $t=O\left(
1\right)  $\ is that, in the bound $O(N^{1-1/2t})$, the big-$O$ hides a
multiplicative factor of $\exp\left(  t\right)  $. \ Thus, we can obtain a
nontrivial upper bound on query complexity as long as $t=o(\sqrt{\log N})$.}
\ Moreover, the randomized queries are nonadaptive.
\end{theorem}

So for example, every $1$-query quantum algorithm can be simulated by an
$O(\sqrt{N})$-query classical randomized algorithm, every $2$-query quantum
algorithm can be simulated by an $O(N^{3/4})$-query randomized algorithm, and
so on. \ Theorem \ref{thm2} resolves the open problem of Buhrman et
al.\ \cite{bfnr} in the negative: it shows that there is no problem
(property-testing or otherwise) with a constant\ versus
linear\ quantum/classical query complexity gap. \ Theorem \ref{thm2}\ does not
rule out the possibility of an $O\left(  \log N\right)  $\ versus
$\widetilde{\Omega}\left(  N\right)  $\ gap, and indeed, we conjecture that
such a gap is possible.

Once again, we deduce Theorem \ref{thm2} as a consequence of a more general
result, which might have independent applications to classical sublinear
algorithms. \ Namely:

\begin{theorem}
\label{thm2.5}Every degree-$k$ real polynomial $p:\left\{  -1,1\right\}
^{N}\rightarrow\mathbb{R}$\ that is

\begin{enumerate}
\item[(i)] bounded in $\left[  -1,1\right]  $\ at every Boolean point, and

\item[(ii)] \textquotedblleft block-multilinear\textquotedblright\ (that is,
the variables can be partitioned into $k$ blocks, such that every monomial is
the product of one variable from each block),
\end{enumerate}

\noindent can be approximated to within $\pm\varepsilon$, with high
probability, by nonadaptively querying only $O(\left(  N/\varepsilon
^{2}\right)  ^{1-1/k})$\ of the variables.
\end{theorem}

In the statement of Theorem \ref{thm2.5}, we strongly conjecture that
condition (ii) can be removed. \ If so, then we would obtain a sublinear
algorithm to estimate any bounded, constant-degree real polynomial. \ In
Appendix \ref{POLY}, we show that condition (ii) can indeed be removed in the
special case $k=2$.

\subsubsection{$k$-fold\label{RESULT3} \textsc{Forrelation}}

In Section \ref{BQP}, we study a natural generalization of
\textsc{Forrelation}. \ In $k$-fold \textsc{Forrelation}, we are given access
to $k$ Boolean functions $f_{1},\ldots,f_{k}:\left\{  0,1\right\}
^{n}\rightarrow\left\{  -1,1\right\}  $. \ We want to estimate the
\textquotedblleft twisted sum\textquotedblright%
\[
\Phi_{f_{1},\ldots,f_{k}}:=\frac{1}{2^{\left(  k+1\right)  n/2}}\sum
_{x_{1},\ldots,x_{k}\in\left\{  0,1\right\}  ^{n}}f_{1}\left(  x_{1}\right)
\left(  -1\right)  ^{x_{1}\cdot x_{2}}f_{2}\left(  x_{2}\right)  \left(
-1\right)  ^{x_{2}\cdot x_{3}}\cdots\left(  -1\right)  ^{x_{k-1}\cdot x_{k}%
}f_{k}\left(  x_{k}\right)  .
\]
It is not hard to show that $\left\vert \Phi_{f_{1},\ldots,f_{k}}\right\vert
\leq1$ for all $f_{1},\ldots,f_{k}$. \ The problem is to decide, say, whether
$\left\vert \Phi_{f_{1},\ldots,f_{k}}\right\vert \leq\frac{1}{100}$\ or
$\Phi_{f_{1},\ldots,f_{k}}\geq\frac{3}{5}$, promised that one of these is the case.

One can give (see Section \ref{PRELIM}) a quantum algorithm that solves
$k$-fold \textsc{Forrelation}, with bounded error probability, using only
$\left\lceil k/2\right\rceil $ quantum queries. \ In Section \ref{BQP}, we
show, conversely, that $k$-fold \textsc{Forrelation}\ \textquotedblleft
captures the full power of quantum computation\textquotedblright:

\begin{theorem}
\label{thm3}If $f_{1},\ldots,f_{k}$\ are described explicitly (say, by
circuits to compute them), and $k=\operatorname*{poly}\left(  n\right)  $,
then $k$-fold \textsc{Forrelation}\ is a $\mathsf{BQP}$-complete\ promise problem.
\end{theorem}

We do not know of any complete problem for quantum computation that is more
self-contained than this. \ Not only can\ one\ state the $k$-fold
\textsc{Forrelation}\ problem without any notions from quantum mechanics, one
does not need any nontrivial mathematical notions, like the condition number
of a matrix or the Jones polynomial of a knot.

We conjecture, moreover, that $k$-fold \textsc{Forrelation}\ achieves the
optimal $k/2$\ versus $\widetilde{\Omega}(N^{1-1/k})$ quantum/classical query
complexity separation for all even $k$. \ If so, then there are \textit{two}
senses in which $k$-fold \textsc{Forrelation}\ captures the power of quantum computation.

\subsubsection{Other Results\label{OTHER}}

The paper also includes several other results.

In Appendix \ref{SAMPREL}, we study the largest possible quantum/classical
separations that are achievable for \textit{approximate sampling} and
\textit{relation} problems. \ We show that there exists a sampling
problem---namely, \textsc{Fourier Sampling}\ of a Boolean function---that is
solvable with $1$ quantum query, but requires $\Omega(N/\log N)$\ classical
queries. \ By our previous results, this exceeds the largest quantum/classical
gap that is possible for decision problems.

In Appendix \ref{POLY}, we generalize our result that every $1$-query quantum
algorithm can be simulated using $O(\sqrt{N})$\ randomized queries, to show
that \textit{every bounded degree-}$2$\textit{ real polynomial} $p:\left\{
-1,1\right\}  ^{N}\rightarrow\left[  -1,1\right]  $\ can be estimated using
$O(\sqrt{N})$\ randomized queries. \ We conjecture that this can be
generalized, to show that every bounded degree-$k$ real polynomial can be
estimated using $O(N^{1-1/k})$\ randomized queries.

In Appendix \ref{KFOLDLB}, we extend our $\Omega(\frac{\sqrt{N}}{\log N}%
)$\ randomized lower bound for the \textsc{Forrelation}\ problem, to show a
$\Omega(\frac{\sqrt{N}}{\log^{7/2}N})$\ lower bound for $k$-fold
\textsc{Forrelation}\ for any $k\geq2$. \ We conjecture that the right lower
bound is $\widetilde{\Omega}(N^{1-1/k})$, but even generalizing our
$\widetilde{\Omega}(\sqrt{N})$\ lower bound to the $k$-fold case\ is nontrivial.

\subsection{Techniques\label{TECH}}

\subsubsection{Randomized Lower Bound\label{RLB}}

Proving that any randomized algorithm for \textsc{Forrelation}\ requires
$\Omega(\frac{\sqrt{N}}{\log N})$\ queries is surprisingly involved. \ As we
mentioned in Section \ref{RESULT1}, the first step, following the work of
Aaronson \cite{aar:ph}, is to convert \textsc{Forrelation}\ into an analogous
problem involving real Gaussian variables. \ In \textsc{Real Forrelation}, we
are given oracle access to two real functions $f,g:\left\{  0,1\right\}
^{n}\rightarrow\mathbb{R}$, and are promised either that (i) every\ $f\left(
x\right)  $\ and $g\left(  y\right)  $ value is an independent $\mathcal{N}%
\left(  0,1\right)  $\ Gaussian, or else (ii) every $f\left(  x\right)
$\ value is an independent $\mathcal{N}\left(  0,1\right)  $\ Gaussian,\ while
every $g\left(  y\right)  $\ value equals $\hat{f}\left(  y\right)  $\ (i.e.,
the Fourier transform of $f$\ evaluated at $y$). \ The problem is to decide
which holds. \ Using a rounding reduction, we show that any query complexity
lower bound for \textsc{Real Forrelation}\ implies the same lower bound for
\textsc{Forrelation}\ itself.

Making the problem continuous allows us to adopt a geometric perspective. \ In
this perspective, we are given oracle access to a real vector $v\in
\mathbb{R}^{2N}$, whose $2N$\ coordinates consist of all values $f\left(  x\right)
$\ and all values $g\left(  y\right)  $\ (recall that $N=2^{n}$). \ We
are trying to distinguish the case where $v$ is simply an $\mathcal{N}\left(
0,1\right)  ^{2N}$ Gaussian, from the case where $v$\ is confined to an
$N$-dimensional subspace of $\mathbb{R}^{2N}$---namely, the subspace defined
by $g=\hat{f}$. \ Now, suppose that values $f\left(  x_{1}\right)
,\ldots,f\left(  x_{t}\right)  $\ and $g\left(  y_{1}\right)  ,\ldots,g\left(
y_{u}\right)  $\ have already been queried. \ Then we can straightforwardly
calculate the Bayesian posterior probabilities for being in case (i) or case
(ii). \ For case (i), the probability turns out to depend solely on the
squared $2$-norm of the vector of empirical data seen so far:%
\[
\Pr\left[  \text{case (i)}\right]  \propto\exp\left(  -\frac{\Delta_{\text{i}%
}}{2}\right)  ,
\]
where%
\[
\Delta_{\text{i}}=f\left(  x_{1}\right)  ^{2}+\cdots+f\left(  x_{t}\right)
^{2}+g\left(  y_{1}\right)  ^{2}+\cdots+g\left(  y_{u}\right)  ^{2}.
\]
For case (ii), by contrast, the probability is proportional to $\exp
(-\Delta_{\text{ii}}/2)$, where $\Delta_{\text{ii}}$\ is the minimum squared
$2$-norm of any point $f\in\mathbb{R}^{N}$ compatible with all the data seen
so far, as well as with the linear constraint $g=\hat{f}$. \ Let
$\mathcal{V}=\left\{  \left\vert 1\right\rangle ,\ldots,\left\vert
N\right\rangle ,|\hat{1}\rangle,\ldots,|\hat{N}\rangle\right\}  $\ be the set
of $2N$ unit vectors in $\mathbb{R}^{N}$\ that consists of all $N$ elements of
the standard basis, together with all $N$ elements of the Fourier basis.
\ Then $\Delta_{\text{ii}}$, in turn, can be calculated using a process of
Gram-Schmidt orthogonalization, on the vectors in $\mathcal{V}$\ corresponding
to the $f$-values and $g$-values that have been queried so far.

Our goal is to show that, with high probability, $\Delta_{\text{i}}$\ and
$\Delta_{\text{ii}}$ \textit{remain close to each other}, even after a large
number of queries have been made---meaning that the algorithm has not yet
succeeded in distinguishing case (i) from case (ii) with non-negligible bias.
\ To show this, the central fact we rely on is that the vectors in
$\mathcal{V}$\ are \textit{nearly-orthogonal}: that is, for all $\left\vert
v\right\rangle ,\left\vert w\right\rangle \in\mathcal{V}$, we have $\left\vert
\left\langle v|w\right\rangle \right\vert \leq\frac{1}{\sqrt{N}}$.
\ Intuitively, this means that, if we restrict attention to any small subset
$S$ of \thinspace$f$-values and $g$-values, then while correlations exist
among those values, the correlations are \textit{weak}: \textquotedblleft to a
first approximation,\textquotedblright\ we have simply asked for the
projections of a Gaussian vector onto $\left\vert S\right\vert $\ orthogonal
directions, and have therefore received $\left\vert S\right\vert
$\ uncorrelated answers.

From this perspective, the key question is: how many values can we query until
the \textquotedblleft orthogonal approximation\textquotedblright\ breaks down
(meaning that we notice the correlations)? \ In his previous work, Aaronson
\cite{aar:ph}\ showed that the approximation holds until $\Omega\left(
N^{1/4}\right)  $ queries are made. \ Indeed, he proved a stronger statement:
even if the $x$'s and $y$'s are chosen \textit{nondeterministically}, still
$\Omega\left(  N^{1/4}\right)  $ values\ must be revealed until we have a
\textit{certificate} showing that we are in case (i) or case (ii) with high probability.

To improve the lower bound from $\Omega\left(  N^{1/4}\right)  $\ to the
optimal $\widetilde{\Omega}(\sqrt{N})$, there are several hurdles to overcome.

Aaronson had assumed, conservatively, that the deviations from orthogonality
all \textquotedblleft pull in the same direction.\textquotedblright\ \ As a
first step, we notice instead that the deviations follow an unbiased random
walk, with some positive and others negative---the martingale property arising
from the fact that the algorithm can control which $x$'s and $y$'s to query,
but not the values of $f\left(  x\right)  $\ and $g\left(  y\right)  $. \ We
then use a Gaussian generalization of Azuma's inequality to upper-bound the
sum of the deviations. \ Doing this improves the lower bound from
$\Omega\left(  N^{1/4}\right)  $\ to $\widetilde{\Omega}\left(  N^{1/3}%
\right)  $, but we then hit an apparent barrier.

In this work, we explain the $\Omega\left(  N^{1/3}\right)  $\ barrier, by
exhibiting a \textquotedblleft model problem\textquotedblright\ that is
extremely similar to \textsc{Real Forrelation}\ (in particular, has exactly
the same near-orthogonality property), yet is solvable with only $O\left(
N^{1/3}\right)  $ queries, by exploiting adaptivity. \ However, we then break
the barrier, by using the fact that, for \textsc{Real Forrelation}\ (but
\textit{not} for the model problem), the total number of vectors in
$\mathcal{V}$\ is only $N^{O\left(  1\right)  }$. \ This fact lets us use the
Gaussian Azuma's inequality a second time, to upper-bound not only the
\textit{sum} of all the deviations from orthogonality, but the individual
deviations themselves. \ Implementing this yields a lower bound of
$\widetilde{\Omega}\left(  N^{2/5}\right)  $: better than $\widetilde{\Omega
}\left(  N^{1/3}\right)  $, but still not all the way up to $\widetilde{\Omega
}(\sqrt{N})$. \ However, we then notice that we can apply Azuma's inequality
\textit{recursively}---once for each layer of the Gram-Schmidt
orthogonalization process---to get better and better upper bounds on the
deviations from orthogonality. \ Doing so gives us a sequence of lower
bounds\ $\widetilde{\Omega}\left(  N^{3/7}\right)  $, $\widetilde{\Omega
}\left(  N^{4/9}\right)  $, $\widetilde{\Omega}\left(  N^{5/11}\right)  $,
etc., with the ultimate limit of the process being $\Omega(\frac{\sqrt{N}%
}{\log N})$.

\subsubsection{Randomized Upper Bound\label{RUB}}

Why did we have to work so hard to prove a\ $\widetilde{\Omega}(\sqrt{N}%
)$\ lower bound on the randomized query complexity of \textsc{Forrelation}?
\ Our other main result provides one possible explanation: namely, we are here
scraping up against the \textquotedblleft ceiling\textquotedblright\ of the
possible separations between randomized and quantum query complexity. \ In
particular, \textit{any} quantum algorithm that makes $1$ query to a Boolean
input $X\in\left\{  0,1\right\}  ^{N}$, can be simulated by a randomized
algorithm (in fact, a nonadaptive randomized algorithm) that makes $O(\sqrt
{N})$ queries to $X$. \ More generally, any quantum algorithm that makes
$t=O\left(  1\right)  $\ queries to $X$, can be simulated by a nonadaptive
randomized algorithm that makes $O(N^{1-1/2t})$\ queries to $X$.

The proof of this result consists of three steps. \ The first involves the
simulation of quantum algorithms by low-degree polynomials. \ In 1998, Beals
et al.\ \cite{bbcmw} famously observed that, if a quantum algorithm makes $t$
queries to a Boolean input $X\in\left\{  -1,1\right\}  ^{N}$, then $p\left(
X\right)  $, the probability that the algorithm accepts $X$, can be written as
a multilinear polynomial in $X$ of degree at most $2t$. \ We extend this
result of Beals et al., in a way that might be of independent interest for
quantum lower bounds. \ Namely, we observe that every $t$-query quantum
algorithm gives rise, not merely to a multilinear polynomial, but to a
\textit{block-multilinear} polynomial. \ By this, we mean a degree-$2t$
polynomial $q$\ that takes as input\ $2t$ blocks of $N$ variables each, and
whose every monomial contains exactly one variable from each block. \ If we
repeat the input $X\in\left\{  -1,1\right\}  ^{N}$ across all $2t$%
\ blocks,\ then $q\left(  X,\ldots,X\right)  $\ represents the quantum
algorithm's acceptance probability on $X$. \ However, the key point is that
$q\left(  Y\right)  $\ is bounded in $\left[  -1,1\right]  $\ for \textit{any}
Boolean input $Y\in\left\{  -1,1\right\}  ^{2tN}$.

This leads to a new complexity measure for Boolean functions $f$: the
\textit{block-multilinear approximate degree}
$\widetilde{\operatorname*{bmdeg}}\left(  f\right)  $, which lower-bounds the
quantum query complexity $\operatorname*{Q}\left(  f\right)  $ just as
$\widetilde{\deg}\left(  f\right)  $\ does, but which might provide a tighter
lower bound in some cases. \ (Indeed, we do not even know whether there is any
asymptotic separation between $\operatorname*{Q}\left(  f\right)  $\ and
$\widetilde{\operatorname*{bmdeg}}\left(  f\right)  $, whereas Ambainis
\cite{ambainis:deg}\ showed an asymptotic separation between
$\operatorname*{Q}\left(  f\right)  $\ and\ $\widetilde{\deg}\left(  f\right)
$.)

Once we have our quantum algorithm's acceptance probability in the form of a
block-multilinear polynomial $q$, the second step is to \textit{preprocess}
$q$, to make it easier to estimate using random sampling. \ The basic problem
is that $q$ might be highly \textquotedblleft unbalanced\textquotedblright:
certain variables might be hugely influential. Such variables are essential to
query, but examining the form of $q$ does not make it obvious which variables
these are. \ To deal with this, we repeatedly perform an operation called
\textquotedblleft variable-splitting,\textquotedblright\ which consists of
identifying an influential variable $x_{i}$, then replacing every occurrence
of $x_{i}$\ in\ $q$ by $\frac{1}{m}\left(  x_{i,1}+\cdots+x_{i,m}\right)  $,
where $x_{i,1},\ldots,x_{i,m}$\ are newly-created variables set equal to
$x_{i}$. \ The point of doing this is that each $x_{i,j}$\ will be less
influential in $q$ than $x_{i}$\ itself was, thereby yielding a more balanced
polynomial. \ We show that variable-splitting can achieve the desired balance
by introducing at most $\exp\left(  t\right)  \cdot O\left(  N\right)  $ new
variables, which is linear in $N$ for constant $t$.

Once we have a balanced polynomial $q$, the last step is to give a
query-efficient randomized algorithm to estimate its value. \ Our algorithm is
the simplest one imaginable: we simply choose $O(N^{1-1/2t})$\ variables
uniformly at random, query them, then form an estimate $\widetilde{q}$\ of
$q$\ by summing only those monomials all of whose variables were queried.
\ The hard part is to prove that this estimator \textit{works}---i.e., that
its variance is bounded. \ The proof of this makes heavy use of the balancedness 
property that was ensured by the preprocessing step.

Examining our estimation algorithm, an obvious question is whether it was
essential that $q$\ be block-multilinear, or whether the algorithm could be
extended to \textit{all} bounded low-degree polynomials. \ In Appendix
\ref{POLY}, we take a first step toward answering that question, by giving an
$O(\sqrt{N})$-query randomized algorithm to estimate any bounded degree-$2$
polynomial in $N$ Boolean variables. \ Once we drop block-multilinearity, our
variable-splitting procedure no longer works, so we rely instead on
Fourier-analytic results of Dinur et al.\ \cite{dfko} to identify influential
variables which we then split.

\subsubsection{Other Results\label{OTHERTECH}}

$\mathsf{BQP}$\textbf{-Completeness.} \ The proof that the $k$-fold
\textsc{Forrelation}\ problem is $\mathsf{P{}romiseBQP}$-complete is simple,
once one has the main idea. \ The sum that defines $k$-fold \textsc{Forrelation}
is, itself, an output amplitude for a particular kind of quantum circuit,
which consists entirely of Hadamard and $f$-phase gates (i.e., gates that map
$\left\vert x\right\rangle $\ to $\left(  -1\right)  ^{f\left(  x\right)
}\left\vert x\right\rangle $ for some Boolean function $f$). \ Since the
Hadamard and CCPHASE gates (corresponding to $f\left(  x,y,z\right)  =xyz$)
are known to be universal for quantum computation, one might think that our
work is done. \ The difficulty is that the quantum circuit for $k$-fold
\textsc{Forrelation}\ contains a Hadamard gate on every qubit, between every
pair of $f$-phase gates, \textit{whether we wanted Hadamards there or not}.
\ Thus, if we want to encode an arbitrary quantum circuit, then we need some
way of \textit{canceling} unwanted Hadamards, while leaving the wanted ones.
\ We achieve this via a gadget construction.

\textbf{Separation for Sampling Problems.} \ To achieve a $1$\ versus
$\Omega(N/\log N)$\ quantum/classical query complexity separation for a
sampling problem, we consider \textsc{Fourier Sampling}: the problem, given
oracle access to a Boolean function $f:\left\{  0,1\right\}  ^{n}%
\rightarrow\left\{  -1,1\right\}  $, of outputting a string $y\in\left\{
0,1\right\}  ^{n}$\ with probability approximately equal to $\hat{f}\left(
y\right)  ^{2}$. \ This problem is trivially solvable with $1$ quantum query,
but proving a $\Omega(N/\log N)$\ classical lower bound takes a few pages of
work. \ The basic idea is to concentrate on the probability of a
\textit{single} string---say, $y=0^{n}$---being output. \ Using a binomial
calculation, we show that this probability cannot depend on $f$'s truth table
in the appropriate way unless $\Omega(N/\log N)$\ function values are queried.

\textbf{Lower Bound for }$\mathbf{k}$\textbf{-Fold }\textsc{Forrelation}.
\ Once we have a $\Omega(\frac{\sqrt{N}}{\log N})$ randomized lower bound for
\textsc{Forrelation}, one might think it would be trivial to prove the same
lower bound for $k$-fold \textsc{Forrelation}: just reduce one to the other!
\ However, \textsc{Forrelation}\ does not embed in any clear way as a
subproblem of $k$-fold \textsc{Forrelation}. \ On the other hand, given an
instance of $k$-fold \textsc{Forrelation}, suppose we \textquotedblleft give
away for free\textquotedblright\ the complete truth tables of all but two of
the functions. \ In that case, we show that the induced subproblem on the
remaining two functions is an instance of \textsc{Gaussian Distinguishing}\ to
which, with high probability, our lower bound techniques can be applied.
\ Pursuing this idea leads to our $\Omega(\frac{\sqrt{N}}{\log^{7/2}N}%
)$\ lower bound on the randomized query complexity of $k$-fold
\textsc{Forrelation}, for all $k\geq2$.

\textbf{Property-Testing Separation.} \ To turn our quantum versus classical
separation for the \textsc{Forrelation}\ problem into a property-testing
separation, we need to prove two interesting statements. \ The first is that
function pairs $\left\langle f,g\right\rangle $ that are far in Hamming
distance from the set of all pairs with low forrelation, actually have high
forrelation. \ The second is that \textquotedblleft generic\textquotedblright%
\ function pairs $\left\langle f,g\right\rangle $\ and $\left\langle
f^{\prime},g^{\prime}\right\rangle $ that have small Hamming distance from one
another, are close in their forrelation values as well. \ In fact, we will
prove both of these statements for the general case of $k$-fold
\textsc{Forrelation}.

\subsection{Discussion\label{DISC}}

To summarize, this paper proves the largest separation between classical and
quantum query complexities yet known, and it also proves that that separation
is in some sense optimal. These results put us in a position to pose an
intriguing open question:

\begin{quotation}
\noindent Among all the problems that admit a superpolynomial quantum
speedup,\ is there any whose classical randomized query complexity is
$\gg\sqrt{N}$?
\end{quotation}

Strikingly, if we look at the known problems with superpolynomial quantum
speedups, for every one of them the classical randomized lower bound seems to
hit a \textquotedblleft ceiling\textquotedblright\ at $\sqrt{N}$. \ Thus,
\textsc{Simon}'\textsc{s Problem} has quantum query complexity $O\left(  \log
N\right)  $\ and randomized query complexity $\widetilde{\Theta}(\sqrt{N})$;
the \textsc{Glued-Trees}\ problem of Childs et al.\ \cite{ccdfgs}\ has quantum
query complexity $\log^{O\left(  1\right)  }(N)$\ and randomized query
complexity $\widetilde{\Theta}(\sqrt{N})$;\footnote{The randomized lower bound
for \textsc{Glued-Trees}\ proved by Childs et al.\ \cite{ccdfgs}\ was only
$\Omega(N^{1/6})$. \ However, Fenner and Zhang \cite{fennerzhang}\ improved
the lower bound to $\Omega(N^{1/3})$; and if we allow a success probability
that is merely (say) $1/3$, rather than exponentially small, then their bound
can be improved further, to $\Omega(\sqrt{N})$. \ In the other direction, we
are indebted to Shalev Ben-David for proving that \textsc{Glued-Trees}\ can be
solved deterministically using only $O(\sqrt{N}\log N)$\ queries (or
$O(\sqrt{N}\log^{2}N)$, if the queries are required to be Boolean). \ For his
proof, see
http://cstheory.stackexchange.com/questions/25279/the-randomized-query-complexity-of-the-conjoined-trees-problem}
and \textsc{Forrelation}\ has quantum query complexity $1$\ and randomized
query complexity $\widetilde{\Theta}(\sqrt{N})$.

If we insist on making the randomized query complexity $\Omega(N^{1/2+c})$,
for some $c>0$, and then try to minimize the quantum query complexity, then
the best thing we know how to do is to take the OR of $N^{2c}$\ independent
instances of \textsc{Forrelation}, each of size $N^{1-2c}$. \ This gives us a
problem whose quantum query complexity is $\Theta(N^{c})$,\footnote{Here the
upper bound comes from combining Grover's algorithm with the
\textsc{Forrelation}\ algorithm: the \textquotedblleft
na\"{\i}ve\textquotedblright\ way of doing this would produce an additional
$\log N$\ factor for error reduction, but it is well-known that that
log\ factor can be eliminated \cite{reichardt:reflections}. \ Meanwhile, the
lower bound comes from the optimality of Grover's algorithm.} and whose
classical randomized query complexity is $\widetilde{\Theta}(N^{1/2+c}%
)$.\footnote{Here the upper bound comes from simply taking the best randomized
\textsc{Forrelation}\ algorithm, which uses $O(\sqrt{N^{1-2c}})$\ queries, and
running it $N^{2c}$\ times, with an additional $\log N$\ factor for error
reduction. \ Meanwhile, the lower bound comes from combining this paper's
$\Omega(\sqrt{N}/\log N)$\ lower bound for \textsc{Forrelation}, with a
general result stating that the randomized query complexity of
$\operatorname*{OR}\left(  f,\ldots,f\right)  $, the $\operatorname*{OR}$\ of
$k$ disjoint copies of a function $f$, is $\Omega\left(  k\right)  $ times the
query complexity of a single copy. \ This result can be proved by adapting
ideas from a direct product theorem for\ randomized query complexity given by
Drucker \cite{drucker:dpt}\ (we thank A. Drucker, personal communication).}
\ Of course, this is not an exponential separation.

In this paper, we gave a candidate for a problem that breaks the
\textquotedblleft$\sqrt{N}$ barrier\textquotedblright: namely, $k$-fold
\textsc{Forrelation}. \ Indeed, we conjecture that $k$-fold
\textsc{Forrelation}\ achieves the \textit{optimal} separation for all
$k=O\left(  1\right)  $, requiring $\widetilde{\Omega}\left(  N^{1-1/k}%
\right)  $\ classical randomized queries but only $\left\lceil k/2\right\rceil
$\ quantum queries.\footnote{And perhaps $k$-fold \textsc{Forrelation}%
\ continues to give optimal separations, all the way up to $k=O\left(  \log
N\right)  $.} \ Proving this conjecture is an enticing problem.
\ Unfortunately, $k$-fold \textsc{Forrelation}\ becomes extremely hard to
analyze when $k>2$, because we can no longer view the functions $f_{1}%
,\ldots,f_{k}$\ as confined to a low-dimensional subspace: now we have to view
them as confined to a low-dimensional \textit{manifold}, which is defined by
degree-$\left(  k-1\right)  $ polynomials. \ As such, we can no longer compute
posterior probabilities by simply appealing to the rotational invariance of
the Gaussian measure, which made our lives easier in the $k=2$\ case.
\ Instead we need to calculate integrals over a nonlinear manifold.

Short of proving our conjecture about $k$-fold \textsc{Forrelation}, it would
of course be nice to find \textit{any} partial Boolean function whose quantum
query complexity is $\operatorname*{polylog}N$, and whose randomized query
complexity is $N^{1/2+\Omega\left(  1\right)  }$.

Another problem we leave is to generalize our $O\left(  N^{1-1/k}\right)
$\ randomized estimation algorithm from block-multilinear polynomials to
\textit{arbitrary} bounded polynomials of degree $k$. \ As we said, Appendix
\ref{POLY}\ achieves this in the special case $k=2$. \ Achieving it for
arbitrary $k$ seems likely to require generalizing the machinery of Dinur et
al. \cite{dfko}.

A third problem concerns the notion of block-multilinear approximate degree,
$\widetilde{\operatorname*{bmdeg}}\left(  f\right)  $, that we introduced to
prove Theorem \ref{thm2.5}. \ Is there any asymptotic separation between
$\widetilde{\operatorname*{bmdeg}}\left(  f\right)  $\ and ordinary
approximate degree? \ What about a separation between $\widetilde{\operatorname*{bmdeg}%
}\left(  f\right)  $\ and quantum query complexity?

A fourth, more open-ended problem is whether there are any applications of
\textsc{Forrelation}, in the same sense that factoring and discrete log
provide \textquotedblleft applications\textquotedblright\ of Shor's
period-finding problem. \ Concretely, are there any situations where one has
two efficiently-computable Boolean functions $f,g:\left\{  0,1\right\}
^{n}\rightarrow\left\{  -1,1\right\}  $ (described, for example, by circuits),
one wants to estimate how forrelated they are, \textit{and} the structure of
$f$ and $g$ does not provide a fast classical way to do this?

Here are five other open problems:

\begin{enumerate}
\item[(1)] Can we tighten the lower bound on the randomized query complexity
of \textsc{Forrelation}\ from $\Omega(\frac{\sqrt{N}}{\log N})$\ to
$\Omega(\sqrt{N})$, or give an $O(\frac{\sqrt{N}}{\log N})$\ upper bound?

\item[(2)] Can we generalize our results from Boolean to non-Boolean functions?

\item[(3)] What are the largest possible quantum versus classical query
complexity separations for \textit{sampling} problems? \ Is an $O\left(
1\right)  $\ versus $\Omega(N)$\ separation possible in this case? \ Also,
what separations are possible for search or relation problems? \ (For our
results on these questions, see Appendix \ref{SAMPREL}.)

\item[(4)] While there exists a $1$-query quantum algorithm that solves
\textsc{Forrelation} with bounded error probability, the error probability we
are able to achieve is about $0.4$---more than the customary $1/3$. \ If we
want (say) a $1$ versus $N^{\Omega\left(  1\right)  }$\ quantum versus
classical query complexity separation, then how small can the quantum
algorithm's error be?

\item[(5)] While we show in Appendix \ref{PROP}\ that being \textquotedblleft%
\textit{un}forrelated\textquotedblright---that is, having $\Phi_{f,g}\leq
\frac{1}{100}$---behaves nicely as a property-testing problem, it would be
interesting to show the same for being forrelated.
\end{enumerate}

\section{Acknowledgments}

We are grateful to Ronald de Wolf for early discussions, Andy Drucker for
discussions about the randomized query complexity of $\operatorname*{OR}%
\left(  f,\ldots,f\right)  $, Shalev Ben-David and Sean Hallgren for
discussions about the glued-trees problem, and Saeed Mehraban for a numerical calculation.

\bibliographystyle{plain}
\bibliography{thesis}

\section{Preliminaries\label{PRELIM}}

We assume familiarity with basic concepts of quantum computing, as covered
(for example) in Nielsen and Chuang \cite{nc}. \ We also assume some
familiarity with the model of query or decision-tree complexity; see Buhrman
and de Wolf \cite{bw} for a good survey. \ In this section, we first give a
brief recap of query complexity (in Section \ref{QUERY}), then observe some
properties of the $k$-fold \textsc{Forrelation} problem (in Section
\ref{FORPRELIM}), and finally collect some lemmas about Gram-Schmidt
orthogonalization (in Section \ref{GS}) and Gaussian martingales (in Section
\ref{AZUMA}) that will be important for our randomized lower bound in Section
\ref{GAP}.

\subsection{Query Complexity\label{QUERY}}

Briefly, by the \textit{query complexity} of an algorithm $\mathcal{A}$, we
mean the number of queries that $\mathcal{A}$ makes to its input $z=\left(
z_{1},\ldots,z_{N}\right)  $, maximized over all valid inputs $z$.\footnote{If
we are talking about a partial Boolean function, then a \textquotedblleft
valid\textquotedblright\ input is simply any input that satisfies the
promise.} \ The query complexity of a function $F$\ is then the minimum query
complexity of \textit{any} algorithm $\mathcal{A}$\ (of a specified
type---classical, quantum, etc.) that outputs $F\left(  z\right)  $, with
bounded probability of error, given any valid input $z$.

One slightly unconventional choice that we make is to define \textquotedblleft
bounded probability of error\textquotedblright\ to mean \textquotedblleft
error probability at most $1/2-\varepsilon$, for some constant $\varepsilon
>0$\textquotedblright\ rather than \textquotedblleft error probability at most
$1/3$.\textquotedblright\ \ The reason is that we will be able to design a
$1$-query quantum algorithm that solves the \textsc{Forrelation}\ problem with
error probability $2/5$, but \textit{not} one that solves it with error
probability $1/3$. \ Of course, one can make the error probability as small as
one likes using amplification, but doing so increases the query complexity by
a constant factor.

We assume throughout this paper that the input $z\in\left\{  -1,1\right\}
^{N}$\ is Boolean, and we typically work in the $\left\{  -1,1\right\}
$\ basis for convenience. \ In the classical setting, each query returns a
single bit $z_{i}$, for some index $i\in\left[  N\right]  $ specified by
$\mathcal{A}$. \ In the quantum setting, each query performs a diagonal
unitary transformation%
\[
\left\vert i,w\right\rangle \rightarrow z_{i}\left\vert i,w\right\rangle ,
\]
where $w$ represents \textquotedblleft workspace qubits\textquotedblright%
\ that do not participate in the query.\footnote{For Boolean inputs $z$, this
is well-known to be exactly equivalent to a different definition of a quantum
query, wherein each basis state $\left\vert i,a,w\right\rangle $ gets mapped
to $\left\vert i,a\oplus z_{i},w\right\rangle $. \ Here $a$\ represents a
$1$-qubit \textquotedblleft answer register.\textquotedblright} \ Between two
queries, $\mathcal{A}$\ can apply any unitary transformation it likes that
does not depend on $z$.

In this paper, the input $z=\left(  z_{1},\ldots,z_{N}\right)  $ will
typically consist of the truth tables of one or more Boolean functions: for
example, $f,g:\left\{  0,1\right\}  ^{n}\rightarrow\left\{  -1,1\right\}  $,
or $f_{1},\ldots,f_{k}:\left\{  0,1\right\}  ^{n}\rightarrow\left\{
-1,1\right\}  $. \ Throughout, we use $n$ for the number of input bits that
these Boolean functions take (which roughly corresponds to the number of
\textit{qubits} in a quantum algorithm), and we use $N=2^{n}$\ for the number
of bits being queried in superposition. \ (Strictly speaking, we should set
$N=k\cdot2^{n}$, where $k$ is the number of Boolean functions. \ But this
constant-factor difference will not matter for us.) \ Thus, for the purposes
of query complexity, $N$ is the \textquotedblleft input
size,\textquotedblright\ in terms of which we state our upper and lower bounds.

\subsection{\label{FORPRELIM}\textsc{Forrelation}}

The \textsc{Forrelation}\ and $k$-fold \textsc{Forrelation}\ problems were
defined in Sections \ref{RESULT1}\ and \ref{RESULT3} respectively.
\ Informally, though, one could define $k$-fold \textsc{Forrelation}\ simply
as the problem of simulating the quantum circuit shown in Figure
\ref{kfold}---and in particular, of estimating the amplitude, call it
$\alpha_{0\cdots0}$, with which this circuit returns $\left\vert
0\right\rangle ^{\otimes n}$\ as its output.%
\begin{figure}[ptb]%
\centering
\includegraphics[
trim=1.424685in 5.571248in 1.145084in 0.143669in,
height=0.8882in,
width=2.8764in
]%
{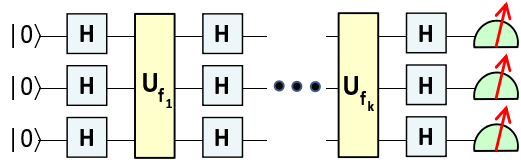}%
\caption{A quantum circuit that can be taken to define the $k$-fold
\textsc{Forrelation}\ problem. \ The circuit consists of $k$\ query
transformations $U_{f_{1}},\ldots,U_{f_{k}}$,\ which map each basis state
$\left\vert x\right\rangle $\ to $f_{i}\left(  x\right)  \left\vert
x\right\rangle $, sandwiched between rounds of Hadamard gates.}%
\label{kfold}%
\end{figure}
Observe that $\alpha_{0\cdots0}$\ is \textit{precisely} the quantity%
\[
\Phi_{f_{1},\ldots,f_{k}}:=\frac{1}{2^{\left(  k+1\right)  n/2}}\sum
_{x_{1},\ldots,x_{k}\in\left\{  0,1\right\}  ^{n}}f_{1}\left(  x_{1}\right)
\left(  -1\right)  ^{x_{1}\cdot x_{2}}f_{2}\left(  x_{2}\right)  \left(
-1\right)  ^{x_{2}\cdot x_{3}}\cdots\left(  -1\right)  ^{x_{k-1}\cdot x_{k}%
}f_{k}\left(  x_{k}\right)
\]
defined in Section \ref{RESULT3}. \ From this, it follows that we can decide
whether $\left\vert \Phi_{f_{1},\ldots,f_{k}}\right\vert \leq\frac{1}{100}%
$\ or $\Phi_{f_{1},\ldots,f_{k}}\geq\frac{3}{5}$\ with bounded probability of
error, and thereby solve the $k$-fold \textsc{Forrelation} problem, by making
only $k$ quantum queries to $f_{1},\ldots,f_{k}$.

Slightly more interesting is that we can improve the quantum query complexity
further, to $\left\lceil k/2\right\rceil $:

\begin{proposition}
\label{inpromisebqp}The $k$-fold \textsc{Forrelation} problem is solvable,
with error probability $0.4$, using $\left\lceil k/2\right\rceil $\ quantum
queries to the functions $f_{1},\ldots,f_{k}:\left\{  0,1\right\}
^{n}\rightarrow\left\{  -1,1\right\}  $, as well as $O\left(  nk\right)  $
quantum gates.
\end{proposition}

\begin{proof}
Let $\operatorname*{H}$\ be the Hadamard gate, and let $U_{f_{i}}$\ be the
query transformation that maps each computational basis state $\left\vert
x\right\rangle $\ to $f_{i}\left(  x\right)  \left\vert x\right\rangle $.
\ Then to improve from $k$ to $\left\lceil k/2\right\rceil $\ queries, we
modify the circuit of Figure \ref{kfold} in the following way.

In addition to the initial state $\left\vert 0\right\rangle ^{\otimes n}$, we
prepare a control qubit in the state $\left\vert +\right\rangle =\frac
{\left\vert 0\right\rangle +\left\vert 1\right\rangle }{\sqrt{2}}$. \ Then,
conditioned on the control qubit being $\left\vert 0\right\rangle $, we apply
the following sequence of operations to the initial state:%
\[
\operatorname*{H}\nolimits^{\otimes n}\rightarrow U_{f_{1}}\rightarrow
\operatorname*{H}\nolimits^{\otimes n}\rightarrow U_{f_{2}}\rightarrow
\cdots\rightarrow\operatorname*{H}\nolimits^{\otimes n}\rightarrow
U_{f_{\left\lceil k/2\right\rceil }}\rightarrow\operatorname*{H}%
\nolimits^{\otimes n}.
\]
Meanwhile, conditioned on the control qubit being $\left\vert 1\right\rangle
$, we apply the following sequence of operations:%
\[
\operatorname*{H}\nolimits^{\otimes n}\rightarrow U_{f_{k}}\rightarrow
\operatorname*{H}\nolimits^{\otimes n}\rightarrow U_{f_{k-1}}\rightarrow
\cdots\rightarrow\operatorname*{H}\nolimits^{\otimes n}\rightarrow
U_{f_{\left\lceil k/2\right\rceil +1}}.
\]
Finally, we measure the control qubit in the $\left\{  \left\vert
+\right\rangle ,\left\vert -\right\rangle \right\}  $\ basis, and
\textquotedblleft accept\textquotedblright\ (i.e., say that $\Phi
_{f_{1},\ldots,f_{k}}$\ is large) if and only if we find it in the state
$\left\vert +\right\rangle $.

It is not hard to see that the probability that this circuit accepts is
exactly%
\[
\frac{1+\Phi_{f_{1},\ldots,f_{k}}}{2}.
\]
Thus, consider an algorithm $\mathcal{A}$\ that rejects with probability
$1/4$, and runs the circuit with probability $3/4$. \ We have%
\[
\Pr\left[  \mathcal{A}\text{ accepts}\right]  =\frac{3}{4}\left(  \frac
{1+\Phi_{f_{1},\ldots,f_{k}}}{2}\right)  .
\]
If $\left\vert \Phi_{f_{1},\ldots,f_{k}}\right\vert \leq\frac{1}{100}%
$\ then\ the above is less than $0.4$, while if $\Phi_{f_{1},\ldots,f_{k}}%
\geq\frac{3}{5}$\ then it is at least $0.6$.
\end{proof}

Purely from the unitarity of the quantum algorithm to compute $\Phi
_{f_{1},\ldots,f_{k}}$, we can derive some interesting facts about
$\Phi_{f_{1},\ldots,f_{k}}$\ itself. \ Most obviously, we have $\left\vert
\Phi_{f_{1},\ldots,f_{k}}\right\vert \leq1$. \ But beyond that, let
$f_{k}^{\left(  x\right)  }\left(  x_{k}\right)  :=f_{k}\left(  x_{k}\right)
\left(  -1\right)  ^{x_{k}\cdot x}$. \ Then%
\begin{equation}
\sum_{x\in\left\{  0,1\right\}  ^{n}}\Phi_{f_{1},\ldots,f_{k-1},f_{k}^{\left(
x\right)  }}^{2}=1; \label{ss}%
\end{equation}
this is just saying that the sum of the squares of the final amplitudes in the
\textsc{Forrelation}\ algorithm must be $1$. \ Since there is nothing
\textquotedblleft special\textquotedblright\ about the outcome $\left\vert
0\cdots0\right\rangle $,\ it follows by symmetry that \
\[
\operatorname*{E}\left[  \Phi_{f_{1},\ldots,f_{k}}^{2}\right]  =\frac{1}{N}%
\]
if $f_{1},\ldots,f_{k}$\ are chosen are uniformly at random.

\subsection{Gram-Schmidt Orthogonalization\label{GS}}

Given an arbitrary collection of linearly-independent unit vectors $\left\vert
v_{1}\right\rangle ,\left\vert v_{2}\right\rangle \ldots$,\ the
\textit{Gram-Schmidt process} produces orthonormal vectors by recursively
projecting each $\left\vert v_{i}\right\rangle $\ onto the orthogonal
complement of the subspace spanned by $\left\vert v_{1}\right\rangle
,\ldots,\left\vert v_{i-1}\right\rangle $, and then normalizing the result.
\ That is:%

\begin{align*}
\left\vert z_{i}\right\rangle  &  =\left\vert v_{i}\right\rangle -\sum
_{j=1}^{i-1}\left\langle v_{i}|w_{j}\right\rangle \left\vert w_{j}%
\right\rangle ,\\
\left\vert w_{i}\right\rangle  &  =\beta_{i}\left\vert z_{i}\right\rangle
\end{align*}
where $\beta_{i}=\frac{1}{\sqrt{\left\langle z_{i}|z_{i}\right\rangle }}$\ is
a normalizing constant. \ Note that $\left\langle z_{i}|z_{i}\right\rangle
\leq1$\ (since $\left\vert z_{i}\right\rangle $\ is the projection of a unit
vector onto a subspace), and hence $\beta_{i}\geq1$.

We will be interested in the behavior of this process when the $\left\vert
v_{i}\right\rangle $'s are already very close to orthogonal. \ We can control
that behavior with the help of the following lemma.

\begin{lemma}
[Gram-Schmidt Lemma]\label{gslem}Let $\left\vert v_{1}\right\rangle
,\ldots,\left\vert v_{t}\right\rangle $\ be unit vectors\ with $\left\vert
\left\langle v_{i}|v_{j}\right\rangle \right\vert \leq\varepsilon$\ for all
$i\neq j$, and suppose $t\leq0.1/\varepsilon$ (so in particular,
$\varepsilon\leq0.1$). \ Let $\left\vert w_{i}\right\rangle $\ and $\beta_{i}%
$\ be as above. \ Then for all $i>j$, we have%
\begin{align*}
\left\vert \left\langle v_{i}|w_{j}\right\rangle \right\vert  &
\leq\varepsilon+2j\varepsilon^{2},\\
\beta_{i}  &  \leq1+2i\varepsilon^{2}.
\end{align*}
So in particular, under the stated hypothesis, $\left\vert \left\langle
v_{i}|w_{j}\right\rangle \right\vert \leq1.2\varepsilon$\ and $\beta_{i}%
\leq1+0.2\varepsilon$.
\end{lemma}

\begin{proof}
We will do an induction on ordered pairs $\left(  i,j\right)  $, in the order
$\left(  2,1\right)  ,\left(  3,1\right)  ,\left(  3,2\right)  ,\left(
4,1\right)  ,\ldots$, with two induction hypotheses. \ Here are the
hypotheses: for all $i>j$,%
\begin{align*}
\left\vert \left\langle v_{i}|w_{j}\right\rangle \right\vert  &
\leq\varepsilon+Aj\varepsilon^{2},\\
1-\left\langle z_{i}|z_{i}\right\rangle  &  \leq Bi\varepsilon^{2}.
\end{align*}
for some constants $A,B$\ to be determined later.

For the base case ($i=2$ and $j=1$), we have $\left\vert \left\langle
v_{2}|w_{1}\right\rangle \right\vert =\left\vert \left\langle v_{2}%
|v_{1}\right\rangle \right\vert \leq\varepsilon$\ and%
\begin{align*}
\left\langle z_{2}|z_{2}\right\rangle  &  =\left(  \left\langle v_{2}%
\right\vert -\left\langle v_{1}|v_{2}\right\rangle \left\langle v_{1}%
\right\vert \right)  \left(  \left\vert v_{2}\right\rangle -\left\langle
v_{1}|v_{2}\right\rangle \left\vert v_{1}\right\rangle \right) \\
&  =1-\left\langle v_{1}|v_{2}\right\rangle ^{2}\\
&  \geq1-\varepsilon^{2}.
\end{align*}
For the induction step: first,%
\begin{align*}
\left\langle v_{i}|w_{j}\right\rangle  &  =\left\langle v_{i}\right\vert
\beta_{j}\left(  \left\vert v_{j}\right\rangle -\sum_{k=1}^{j-1}\left\langle
v_{j}|w_{k}\right\rangle \left\vert w_{k}\right\rangle \right) \\
&  =\frac{1}{\sqrt{\left\langle z_{j}|z_{j}\right\rangle }}\left(
\left\langle v_{i}|v_{j}\right\rangle +\sum_{k=1}^{j-1}\left\langle
v_{i}|w_{k}\right\rangle \left\langle v_{j}|w_{k}\right\rangle \right)  .
\end{align*}
So%
\begin{align*}
\left\vert \left\langle v_{i}|w_{j}\right\rangle \right\vert  &  \leq\frac
{1}{\sqrt{1-Bj\varepsilon^{2}}}\left(  \varepsilon+\sum_{k=1}^{j-1}\left(
\varepsilon+Ak\varepsilon^{2}\right)  ^{2}\right) \\
&  \leq\frac{1}{1-Bj\varepsilon^{2}}\left(  \varepsilon+j\varepsilon
^{2}+Aj^{2}\varepsilon^{3}+\frac{A^{2}j^{3}\varepsilon^{4}}{3}\right) \\
&  \leq\left(  1+\frac{B}{1-0.01B}j\varepsilon^{2}\right)  \left(
\varepsilon+\left(  1+0.1A+\frac{0.01A^{2}}{3}\right)  j\varepsilon^{2}\right)
\\
&  \leq\varepsilon+\left[  \left(  1+0.1A+\frac{0.01A^{2}}{3}\right)  \left(
1+\frac{0.01B}{1-0.01B}\right)  +\frac{0.1B}{1-0.01B}\right]  j\varepsilon
^{2},
\end{align*}
where we repeatedly made the substitutions $\varepsilon\leq0.1$ and
$j\varepsilon\leq0.1$\ to produce multiples of $j\varepsilon^{2}$ in the
numerator, and get rid of $j$ and $\varepsilon$ in the denominator.\ \ Second,%
\begin{align*}
\left\langle z_{i}|z_{i}\right\rangle  &  =\left(  \left\langle v_{i}%
\right\vert -\sum_{j=1}^{i-1}\left\langle v_{i}|w_{j}\right\rangle
\left\langle w_{j}\right\vert \right)  \left(  \left\vert v_{i}\right\rangle
-\sum_{j=1}^{i-1}\left\langle v_{i}|w_{j}\right\rangle \left\vert
w_{j}\right\rangle \right) \\
&  =\left\langle v_{i}|v_{i}\right\rangle -\sum_{j=1}^{i-1}\left\langle
v_{i}|w_{j}\right\rangle ^{2}+\sum_{j\neq k\in\left[  i-1\right]
}\left\langle v_{i}|w_{j}\right\rangle \left\langle v_{i}|w_{k}\right\rangle
\left\langle w_{j}|w_{k}\right\rangle \\
&  =1-\sum_{j=1}^{i-1}\left\langle v_{i}|w_{j}\right\rangle ^{2}%
\end{align*}
where we used the fact that $\left\langle w_{j}|w_{k}\right\rangle =0$. \ So%
\begin{align*}
1-\left\langle z_{i}|z_{i}\right\rangle  &  \leq\sum_{j=1}^{i-1}\left(
\varepsilon+Aj\varepsilon^{2}\right)  ^{2}\\
&  \leq i\varepsilon^{2}+Ai^{2}\varepsilon^{3}+\frac{A^{2}i^{3}\varepsilon
^{4}}{3}\\
&  \leq\left(  1+0.1A+\frac{0.01A^{2}}{3}\right)  i\varepsilon^{2}.
\end{align*}
If we now make the choice (say) $A=2$\ and $B=1.5$, we find that both parts of
the induction are satisfied:%
\begin{align*}
\left\vert \left\langle v_{i}|w_{j}\right\rangle \right\vert  &
\leq\varepsilon+1.39j\varepsilon^{2}\leq\varepsilon+Aj\varepsilon^{2},\\
1-\left\langle z_{i}|z_{i}\right\rangle  &  \leq1.22i\varepsilon^{2}\leq
Bi\varepsilon^{2}.
\end{align*}
Furthermore, we now have the lemma, since%
\[
\beta_{i}=\frac{1}{\sqrt{\left\langle z_{i}|z_{i}\right\rangle }}\leq\frac
{1}{\left\langle z_{i}|z_{i}\right\rangle }=1+\frac{1-\left\langle z_{i}%
|z_{i}\right\rangle }{\left\langle z_{i}|z_{i}\right\rangle }\leq
1+\frac{Bi\varepsilon^{2}}{1-Bi\varepsilon^{2}}\leq1+\frac{Bi\varepsilon^{2}%
}{1-0.01B}\leq1+2i\varepsilon^{2}.
\]

\end{proof}

\subsection{Gaussian Azuma's Inequality\label{AZUMA}}

Azuma's inequality is a well-known generalization of the Chernoff/Hoeffing
tail bound to the case of martingales with bounded differences. \ We will need
a generalization of Azuma's inequality to martingale difference sequences in
which each term is Gaussian (and therefore, unbounded). \ Fortunately, Shamir
\cite[Theorem 2]{ohadshamir} recently proved a useful such generalization.
\ We now state Shamir's bound, in a slightly different form than in
\cite{ohadshamir} (but easily seen to be equivalent).

\begin{lemma}
[Gaussian Azuma's Inequality \cite{ohadshamir}]\label{azumagauss}Suppose
$x_{1},x_{2},\ldots$\ form a martingale difference sequence, in the sense that
$\operatorname{E}\left[  x_{i}|x_{1},\ldots,x_{i-1}\right]  =0$. \ Suppose
further that, conditioned on its predecessors, $x_{i}$\ is always
\textquotedblleft dominated by an $\mathcal{N}\left(  0,\sigma^{2}\right)
$\ Gaussian,\textquotedblright\ in the sense that $\Pr\left[  \left\vert
x_{i}\right\vert >\sigma B\right]  <\exp\left(  -B^{2}/2\right)  $ for all
$B$. \ Then%
\[
\Pr\left[  \left\vert x_{1}+\cdots+x_{t}\right\vert >c\sigma\sqrt{t}\right]
<2\exp\left(  -\frac{c^{2}}{56}\right)  .
\]

\end{lemma}

Note, in particular, that if the $x_{i}$'s themselves are $\mathcal{N}\left(
0,\tau_{i}^{2}\right)  $\ Gaussians\ for some (possibly-differing) variances
$\tau_{i}<\sigma$, then the $x_{i}$'s are dominated by $\mathcal{N}\left(
0,\sigma^{2}\right)  $\ Gaussians, so Lemma \ref{azumagauss} can be applied.

\section{Maximal Quantum/Classical Query Complexity Gap\label{GAP}}

In this section, we prove that the randomized query complexity of
\textsc{Forrelation}\ is $\Omega(\frac{\sqrt{N}}{\log N})$. \ Previously,
Aaronson \cite{aar:ph} proved an $\Omega(N^{1/4})$\ randomized lower bound for
this problem. \ We will need a further idea to improve the lower bound to
$\widetilde{\Omega}(N^{1/3})$, a still further idea to improve it to
$\widetilde{\Omega}(N^{2/5})$, and then yet another idea to get all the way up
to $\widetilde{\Omega}(\sqrt{N})$.

Following \cite{aar:ph}, the first step is to replace \textsc{Forrelation}\ by
a \textquotedblleft continuous relaxation\textquotedblright\ of the problem: a
version that is strictly easier (and thus, for which proving a lower bound is
\textit{harder}), but which has rotational symmetry that will be extremely
convenient for us. \ Thus, in \textsc{Real Forrelation}, we are given oracle
access to two real functions $f,g:\left\{  0,1\right\}  ^{n}\rightarrow
\mathbb{R}$. \ As usual, the \textquotedblleft input size\textquotedblright%
\ is $N=2^{n}$. \ We are promised that the pair $\left\langle f,g\right\rangle
$\ was drawn from one of two probability measures:

\begin{enumerate}
\item[(i)] In the uniform measure $\mathcal{U}$, every $f\left(  x\right)
$\ and $g\left(  y\right)  $\ value is an independent $\mathcal{N}\left(
0,1\right)  $\ Gaussian.

\item[(ii)] In the forrelated measure $\mathcal{F}$, every $f\left(  x\right)
$\ value is an independent $\mathcal{N}\left(  0,1\right)  $\ Gaussian, while
every $g\left(  y\right)  $\ value is fixed to%
\[
\hat{f}\left(  y\right)  =\frac{1}{\sqrt{N}}\sum_{x\in\left\{  0,1\right\}
^{n}}\left(  -1\right)  ^{x\cdot y}f\left(  x\right)  .
\]

\end{enumerate}

The problem is to decide, with constant bias, whether (i) or (ii) holds (i.e.,
whether $\left\langle f,g\right\rangle $\ was drawn from $\mathcal{U}$\ or
from $\mathcal{F}$).

We will often treat (the truth tables of) $f$ and $g$ as vectors in
$\mathbb{R}^{N}$. \ Then another way to think about \textsc{Real Forrelation}
is this: in case (i), $f$\ and $g$ are drawn independently from $\mathcal{N}%
\left(  0,1\right)  ^{N}$. \ In case (ii), $f$ and $g$ are \textit{also} both
distributed according to $\mathcal{N}\left(  0,1\right)  ^{N}$, by the
rotational symmetry of the Gaussian measure and the unitarity of the Hadamard
transform. \ But they are no longer independent: they are related by $g=Hf$,
where $H$\ is the $N\times N$\ Hadamard matrix, given by $H_{x,y}=\left(
-1\right)  ^{x\cdot y}/\sqrt{N}$. \ The problem is to detect whether this
correlation is present.

An algorithm for \textsc{Real Forrelation}\ proceeds by querying $f\left(
x\right)  $\ and $g\left(  y\right)  $\ values one at a time, deciding which
$x$\ or $y$\ to query next based on the values seen so far. \ We are
interested in the expected number of queries needed by the best algorithm.

\subsection{Discrete Versus Continuous\label{CONT}}

As a first step, we need to show that a lower bound for \textsc{Real
Forrelation}\ really does imply the same lower bound for the original, Boolean
\textsc{Forrelation}\ problem. \ The key to doing so is the following result,
which calculates the expected value of $\Phi_{F,G}$, for Boolean
\textsc{Forrelation}\ instances $\left\langle F,G\right\rangle $\ that are
produced by \textquotedblleft rounding\textquotedblright\ real instances
$\left\langle f,g\right\rangle $\ in a natural way.

\begin{theorem}
\label{2overpi}Suppose $\left\langle f,g\right\rangle $ are drawn from the
forrelated measure $\mathcal{F}$. \ Define Boolean functions $F,G:\left\{
0,1\right\}  ^{n}\rightarrow\left\{  -1,1\right\}  $\ by $F\left(  x\right)
:=\operatorname*{sgn}\left(  f\left(  x\right)  \right)  $ and $G\left(
y\right)  :=\operatorname*{sgn}\left(  g\left(  y\right)  \right)  $. \ Then%
\[
\operatorname{E}_{f,g\sim\mathcal{F}}\left[  \Phi_{F,G}\right]  =\frac{2}{\pi
}\pm O\left(  \frac{\log N}{N}\right)  .
\]

\end{theorem}

\begin{proof}
By linearity of expectation, it suffices to calculate $\operatorname{E}\left[
F\left(  x\right)  \left(  -1\right)  ^{x\cdot y}G\left(  y\right)  \right]
$\ for some \textit{specific} $x,y$ pair. \ Let $v\in\mathbb{R}^{N}$ be a
vector of independent $\mathcal{N}\left(  0,1\right)  $\ Gaussians, and let
$H$\ be the $N\times N$\ Hadamard matrix (without normalization), with entries
$H_{x,y}=\left(  -1\right)  ^{x\cdot y}$. \ Then we can consider $\left\langle
F,G\right\rangle $\ to have been generated as follows:%
\begin{align*}
F\left(  x\right)   &  =\operatorname*{sgn}\left(  v_{x}\right)  ,\\
G\left(  y\right)   &  =\operatorname*{sgn}(\left(  Hv\right)  _{y}).
\end{align*}
Now, $\left(  Hv\right)  _{y}$\ can be expressed as the sum of $H_{x,y}v_{x}%
$\ with the independent Gaussian random variable%
\[
Z:=\sum_{x^{\prime}\neq x}H_{x^{\prime},y}v_{x^{\prime}}.
\]
Let $G^{\prime}\left(  y\right)  :=\operatorname*{sgn}\left(  Z\right)  $.
\ Then%
\[
\operatorname{E}\left[  F\left(  x\right)  \left(  -1\right)  ^{x\cdot
y}G^{\prime}\left(  y\right)  \right]  =\operatorname{E}\left[
\operatorname*{sgn}\left(  v_{x}\right)  \left(  -1\right)  ^{x\cdot
y}\operatorname*{sgn}\left(  Z\right)  \right]  =0,
\]
since $v_{x}$\ and $Z$ are independent Gaussians both with mean $0$. \ Note
that adding $H_{x,y}v_{x}$\ back to $Z$\ can only flip $Z$ to having the
\textit{same} sign as $\operatorname*{sgn}\left(  v_{x}\right)  \left(
-1\right)  ^{x\cdot y}$, not the opposite sign---and hence can only increase
$F\left(  x\right)  \left(  -1\right)  ^{x\cdot y}G\left(  y\right)  $. \ It
follows that%
\[
\operatorname{E}\left[  F\left(  x\right)  \left(  -1\right)  ^{x\cdot
y}G\left(  y\right)  \right]  =2\Pr\left[  G\left(  y\right)  \neq G^{\prime
}\left(  y\right)  \right]  .
\]

The event $G\left(  y\right)  \neq G^{\prime}\left(  y\right)  $\ occurs if
and only if the following two events both occur:%
\begin{align*}
\left\vert H_{x,y}v_{x}\right\vert  &  >\left\vert Z\right\vert ,\\
\operatorname*{sgn}\left(  H_{x,y}v_{x}\right)   &  \neq\operatorname*{sgn}%
\left(  Z\right)  .
\end{align*}
Since $H_{x,y}\in\left\{  -1,1\right\}  $\ and\ the distribution of $v_{x}%
$\ is symmetric about $0$, we can assume without loss of generality that
$H_{x,y}=1$.

Let $Z\left(  t\right)  $\ be the probability density function of $Z$. \ Then%
\begin{align*}
\Pr\left[  \left\vert H_{x,y}v_{x}\right\vert >\left\vert Z\right\vert \text{
\ and \ }\operatorname*{sgn}\left(  H_{x,y}v_{x}\right)  \neq
\operatorname*{sgn}\left(  Z\right)  \right]   &  =2\int_{t=0}^{\infty
}Z\left(  t\right)  \Pr\left[  H_{x,y}v_{x}>t\right]  dt\\
&  =2\int_{t=0}^{\infty}Z\left(  t\right)  \Pr\left[  v_{x}>t\right]  dt.
\end{align*}
(Here the factor of $2$ appears because we are restricting to the case $Z>0$,
and there is an equal probability coming from the $Z<0$\ case.)

As a linear combination of $N-1$\ independent $\mathcal{N}\left(  0,1\right)
$\ Gaussians, with $\pm1$\ coefficients, $Z$\ has the $\mathcal{N}\left(
0,N-1\right)  $\ Gaussian distribution. \ Therefore%
\begin{align*}
2\int_{t=0}^{\infty}Z\left(  t\right)  \Pr\left[  v_{x}>t\right]  dt  &
=\frac{2}{\sqrt{2\pi\left(  N-1\right)  }}\int_{t=0}^{\infty}\exp\left(
-\frac{t^{2}}{2\left(  N-1\right)  }\right)  \Pr\left[  v_{x}>t\right]  dt\\
&  \leq\frac{2}{\sqrt{2\pi\left(  N-1\right)  }}\int_{t=0}^{\infty}\Pr\left[
v_{x}>t\right]  dt\\
&  =\frac{2}{\sqrt{2\pi\left(  N-1\right)  }}\operatorname{E}\left[
\left\vert v_{x}\right\vert \right] \\
&  =\frac{2}{\pi\sqrt{N-1}}\\
&  \leq\frac{2}{\pi\sqrt{N}}+O\left(  \frac{1}{N^{3/2}}\right)  .
\end{align*}
Here the fourth line follows from $\operatorname{E}\left[  \left\vert
X\right\vert \right]  =\sqrt{2/\pi}$\ when $X$\ is an $\mathcal{N}\left(
0,1\right)  $\ Gaussian. \ In the other direction, for all $C>0$ we have%
\begin{align*}
2\int_{t=0}^{\infty}Z\left(  t\right)  \Pr\left[  v_{x}>t\right]  dt  &
=\frac{2}{\sqrt{2\pi\left(  N-1\right)  }}\int_{t=0}^{\infty}\exp\left(
-\frac{t^{2}}{2\left(  N-1\right)  }\right)  \Pr\left[  v_{x}>t\right]  dt\\
&  \geq\frac{2}{\sqrt{2\pi N}}\int_{t=0}^{C}\exp\left(  -\frac{t^{2}}{2\left(
N-1\right)  }\right)  \Pr\left[  v_{x}>t\right]  dt\\
&  \geq\frac{2}{\sqrt{2\pi N}}\exp\left(  -\frac{C^{2}}{2\left(  N-1\right)
}\right)  \int_{t=0}^{C}\Pr\left[  v_{x}>t\right]  dt\\
&  =\frac{2}{\sqrt{2\pi N}}\exp\left(  -\frac{C^{2}}{2\left(  N-1\right)
}\right)  \left(  \operatorname{E}\left[  \left\vert v_{x}\right\vert \right]
-\frac{1}{\sqrt{2\pi}}\int_{t=C}^{\infty}te^{-t^{2}/2}dt\right) \\
&  =\frac{2}{\sqrt{2\pi N}}\exp\left(  -\frac{C^{2}}{2\left(  N-1\right)
}\right)  \left(  \sqrt{\frac{2}{\pi}}-\frac{e^{-C^{2}/2}}{\sqrt{2\pi}%
}\right)  .
\end{align*}
If we set $C:=\sqrt{\log N}$, then the above is%
\[
\frac{2}{\sqrt{2\pi N}}\left(  1-O\left(  \frac{\log N}{N}\right)  \right)
\left(  \sqrt{\frac{2}{\pi}}-\frac{1}{\sqrt{2\pi}N}\right)  \geq\frac{2}%
{\pi\sqrt{N}}-O\left(  \frac{\log N}{N^{3/2}}\right)  .
\]

Therefore%
\begin{align*}
\operatorname{E}\left[  \Phi_{F,G}\right]   &  =\frac{1}{2^{3n/2}}\sum
_{x,y\in\left\{  0,1\right\}  ^{n}}\operatorname{E}\left[  F\left(  x\right)
\left(  -1\right)  ^{x\cdot y}G\left(  y\right)  \right] \\
&  =\frac{N^{2}}{N^{3/2}}\cdot\left(  \frac{2}{\pi\sqrt{N}}\pm O\left(
\frac{\log N}{N^{3/2}}\right)  \right) \\
&  =\frac{2}{\pi}\pm O\left(  \frac{\log N}{N}\right)  .
\end{align*}

\end{proof}

Earlier, Aaronson \cite[Theorem 9]{aar:ph}\ proved a variant of Theorem
\ref{2overpi}, but with a badly suboptimal constant: he was only able to show
that%
\[
\operatorname{E}\left[  \Phi_{F,G}\right]  \geq\cos\left(  2\arccos\sqrt
{\frac{2}{\pi}}\right)  -o\left(  1\right)  \approx0.273,
\]
compared to the exact value of $2/\pi\approx0.637$\ that we get here. \ As a
result, if we used \cite{aar:ph}, we would only be able to show hardness for
distinguishing $\Phi_{F,G}\approx0$\ from (say) $\Phi_{F,G}\geq\frac{1}{4}$,
rather than $\Phi_{F,G}\approx0$\ from $\Phi_{F,G}\geq\frac{3}{5}$.

We now use Theorem \ref{2overpi}\ to give the desired reduction from
\textsc{Real Forrelation}\ to \textsc{Forrelation}.

\begin{corollary}
\label{realfor}Suppose there exists a $T$-query algorithm that solves
\textsc{Forrelation}\ with bounded error. \ Then there also exists an
$O\left(  T\right)  $-query algorithm that solves \textsc{Real Forrelation}
with bounded error.
\end{corollary}

\begin{proof}
Let $\left\langle f,g\right\rangle $ be an instance of \textsc{Real
Forrelation}. \ Then we will produce an instance $\left\langle
F,G\right\rangle $\ of Boolean \textsc{Forrelation}\ exactly as in Theorem
\ref{2overpi}: that is, we set $F\left(  x\right)  :=\operatorname*{sgn}%
\left(  f\left(  x\right)  \right)  $ for all $x$ and $G\left(  y\right)
:=\operatorname*{sgn}\left(  g\left(  y\right)  \right)  $ for all $y$. \ If
$\left\langle f,g\right\rangle $\ was drawn from the uniform measure
$\mathcal{U}$, then $\operatorname{E}\left[  \Phi_{F,G}^{2}\right]  =\frac
{1}{N}$\ by symmetry. \ So by Markov's inequality,%
\[
\Pr\left[  \left\vert \Phi_{F,G}\right\vert >\frac{1}{100}\right]
<\frac{10000}{N}.
\]
By contrast, if $\left\langle f,g\right\rangle $\ was drawn from the
forrelated measure $\mathcal{F}$, then%
\[
\operatorname{E}\left[  \Phi_{F,G}\right]  \geq\frac{2}{\pi}-o\left(
1\right)
\]
by Theorem \ref{2overpi}. \ By Markov's inequality (and the fact that
$\Phi_{F,G}\leq1$), it follows that for all constants $\varepsilon\in\left(
0,1/2\right)  $,%
\[
\Pr\left[  \Phi_{F,G}\geq\frac{2}{\pi}-\varepsilon\right]  >\varepsilon.
\]
So in particular,%
\[
\Pr\left[  \Phi_{F,G}\geq\frac{3}{5}\right]  >\frac{1}{30}.
\]
Using a constant amount of amplification, we can clearly produce an $O\left(
T\right)  $-query\ algorithm for \textsc{Forrelation}\ that errs with
probability at most (say) $\frac{1}{100}$\ on all $\left\langle
F,G\right\rangle $. \ By the union bound, such an algorithm distinguishes the
case that $\left\langle f,g\right\rangle $\ was drawn from $\mathcal{U}$\ from
the case that $\left\langle f,g\right\rangle $\ was drawn from $\mathcal{F}$
with bias at least%
\[
\frac{1}{30}-\frac{10000}{N}-2\left(  \frac{1}{100}\right)  =\Omega\left(
1\right)  .
\]

\end{proof}

Because of Corollary \ref{realfor}, we see that, to prove a lower bound for
\textsc{Forrelation}, it suffices to prove the same lower bound for
\textsc{Real Forrelation}. \ Furthermore, because the \textsc{Real
Forrelation}\ problem is to distinguish two probability distributions, we can
assume without loss of generality that any algorithm for the latter is deterministic.

\subsection{\label{LBREAL}Lower Bound for \textsc{Real Forrelation}}

We now proceed to a lower bound on the query complexity of \textsc{Real
Forrelation}. \ As a first step, let us recast our problem more abstractly.
\ For convenience, we will use ket notation ($\left\vert v\right\rangle $,
$\left\vert w\right\rangle $, etc.)\ for vectors in $\mathbb{R}^{N}$, even if
the vectors do not represent quantum states and are not even normalized. \ Let
$\left\vert 1\right\rangle ,\ldots,\left\vert N\right\rangle $\ be an
orthonormal basis for $\mathbb{R}^{N}$, and let $\left\vert \hat{\imath
}\right\rangle =H\left\vert i\right\rangle $\ be the Hadamard transform of
$\left\vert i\right\rangle $\ (so that $|\hat{1}\rangle,\ldots,|\hat{N}%
\rangle$\ is also an orthonormal basis).

Then consider the following generalization of \textsc{Real Forrelation}, which
we call \textsc{Gaussian Distinguishing}. \ We are given a finite set
$\mathcal{V}$\ of unit vectors in $\mathbb{R}^{N}$, called \textquotedblleft
test vectors.\textquotedblright\ \ (In our case, $\mathcal{V}$\ happens to
equal $\left\{  \left\vert 1\right\rangle ,\ldots,\left\vert N\right\rangle
,|\hat{1}\rangle,\ldots,|\hat{N}\rangle\right\}  $.) \ In each step, we are
allowed to pick any test vector $\left\vert v\right\rangle \in\mathcal{V}$
that hasn't been picked in previous steps. \ We then \textquotedblleft
query\textquotedblright\ $\left\vert v\right\rangle $, getting back a
real-valued response $a_{v}\in\mathbb{R}$. \ The problem is to distinguish the
following two cases, with constant bias:

\begin{enumerate}
\item[(i)] Each $a_{v}$\ is drawn independently from $\mathcal{N}\left(
0,1\right)  $.

\item[(ii)] Each $a_{v}$\ equals $\left\langle \Psi|v\right\rangle $, where
$\left\vert \Psi\right\rangle \in\mathbb{R}^{N}$\ is a vector drawn from
$\mathcal{N}\left(  0,1\right)  ^{N}$\ that is fixed throughout the algorithm.
\end{enumerate}

We will actually prove a \textit{general} lower bound for \textsc{Gaussian
Distinguishing}, which works whenever $\left\vert \mathcal{V}\right\vert $\ is
not too large, and every pair of vectors in $\mathcal{V}$\ is sufficiently
close to orthogonal. \ Here is our general result:

\begin{theorem}
\label{gdlb}Suppose $\left\vert \mathcal{V}\right\vert \leq M$, and
$\left\vert \left\langle v|w\right\rangle \right\vert \leq\varepsilon$ for all
distinct vectors $\left\vert v\right\rangle ,\left\vert w\right\rangle
\in\mathcal{V}$. \ Then any classical algorithm for \textsc{Gaussian
Distinguishing} must make $\Omega\left(  \frac{1/\varepsilon}{\log
(M/\varepsilon)}\right)  $ queries.
\end{theorem}

In our case (\textsc{Real Forrelation}), we have $M=2N$ and $\varepsilon
=1/\sqrt{N}$, so the lower bound we get is $\Omega\left(  \frac{\sqrt{N}}{\log
N}\right)  $. \ As a remark, the example of \textsc{Real Forrelation}\ shows
that Theorem \ref{gdlb} is tight in its dependence on $1/\varepsilon$. \ One
can also construct an example to show that the theorem's dependence on $M$ is
in some sense needed (if possibly not tight). \ That is, one does \textit{not}
have a $\widetilde{\Omega}\left(  1/\varepsilon\right)  $\ lower bound on
query complexity for arbitrarily large $M$, but at best a $\Omega(\left(
1/\varepsilon\right)  ^{2/3})$\ lower bound.\footnote{\label{ipexample}Here is
the example that shows this: let $\left\vert 1\right\rangle ,\ldots,\left\vert
n\right\rangle $\ be orthogonal unit vectors. \ Then for all $2^{n}$\ strings
$z=z_{1}\cdots z_{n}\in\left\{  -1,1\right\}  ^{n}$, let $\left\vert
w_{z}\right\rangle $\ be a vector such that $\left\langle w_{z}|i\right\rangle
=z_{i}/n^{3/2}$ for all $i\in\left[  n\right]  $, and also such that the
projections of the $\left\vert w_{z}\right\rangle $'s onto the orthogonal
complement of $\left\vert 1\right\rangle ,\ldots,\left\vert n\right\rangle $
are all orthogonal to one another. \ Let $\mathcal{V}=\left\{  \left\vert
1\right\rangle ,\ldots,\left\vert n\right\rangle \right\}  \cup\left\{
\left\vert w_{z}\right\rangle \right\}  _{z\in\left\{  -1,1\right\}  ^{n}}$.
\ Then the inner product between any two distinct vectors in $\mathcal{V}$\ is
upper-bounded by $\varepsilon=1/n^{3/2}$ in absolute value (the inner product
between any two $\left\vert w_{z}\right\rangle $'s is at most $n/(n^{3/2}%
)^{2}=1/n^{2}$). \ On the other hand, here is an algorithm that solves
\textsc{Gaussian Distinguishing}\ using only $O\left(  n\right)
\ll1/\varepsilon$ queries: first query $\left\vert 1\right\rangle
,\ldots,\left\vert n\right\rangle $ to obtain $a_{1},\ldots,a_{n}$. \ Let
$\left\vert \varphi\right\rangle :=a_{1}\left\vert 1\right\rangle
+\cdots+a_{n}\left\vert n\right\rangle $. \ Next, find $n$ distinct vectors
$\left\vert w_{z}\right\rangle $\ that each have inner product $\Theta\left(
n/n^{3/2}\right)  =\Theta\left(  1/\sqrt{n}\right)  $\ with $\left\vert
\varphi\right\rangle $\ (such vectors\ can always be found, so long as
$\left\vert a_{i}\right\vert =\Omega\left(  1\right)  $\ for $\Omega\left(
n\right)  $\ values of $i$), and query all of them, letting $b_{1}%
,\ldots,b_{n}$ be the results. \ In case (i), we have $\operatorname*{E}%
\left[  b_{1}+\cdots+b_{n}\right]  $\ and $\operatorname*{Var}\left[
b_{1}+\cdots+b_{n}\right]  =n$. \ But in case (ii), we have $\operatorname*{E}%
\left[  b_{1}+\cdots+b_{n}\right]  =\Theta\left(  \sqrt{n}\right)  $\ and
$\operatorname*{Var}\left[  b_{1}+\cdots+b_{n}\right]  =O\left(  n\right)  $,
allowing the two cases to be distinguished with constant bias.} \ In the
context of \textsc{Real Forrelation}, this means that, if the only thing we
knew about $\mathcal{V}$\ was that $\left\vert \left\langle v|w\right\rangle
\right\vert \leq1/\sqrt{N}$\ for all distinct $\left\vert v\right\rangle
,\left\vert w\right\rangle \in\mathcal{V}$ (so in particular, we had no upper
bound on $\mathcal{V}$'s cardinality), then we could not hope to prove any
lower bound better than $\Omega(N^{1/3})$.\footnote{In fact one \textit{can}
prove a $\widetilde{\Omega}(N^{1/3})$\ lower bound even under this
restriction---and more generally, in the statement of Theorem \ref{gdlb}, one
can replace the lower bound $\Omega\left(  \frac{1/\varepsilon}{\log
(M/\varepsilon)}\right)  $\ by $\Omega\left(  \frac{\left(  1/\varepsilon
\right)  ^{2/3}}{\left(  \log1/\varepsilon\right)  ^{1/3}}\right)  $,
independent of $M$. \ We will briefly remark on how to do this at the relevant
point in our proof.}

For the remainder of the proof, we will fix $\varepsilon=1/\sqrt{N}$ for
concreteness; but will leave $M$ unfixed. \ Note that $N$\ will \textit{only}
enter into the proof through its relation with $\varepsilon$; the fact that
$N$\ is also the dimensionality of the vectors will be irrelevant for
us.\footnote{By slightly modifying the example from footnote \ref{ipexample}%
---to make the projections of the $\left\vert w_{z}\right\rangle $'s onto the
orthogonal complement of $\left\vert 1\right\rangle ,\ldots,\left\vert
n\right\rangle $ not exactly orthogonal to each other, but merely
approximately orthogonal---one can produce an instance of\ \textsc{Gaussian
Distinguishing} whose classical query complexity is only $O(\left(
1/\varepsilon\right)  ^{2/3})$, and which \textit{also} satisfies $N=O\left(
n^{4}\right)  =O(\left(  1/\varepsilon\right)  ^{8/3})$. \ This is an
exponential improvement in the dimensionality $N$\ compared to footnote
\ref{ipexample}. \ It would be interesting to know whether enforcing, say,
$N=O((1/\varepsilon)^{2})$\ rules out such examples.}

The first question we need to answer is this: suppose an algorithm has queried
test vectors $\left\vert v_{1}\right\rangle ,\ldots,\left\vert v_{t}%
\right\rangle \in\mathcal{V}$, and has gotten back responses $a_{1}%
,\ldots,a_{t}\in\mathbb{R}$. \ Let $D=\left\{  \left(  \left\vert
v_{i}\right\rangle ,a_{i}\right)  \right\}  _{i}$\ represent the data that the
algorithm has seen. \ Then conditioned on $D$, how likely are we to be in case
(i) or case (ii)? \ How much probability measure do $\mathcal{U}$\ and
$\mathcal{F}$\ respectively assign to $D$?

For case (i), the question is easy to answer: the probability measure that
$\mathcal{U}$\ assigns to $D$\ is just the Gaussian one,%
\[
\mu_{\mathcal{U}}\left(  D\right)  =\frac{e^{-\Delta_{\mathcal{U}}\left(
D\right)  /2}}{\left(  2\pi\right)  ^{t/2}},
\]
where%
\[
\Delta_{\mathcal{U}}\left(  D\right)  :=a_{1}^{2}+\cdots+a_{t}^{2}%
\]
is the squared $2$-norm of the vector of responses seen so far. \ For case
(ii), by contrast, we start with $\left\vert \Psi\right\rangle $\ drawn from
$\mathcal{N}\left(  0,1\right)  ^{N}$; then each data point restricts
$\left\vert \Psi\right\rangle $\ to the affine subspace $S_{i}$\ defined by
$\left\langle \Psi|v_{i}\right\rangle =a_{i}$. \ Let $S\left(  D\right)
=S_{1}\cap\cdots\cap S_{t}$ be the intersection of all these affine subspaces.
\ Then the probability measure that $\mathcal{F}$\ assigns to $D$\ is simply
the measure that $\mathcal{N}\left(  0,1\right)  ^{N}$\ assigns to $S\left(
D\right)  $, which in turn (by rotational symmetry) is just the minimum
squared $2$-norm of any point in $S\left(  D\right)  $, scaled by a dimension
factor. \ That is,%
\[
\mu_{\mathcal{F}}\left(  D\right)  =\frac{e^{-\Delta_{\mathcal{F}}\left(
D\right)  /2}}{\left(  2\pi\right)  ^{t/2}},
\]
where%
\[
\Delta_{\mathcal{F}}\left(  D\right)  :=\min_{\left\vert \Phi\right\rangle \in
S\left(  D\right)  }\left\langle \Phi|\Phi\right\rangle .
\]

Putting the two things together, we have%
\[
\frac{\mu_{\mathcal{F}}\left(  D\right)  }{\mu_{\mathcal{U}}\left(  D\right)
}=\exp\left(  \frac{\Delta_{\mathcal{U}}\left(  D\right)  -\Delta
_{\mathcal{F}}\left(  D\right)  }{2}\right)  .
\]
Thus, let%
\[
\Delta\left(  D\right)  :=\Delta_{\mathcal{U}}\left(  D\right)  -\Delta
_{\mathcal{F}}\left(  D\right)  .
\]
Then if we can just show that $\left\vert \Delta\left(  D\right)  \right\vert
$\ remains $o\left(  1\right)  $\ after $t$\ queries, we will have shown that
the algorithm cannot have distinguished case (i)\ from case (ii) with constant
bias after $t$\ queries. \ Thus, upper-bounding $\left\vert \Delta\left(
D\right)  \right\vert $\ will be our focus for the rest of the proof.

\subsection{\label{DELTA}Upper-Bounding $\left\vert \Delta\left(  D\right)
\right\vert $}

As a first observation, we cannot hope to show that $\left\vert \Delta\left(
D\right)  \right\vert $\ remains small with \textit{certainty}. \ Indeed, even
after just $2$\ queries, $\left\vert \Delta\left(  D\right)  \right\vert
$\ could be unboundedly large, if the responses $a_{1}$\ and $a_{2}$\ were far
out in the tails of $\mathcal{N}\left(  0,1\right)  $. \ Thus, our only hope
is to show that, after few enough queries, $\left\vert \Delta\left(  D\right)
\right\vert $\ remains small with \textit{high probability}. \ But do we mean
high probability with respect to $\mathcal{U}$\ or $\mathcal{F}$? \ Crucially,
we claim that the answer doesn't matter. \ To see this, suppose (for example)
that we have $\left\vert \Delta\left(  D\right)  \right\vert =o\left(
1\right)  $\ with probability $1-o\left(  1\right)  $\ over data $D$ drawn
according to $\mathcal{U}$. \ Then with probability\ $1-o\left(  1\right)  $
over $\mathcal{U}$, we have%
\[
\frac{\mu_{\mathcal{F}}\left(  D\right)  }{\mu_{\mathcal{U}}\left(  D\right)
}=\exp\left(  \frac{\Delta\left(  D\right)  }{2}\right)  =\exp\left(  \pm
o\left(  1\right)  \right)  =1\pm o\left(  1\right)  .
\]
It follows that we also have $\left\vert \Delta\left(  D\right)  \right\vert
=o\left(  1\right)  $\ with probability $1-o\left(  1\right)  $\ over data $D$
drawn according to $\mathcal{F}$. \ So for simplicity, we will assume the data
to be drawn according to $\mathcal{U}$.

Let us look more closely at the difference $\Delta\left(  D\right)
=\Delta_{\mathcal{U}}\left(  D\right)  -\Delta_{\mathcal{F}}\left(  D\right)
$. \ The $\Delta_{\mathcal{U}}\left(  D\right)  $\ component is easy to
compute, since it is just $a_{1}^{2}+\cdots+a_{t}^{2}$. \ For the
$\Delta_{\mathcal{F}}\left(  D\right)  $ component, on the other hand, we need
to solve the linear-algebra problem of finding the distance between the affine
subspace $S\left(  D\right)  $\ and the origin. \ We can do this using
Gram-Schmidt orthogonalization (see Section \ref{GS}). \ That is, for each
$i\in\left[  t\right]  $, we define $\left\vert w_{i}\right\rangle $
recursively as the normalized projection of $\left\vert v_{i}\right\rangle
$\ onto the orthogonal complement of the subspace spanned by $\left\vert
w_{1}\right\rangle ,\ldots,\left\vert w_{i-1}\right\rangle $.\ \ We can
express $\left\vert w_{i}\right\rangle $ recursively as%
\[
\left\vert w_{i}\right\rangle =\beta_{i}\left(  \left\vert v_{i}\right\rangle
-\sum_{j=1}^{i-1}\left\langle v_{i}|w_{j}\right\rangle \left\vert
w_{j}\right\rangle \right)  ,
\]
where $\beta_{i}$\ is a normalizing constant. \ Let us also define%
\begin{align*}
b_{i}  &  =\left\langle \Psi|w_{i}\right\rangle \\
&  =\beta_{i}\left(  \left\langle \Psi|v_{i}\right\rangle -\sum_{j=1}%
^{i-1}\left\langle v_{i}|w_{j}\right\rangle \left\langle \Psi|w_{j}%
\right\rangle \right) \\
&  =\beta_{i}\left(  a_{i}-\sum_{j=1}^{i-1}\left\langle v_{i}|w_{j}%
\right\rangle b_{j}\right)  .
\end{align*}
Then we have:%
\begin{align*}
\Delta_{\mathcal{F}}\left(  D\right)   &  =\min_{\left\vert \Phi\right\rangle
\in S\left(  D\right)  }\left\langle \Phi|\Phi\right\rangle \\
&  =\min_{\left\vert \Phi\right\rangle \in S\left(  D\right)  }\sum_{i=1}%
^{t}\left\langle \Phi|w_{i}\right\rangle ^{2}\\
&  =\sum_{i=1}^{t}\left\langle \Psi|w_{i}\right\rangle ^{2}\\
&  =b_{1}^{2}+\cdots+b_{t}^{2}%
\end{align*}
where the third line used the orthogonality of the $\left\vert w_{i}%
\right\rangle $'s.

To simplify matters, let us define a variant of $b_{i}$\ where we omit all the
normalization factors $\beta_{i}$:%
\begin{equation}
c_{i}:=a_{i}-\sum_{j=1}^{i-1}\left\langle v_{i}|w_{j}\right\rangle c_{j}.
\label{cdef}%
\end{equation}
Also, call the data $D$ \textit{well-behaved} if $\left\vert a_{i}\right\vert
\leq\sqrt{2\ln100t}$ for all $i\in\left[  t\right]  $.

\begin{proposition}
\label{wellbehaved}$D$ is well-behaved with probability at least $0.99$ over
$\mathcal{U}$.
\end{proposition}

\begin{proof}
Follows from the union bound, together with the fact that each $a_{i}$\ is an
independent\ $\mathcal{N}\left(  0,1\right)  $\ Gaussian, so%
\[
\Pr\left[  \left\vert a_{i}\right\vert >\sqrt{2\ln100t}\right]  <\frac
{1}{100t}.
\]

\end{proof}

Then we have the following useful lemma.

\begin{lemma}
\label{nobeta}Let $t\leq\sqrt{N}/10$, and suppose $D$ is well-behaved. \ Then
$\left\vert c_{i}-b_{i}\right\vert =O\left(  i\frac{\sqrt{\log N}}{N}\right)
$ for all $i\in\left[  t\right]  $.
\end{lemma}

\begin{proof}
If $\left\vert a_{i}\right\vert \leq\sqrt{2\ln100t}$\ for all $i\in\left[
t\right]  $, then certainly $\left\vert b_{i}\right\vert =O(\sqrt{\log
t})=O(\sqrt{\log N})$\ for all $i\in\left[  t\right]  $ as well, since%
\begin{align*}
\left\vert b_{i}\right\vert  &  \leq\beta_{i}\left(  \left\vert a_{i}%
\right\vert +\sum_{j=1}^{i-1}\left\vert \left\langle v_{i}|w_{j}\right\rangle
\right\vert \left\vert b_{j}\right\vert \right) \\
&  \leq\left(  1+\frac{0.2}{\sqrt{N}}\right)  \left(  \left\vert
a_{i}\right\vert +\frac{0.2}{\sqrt{N}}\sum_{j=1}^{i-1}\left\vert
b_{j}\right\vert \right) \\
&  =O\left(  \max_{j\in\left[  i\right]  }\left\vert a_{j}\right\vert \right)
,
\end{align*}
where the second line used Lemma \ref{gslem} and the third used $i\leq
t\leq\sqrt{N}/10$, together with induction on $j$. \ Now,%
\begin{align*}
c_{i}-b_{i}  &  =a_{i}-\sum_{j=1}^{i-1}\left\langle v_{i}|w_{j}\right\rangle
c_{j}-\beta_{i}\left(  a_{i}-\sum_{j=1}^{i-1}\left\langle v_{i}|w_{j}%
\right\rangle b_{j}\right) \\
&  =\left(  1-\beta_{i}\right)  a_{i}-\sum_{j=1}^{i-1}\left\langle v_{i}%
|w_{j}\right\rangle \left(  c_{j}-\beta_{i}b_{j}\right)  .
\end{align*}
So%
\begin{align*}
\left\vert c_{i}-b_{i}\right\vert  &  \leq\left(  \beta_{i}-1\right)
\left\vert a_{i}\right\vert +\sum_{j=1}^{i-1}\left\vert \left\langle
v_{i}|w_{j}\right\rangle \right\vert \left(  \left\vert c_{j}-b_{j}\right\vert
+\left(  \beta_{i}-1\right)  \left\vert b_{j}\right\vert \right) \\
&  \leq\frac{2i\left\vert a_{i}\right\vert }{N}+\frac{0.2}{\sqrt{N}}\sum
_{j=1}^{i-1}\left(  \frac{2i\left\vert b_{j}\right\vert }{N}+\left\vert
c_{j}-b_{j}\right\vert \right) \\
&  =O\left(  \frac{i\sqrt{\log N}}{N}+\frac{i^{2}\sqrt{\log N}}{N^{3/2}%
}\right)  +\frac{1}{\sqrt{N}}\sum_{j=1}^{i-1}\left\vert c_{j}-b_{j}\right\vert
\\
&  =O\left(  \frac{i\sqrt{\log N}}{N}\right)  +\frac{1}{\sqrt{N}}\sum
_{j=1}^{i-1}\left\vert c_{j}-b_{j}\right\vert
\end{align*}
where the second line used Lemma \ref{gslem}\ and the last used $i\leq
t\leq\sqrt{N}/10$. \ So, letting $\varepsilon_{i}$\ be an upper bound on
$\left\vert c_{j}-b_{j}\right\vert $ for all $j\leq i$, we have%
\begin{align*}
\varepsilon_{i}  &  =O\left(  \frac{i\sqrt{\log N}}{N}\right)  +\frac{i}%
{\sqrt{N}}\varepsilon_{i-1}\\
&  =O\left(  \frac{i\sqrt{\log N}}{N}\right)  +\frac{i}{\sqrt{N}}%
\varepsilon_{i}.
\end{align*}
Rearranging, we have%
\[
0.9\varepsilon_{i}=O\left(  \frac{i\sqrt{\log N}}{N}\right)
\]
and are done.
\end{proof}

As a first consequence of Lemma \ref{nobeta}, if $D$ is well-behaved, then%
\[
\left\vert c_{i}\right\vert =O\left(  \sqrt{\log t}+i\frac{\sqrt{\log N}}%
{N}\right)  =O(\sqrt{\log t})
\]
for all $i\in\left[  t\right]  $. \ As a more important consequence, let%
\[
\Delta_{\mathcal{F}}^{\prime}\left(  D\right)  :=c_{1}^{2}+\cdots+c_{t}^{2},
\]
and let%
\[
\Delta^{\prime}\left(  D\right)  :=\Delta_{\mathcal{U}}\left(  D\right)
-\Delta_{\mathcal{F}}^{\prime}\left(  D\right)  .
\]
Then we can restrict our attention to upper-bounding $\left\vert
\Delta^{\prime}\left(  D\right)  \right\vert $, rather than the more
complicated $\left\vert \Delta\left(  D\right)  \right\vert $. \ For by Lemma
\ref{nobeta}, if $D$ is well-behaved, then%
\begin{align*}
\left\vert \Delta_{\mathcal{F}}\left(  D\right)  -\Delta_{\mathcal{F}}%
^{\prime}\left(  D\right)  \right\vert  &  =\left\vert \sum_{i=1}^{t}\left(
b_{i}^{2}-c_{i}^{2}\right)  \right\vert \\
&  \leq\sum_{i=1}^{t}\left(  \left\vert b_{i}\right\vert +\left\vert
c_{i}\right\vert \right)  \left\vert c_{i}-b_{i}\right\vert \\
&  =\sum_{i=1}^{t}O\left(  \sqrt{\log N}\cdot i\frac{\sqrt{\log N}}{N}\right)
\\
&  =O\left(  t^{2}\frac{\log N}{N}\right)  .
\end{align*}
So if $\left\vert \Delta^{\prime}\left(  D\right)  \right\vert =o\left(
1\right)  $ and $t=o\left(  \sqrt{\frac{N}{\log N}}\right)  $, then by the
triangle inequality,%
\[
\left\vert \Delta\left(  D\right)  \right\vert \leq\left\vert \Delta^{\prime
}\left(  D\right)  \right\vert +\left\vert \Delta_{\mathcal{F}}\left(
D\right)  -\Delta_{\mathcal{F}}^{\prime}\left(  D\right)  \right\vert
\]
is $o\left(  1\right)  $\ as well. \ Thus, from now on our goal is to
upper-bound $\left\vert \Delta^{\prime}\left(  D\right)  \right\vert $.

Let%
\begin{equation}
r_{i}:=a_{i}-c_{i}=\sum_{j=1}^{i-1}\left\langle v_{i}|w_{j}\right\rangle
c_{j}. \label{ridef}%
\end{equation}
Notice that, if we unravel the recursive definition of $c_{j}$, we find that
$r_{i}$\ is a linear combination of $a_{1},\ldots,a_{i-1}$, with no dependence
on $a_{i}$. \ Though we will not need this for the proof, $r_{i}$\ has an
interesting interpretation, as the \textit{expected} value of $a_{i}$\ after
$a_{1},\ldots,a_{i-1}$\ have been queried but before $a_{i}$\ has been
queried, assuming the data were drawn from the forrelated distribution
$\mathcal{F}$. \ Now,%
\begin{align}
\Delta^{\prime}\left(  D\right)   &  =\Delta_{\mathcal{U}}\left(  D\right)
-\Delta_{\mathcal{F}}^{\prime}\left(  D\right) \nonumber\\
&  =\sum_{i=1}^{t}\left(  a_{i}^{2}-c_{i}^{2}\right) \nonumber\\
&  =\sum_{i=1}^{t}r_{i}\left(  2a_{i}-r_{i}\right)  . \label{rieq}%
\end{align}
As we show in the next lemma, the above means that our problem can in turn be
reduced to upper-bounding the $r_{i}$'s.

\begin{lemma}
\label{azumapart}Suppose $\left\vert r_{i}\right\vert \leq\frac{1}%
{1750\sqrt{t}}$\ for all $i\in\left[  t\right]  $.\ \ Then%
\[
\left\vert \sum_{i=1}^{t}r_{i}a_{i}\right\vert \leq0.01
\]
with probability at least $0.99$\ over the data $D$.
\end{lemma}

\begin{proof}
Notice that each $r_{i}a_{i}$\ has an expectation of $0$, even after
conditioning on $a_{1},\ldots,a_{i-1}$. \ This is because, according to the
measure $\mathcal{U}$, each $a_{i}$\ is a \textquotedblleft
fresh\textquotedblright\ $\mathcal{N}\left(  0,1\right)  $\ Gaussian,
uncorrelated with $a_{1},\ldots,a_{i-1}$, whereas $r_{i}$\ is a linear
combination of $a_{1},\ldots,a_{i-1}$ that does not depend on $a_{i}$. \ Thus,
$r_{1}a_{1},\ldots,r_{t}a_{t}$\ forms a martingale difference sequence, in
which, conditioned on its predecessors, each $r_{i}a_{i}$\ is an
$\mathcal{N}\left(  0,r_{i}^{2}\right)  $\ Gaussian, for some $\left\vert
r_{i}\right\vert \leq\frac{1}{1750\sqrt{t}}$. \ Set $\epsilon:=1/1750$. \ Then
by Lemma \ref{azumagauss},%
\begin{align*}
\Pr\left[  \left\vert \sum_{i=1}^{t}r_{i}a_{i}\right\vert >0.01\right]   &
\leq\Pr\left[  \left\vert \sum_{i=1}^{t}r_{i}a_{i}\right\vert >\sqrt{56\ln
200}\frac{\epsilon}{\sqrt{t}}\sqrt{t}\right] \\
&  <2\exp\left(  -\frac{(\sqrt{56\ln200})^{2}}{56}\right) \\
&  =0.01.
\end{align*}

\end{proof}

Thus, suppose $\left\vert r_{i}\right\vert \leq\frac{1}{1750\sqrt{t}}$ for all
$i\in\left[  t\right]  $. \ Then by Lemma \ref{azumapart} and equation
(\ref{rieq}), we have%
\begin{align*}
\left\vert \Delta^{\prime}\left(  D\right)  \right\vert  &  =\left\vert
\sum_{i=1}^{t}r_{i}\left(  2a_{i}-r_{i}\right)  \right\vert \\
&  \leq2\left\vert \sum_{i=1}^{t}r_{i}a_{i}\right\vert +\sum_{i=1}^{t}%
r_{i}^{2}\\
&  \leq0.02+\frac{1}{1750^{2}}%
\end{align*}
with probability at least $0.99$\ over $D$. \ This implies that the algorithm
has not yet succeeded at distinguishing $\mathcal{F}$\ from $\mathcal{U}%
$\ with bias (say) $1/2$. \ So in summary, if we can show that with high
probability, $\left\vert r_{i}\right\vert =O(1/\sqrt{t})$\ for all
$i\in\left[  t\right]  $, then we have shown that the algorithm must make
$\Omega\left(  t\right)  $ queries.

\subsection{\label{RISEC}Upper-Bounding $\left\vert r_{i}\right\vert $}

We now turn to the problem of upper-bounding $\left\vert r_{i}\right\vert $
(with high probability over $D$), for all $i\in\left[  t\right]  $. \ The
better the upper bound on $\left\vert r_{i}\right\vert $ we can achieve, the
better will be our lower bound on $t$. \ To illustrate, it is easy to prove
the following crude bound:

\begin{proposition}
\label{crude}If $D$ is well-behaved and $t<\sqrt{N}/10$, then $\left\vert
r_{i}\right\vert =O\left(  i\sqrt{\frac{\log N}{N}}\right)  $ for all
$i\in\left[  t\right]  $.
\end{proposition}

\begin{proof}
We noted before that if $D$\ is well-behaved then $\left\vert c_{i}\right\vert
=O(\sqrt{\log t})=O(\sqrt{\log N})$ for all $i$. \ So by Lemma \ref{gslem},%
\begin{align*}
\left\vert r_{i}\right\vert  &  \leq\sum_{j=1}^{i-1}\left\vert \left\langle
v_{i}|w_{j}\right\rangle \right\vert \left\vert c_{j}\right\vert \\
&  \leq i\cdot\frac{2}{\sqrt{N}}\cdot O(\sqrt{\log N}).
\end{align*}

\end{proof}

Setting $t\sqrt{\frac{\log N}{N}}=1/\sqrt{t}$\ and solving, Proposition
\ref{crude} yields a lower bound of $t=\Omega\left(  \left(  \frac{N}{\log
N}\right)  ^{1/3}\right)  $\ queries.\footnote{Furthermore, notice that
Proposition \ref{crude}\ has no dependence on the number of test vectors $M$.
\ This is why it implies a $\Omega\left(  \frac{\left(  1/\varepsilon\right)
^{2/3}}{\left(  \log1/\varepsilon\right)  ^{1/3}}\right)  $ lower bound for
\textsc{Gaussian Distinguishing}, independent of $M$.}

With some more work, one can prove a bound of $\left\vert r_{i}\right\vert
=O\left(  \sqrt{\frac{i\log Mt}{N}}+\frac{i^{2}}{N}\right)  $, which yields a
lower bound of $t=\Omega\left(  N^{2/5}\right)  $ queries whenever $M\leq
\exp\left(  O\left(  N^{1/5}\right)  \right)  $. \ In this section, however,
we will go for the bound $\left\vert r_{i}\right\vert =O\left(  \sqrt
{\frac{i\log Mt}{N}}\right)  $, which yields a lower bound of $t=\Omega\left(
\frac{\sqrt{N}}{\log MN}\right)  $ queries. \ For \textsc{Real Forrelation},
of course, we have $M=2N$, and therefore $t=\Omega\left(  \frac{\sqrt{N}}{\log
N}\right)  $\ as desired.

Our strategy will be to make repeated use of the following lemma.

\begin{lemma}
[Central Martingale Lemma]\label{cml}Suppose $200\leq t<\sqrt{N}/10$. \ Then
with probability at least $0.99$\ over data $D$ drawn from $\mathcal{U}$, we
have%
\[
\left\vert \sum_{j=1}^{i-1}\left\langle v|w_{j}\right\rangle a_{j}\right\vert
\leq30\sqrt{\frac{i\ln Mt}{N}}%
\]
for all $i\in\left[  t\right]  $\ and for all test vectors $\left\vert
v\right\rangle \in\mathcal{V}$.
\end{lemma}

\begin{proof}
Fix any $\left\vert v\right\rangle \in\mathcal{V}$. \ By Lemma \ref{gslem}, we
have $\left\vert \left\langle v|w_{j}\right\rangle \right\vert \leq2/\sqrt{N}$
for all $\left\vert v\right\rangle $\ and $\left\vert w_{j}\right\rangle $.
Also, recall that each $a_{j}$\ is a \textquotedblleft fresh\textquotedblright%
\ $\mathcal{N}\left(  0,1\right)  $\ Gaussian, and that $\left\langle
v|w_{j}\right\rangle $\ does not depend on $a_{j}$. \ Thus, $\left\langle
v|w_{1}\right\rangle a_{1},\ldots,\left\langle v|w_{i-1}\right\rangle a_{i-1}%
$\ forms a martingale difference sequence, in which, conditioned on its
predecessors, each $\left\langle v|w_{j}\right\rangle a_{j}$\ is an
$\mathcal{N}(0,\left\langle v|w_{j}\right\rangle ^{2})$\ Gaussian. \ So by
Lemma \ref{azumagauss},%
\begin{align*}
\Pr\left[  \left\vert \sum_{j=1}^{i-1}\left\langle v|w_{j}\right\rangle
a_{j}\right\vert >30\sqrt{\frac{i\ln Mt}{N}}\right]   &  <\Pr\left[
\left\vert \sum_{j=1}^{i-1}\left\langle v|w_{j}\right\rangle a_{j}\right\vert
>\sqrt{56\ln(200Mt)}\frac{2}{\sqrt{N}}\sqrt{i}\right] \\
&  <2\exp\left(  -\frac{(\sqrt{56\ln(200Mt)})^{2}}{56}\right) \\
&  =\frac{1}{100Mt}.
\end{align*}
The result now follows by taking the union bound over all $\left\vert
v\right\rangle \in\mathcal{V}$\ and $i\in\left[  t\right]  $.
\end{proof}

We can now prove the desired upper bound on $\left\vert r_{i}\right\vert $.

\begin{lemma}
\label{rilem}Suppose $t<\sqrt{N}/10$. \ Then with probability at least
$0.99$\ over $D$, we have $\left\vert r_{i}\right\vert =O\left(  \sqrt
{\frac{i\log Mt}{N}}\right)  $\ for all $i\in\left[  t\right]  $.
\end{lemma}

\begin{proof}
Taking the equation for $r_{i}$\ (equation (\ref{ridef})), and unraveling the
recursive definition of $c_{j}$\ (equation (\ref{cdef})), we get%
\begin{align*}
r_{i}  &  =\sum_{j=1}^{i-1}\left\langle v_{i}|w_{j}\right\rangle c_{j}\\
&  =\sum_{j=1}^{i-1}\left\langle v_{i}|w_{j}\right\rangle \left(  a_{j}%
-\sum_{k=1}^{j-1}\left\langle v_{j}|w_{k}\right\rangle c_{k}\right) \\
&  =\sum_{j=1}^{i-1}\left\langle v_{i}|w_{j}\right\rangle \left(  a_{j}%
-\sum_{k=1}^{j-1}\left\langle v_{j}|w_{k}\right\rangle \left(  a_{k}%
-\sum_{\ell=1}^{k-1}\left\langle v_{k}|w_{\ell}\right\rangle c_{\ell}\right)
\right) \\
&  \vdots
\end{align*}
Thus,%
\begin{align*}
\left\vert r_{i}\right\vert  &  \leq\left\vert \sum_{j=1}^{i-1}\left\langle
v_{i}|w_{j}\right\rangle a_{j}\right\vert +\sum_{j=1}^{i-1}\left\vert
\left\langle v_{i}|w_{j}\right\rangle \right\vert \left\vert \sum_{k=1}%
^{j-1}\left\langle v_{j}|w_{k}\right\rangle a_{k}\right\vert +\sum_{j=1}%
^{i-1}\left\vert \left\langle v_{i}|w_{j}\right\rangle \right\vert \sum
_{k=1}^{j-1}\left\vert \left\langle v_{j}|w_{k}\right\rangle \right\vert
\left\vert \sum_{\ell=1}^{k-1}\left\langle v_{k}|w_{\ell}\right\rangle
a_{\ell}\right\vert +\cdots\\
&  \leq30\sqrt{\frac{i\ln Mt}{N}}+\sum_{j=1}^{i-1}\frac{2}{\sqrt{N}}%
30\sqrt{\frac{j\ln Mt}{N}}+\sum_{j=1}^{i-1}\sum_{k=1}^{j-1}\left(  \frac
{2}{\sqrt{N}}\right)  ^{2}30\sqrt{\frac{k\ln Mt}{N}}+\cdots\\
&  \leq30\sqrt{\frac{\ln Mt}{N}}\left[  \sqrt{i}+\frac{2}{\sqrt{N}}%
i^{3/2}+\left(  \frac{2}{\sqrt{N}}\right)  ^{2}i^{5/2}+\cdots\right] \\
&  =30\sqrt{\frac{i\ln Mt}{N}}\left[  1+\frac{2i}{\sqrt{N}}+\left(  \frac
{2i}{\sqrt{N}}\right)  ^{2}+\cdots\right] \\
&  =O\left(  \sqrt{\frac{i\log Mt}{N}}\right)
\end{align*}
where the second line used Lemmas \ref{gslem}\ and \ref{cml}.
\end{proof}

\section{Simulation of $t$-Query Quantum Algorithms\label{SIM}}

Let $\mathcal{A}$\ be a quantum algorithm that makes $t=O\left(  1\right)  $
queries to a Boolean input $x\in\left\{  -1,1\right\}  ^{N}$, and then either
accepts or rejects. \ In this section, we show how to estimate $\mathcal{A}$'s
acceptance probability, on all inputs $x$, by a classical, nonadaptive
randomized algorithm that makes only $O(N^{1-1/2t})$\ queries to $x$.

So for example, we can simulate any $1$-query quantum algorithm using
$O(\sqrt{N})$\ classical queries---thereby showing that the $1$\ versus
$\Omega(\frac{\sqrt{N}}{\log N})$\ separation of Section \ref{GAP} is nearly
tight. \ More generally, resolving an open problem of Buhrman et
al.\ \cite{bfnr}, we find that there is \textit{no} partial Boolean function
whose quantum query complexity is constant\ but whose randomized query
complexity is linear.

We obtain our simulation of quantum algorithms as a consequence of a much more
general result: namely, that \textit{any bounded, degree-}$k$\textit{
polynomial }$p:\left\{  -1,1\right\}  ^{N}\rightarrow\mathbb{R}$\textit{,
which satisfies a technical condition called \textquotedblleft
block-multilinearity,\textquotedblright\ can be estimated by querying only
}$O(N^{1-1/k})$\textit{\ of its variables.} \ This result makes no direct
reference to quantum computing, and seems likely to have independent
applications---for example, to the design of classical sublinear algorithms.
\ We strongly conjecture that the block-multilinearity condition can be
removed, which would further heighten the non-quantum interest of this result.
\ In Appendix \ref{POLY}, we prove that conjecture in the special case $k=2$.

More formally, let $p:\mathbb{R}^{N}\rightarrow\mathbb{R}$\ be a real
polynomial of degree $k$. \ Since we will only care about $p$'s behavior on
the Boolean hypercube $\left\{  -1,1\right\}  ^{N}$, we can assume without
loss of generality that $p$ is multilinear (that is, that no variable is
raised to a higher power than $1$). \ We call $p$ \textit{bounded} if
$p\left(  x\right)  \in\left[  -1,1\right]  $\ for all $x\in\left\{
-1,1\right\}  ^{N}$.\footnote{In quantum query complexity, normally we would
call a polynomial $p$ \textquotedblleft bounded\textquotedblright\ if
$p\left(  x\right)  \in\left[  0,1\right]  $\ for all $x\in\left\{
-1,1\right\}  ^{N}$---in other words, if $p$ represents a probability. \ As we
will see, however, we need to consider polynomials that can represent
arbitrary inner products between vectors of norm at most $1$, and can
therefore only assume $p\left(  x\right)  \in\left[  -1,1\right]  $.} \ Now,
we call $p$ \textit{block-multilinear} if its $N$\ variables $x_{1}%
,\ldots,x_{N}$\ can be partitioned into $k$ blocks, $R_{1},\ldots,R_{k}$, so
that every monomial of $p$ contains exactly one variable from each block.
\ Note that block-multilinearity implies, in particular, that $p$ is
homogeneous. \ Also, by introducing at most $O\left(  N\right)  $\ dummy
variables, we can assume without loss of generality that every block has the
same size, $\left\vert R_{1}\right\vert =\cdots=\left\vert R_{k}\right\vert
=n=N/k$.

We can now state the main result of this section.

\begin{theorem}
\label{estimatorthm}Let $p:\left\{  -1,1\right\}  ^{N}\rightarrow\left[
-1,1\right]  $\ be any bounded block-multilinear polynomial of degree $k$.
\ Then there exists a classical randomized algorithm that, on input
$x\in\left\{  -1,1\right\}  ^{N}$, nonadaptively queries $O(\left(
N/\varepsilon^{2}\right)  ^{1-1/k})$ bits of $x$, and then outputs an estimate
$\tilde{p}$ such that with high probability,
\[
\left\vert \tilde{p}-p(x_{1},\ldots,x_{N})\right\vert \leq\varepsilon.
\]
(Here the big-O hides a multiplicative constant that is exponential in $k$.)
\end{theorem}

Before plunging into the proof of Theorem \ref{estimatorthm}, let us explain
why it implies the desired conclusion about quantum algorithms. \ The key
observation relating quantum query complexity to low-degree polynomials was
made by Beals et al.\ \cite{bbcmw} in 1998:

\begin{lemma}
[Beals et al.\ \cite{bbcmw}]\label{bbcmwlem}Given any quantum algorithm
$\mathcal{A}$\ that makes $t$ queries to a Boolean input $x\in\left\{
-1,1\right\}  ^{N}$, the probability that $\mathcal{A}$\ accepts can be
expressed as a real multilinear polynomial $p\left(  x\right)  $,\ of degree
at most $2t$. \ (Thus, in particular, $p\left(  x\right)  \in\left[
0,1\right]  $\ for all $x\in\left\{  -1,1\right\}  ^{N}$.)
\end{lemma}

Note that, \textit{if} Theorem \ref{estimatorthm} worked for arbitrary
polynomials (rather than only block-multilinear ones), then combining it with
Lemma \ref{bbcmwlem} would immediately give the simulation of quantum
algorithms that we want.

Fortunately, one can strengthen Lemma \ref{bbcmwlem}, to show that a $t$-query
quantum algorithm gives rise, not just to \textit{any} bounded degree-$2t$
polynomial, but to a block-multilinear one.

\begin{lemma}
\label{blockbeals}Let $\mathcal{A}$\ be a quantum algorithm that makes $t$
queries to a Boolean input $x\in\left\{  -1,1\right\}  ^{N}$. \ Then there
exists a degree-$2t$ block-multilinear polynomial $p:\mathbb{R}^{2tN}%
\rightarrow\mathbb{R}$, with $2t$\ blocks of $N$ variables each, such that

\begin{enumerate}
\item[(i)] the probability that $\mathcal{A}$\ accepts $x$ equals $p\left(
x,\ldots,x\right)  $ (with $x$ repeated $2t$\ times), and

\item[(ii)] $p\left(  z\right)  \in\left[  -1,1\right]  $\ for all
$z\in\left\{  -1,1\right\}  ^{2tN}$.
\end{enumerate}
\end{lemma}

\begin{proof}
Assume for simplicity (and without loss of generality) that $\mathcal{A}%
$\ involves real amplitudes only.

For all $j\in\left[  t\right]  $\ and $i\in\left[  N\right]  $, let $x_{j,i}%
$\ be the value of $x_{i}$\ that $\mathcal{A}$'s oracle returns in response to
its $j^{th}$\ query. \ Of course, in any \textquotedblleft
normal\textquotedblright\ run of $\mathcal{A}$, we will have $x_{j,i}%
=x_{j^{\prime},i}$\ for all $j,j^{\prime}$: that is, the value of $x_{i}%
$\ will be consistent across all $t$\ queries. \ But it is perfectly
legitimate to ask what happens if $x$\ changes from one query to the next.
\ In any case, $\mathcal{A}$\ will have some normalized final state, of the
form%
\[
\sum_{i,w}\alpha_{i,w}\left(  x_{1,1},\ldots,x_{t,N}\right)  \left\vert
i,w\right\rangle .
\]
Furthermore, following Beals et al.\ \cite{bbcmw},\ it is easy to see that
each amplitude $\alpha_{i,w}$\ can be written as a degree-$t$
block-multilinear polynomial\ in the $tN$\ variables $x_{1,1},\ldots,x_{t,N}$,
with one block of $N$ variables, $R_{j}=\left\{  x_{j,1},\ldots,x_{j,N}%
\right\}  $, corresponding to each of the $t$ queries. \ (If $\mathcal{A}$ has
basis states that do not participate in queries, then we can deal with that by
introducing dummy variables, $x_{1,0},\ldots,x_{t,0}$, which are set to
$1$\ in any \textquotedblleft normal\textquotedblright\ run of $\mathcal{A}$.)

Next, for all $j\in\left[  t\right]  $\ and $i\in\left[  N\right]  $, we
create a \textit{second} variable $x_{t+j,i}$, which just like $x_{j,i}$,
represents the value of $x_{i}$\ that $\mathcal{A}$'s oracle returns in
response to its $j^{th}$\ query. \ Let $\operatorname*{Acc}$\ be the set of
all accepting basis states, and consider the polynomial%
\[
p\left(  x_{1,1},\ldots,x_{2t,N}\right)  :=\sum_{\left(  i,w\right)
\in\operatorname*{Acc}}\alpha_{i,w}\left(  x_{1,1},\ldots,x_{t,N}\right)
\alpha_{i,w}\left(  x_{t+1,1},\ldots,x_{2t,N}\right)  .
\]
By construction, $p$ is a degree-$2t$ block-multilinear polynomial in the
$2tN$\ variables $x_{1,1},\ldots,x_{2t,N}$, with one block of $N$ variables,
$R_{j}=\left\{  x_{j,1},\ldots,x_{j,N}\right\}  $, for each $j\in\left[
2t\right]  $. \ Furthermore, if we repeat the same input\ $x\in\left\{
-1,1\right\}  ^{N}$\ across all $2t$\ blocks, then%
\[
p\left(  x,\ldots,x\right)  =\sum_{\left(  i,w\right)  \in\operatorname*{Acc}%
}\alpha_{i,w}^{2}\left(  x,\ldots,x\right)
\]
is simply the probability that $\mathcal{A}$ accepts $x$. \ Finally, even if
$x_{1,1},\ldots,x_{2t,N}\in\left\{  -1,1\right\}  ^{2tN}$ is completely
arbitrary, $p\left(  x_{1,1},\ldots,x_{2t,N}\right)  $\ still represents an
inner product between two vectors,%
\[
\sum_{\left(  i,w\right)  \in\operatorname*{Acc}}\alpha_{i,w}\left(
x_{1,1},\ldots,x_{t,N}\right)  \left\vert i,w\right\rangle \text{ \ \ \ \ and
\ \ \ \ }\sum_{\left(  i,w\right)  \in\operatorname*{Acc}}\alpha_{i,w}\left(
x_{t+1,1},\ldots,x_{2t,N}\right)  \left\vert i,w\right\rangle .
\]
Since both of these vectors have norm at most $1$, their inner product is
bounded in $\left[  -1,1\right]  $.
\end{proof}

As a side note, given any Boolean function $f:\left\{  -1,1\right\}
^{N}\rightarrow\left\{  0,1\right\}  $, one can consider the minimum degree of
any block-multilinear polynomial $p$\ that approximates $f$. \ More formally,
let the \textit{block-multilinear approximate degree} of $f$, or
$\widetilde{\operatorname*{bmdeg}}\left(  f\right)  $, be the minimum degree
of any block-multilinear polynomial $p:\mathbb{R}^{kN}\rightarrow\mathbb{R}$,
with $k$\ blocks of $N$ variables each, such that

\begin{enumerate}
\item[(i)] $\left\vert p\left(  x,\ldots,x\right)  -f\left(  x\right)
\right\vert \leq\frac{1}{3}$ for all $x\in\left\{  -1,1\right\}  ^{N}$ (or
alternatively, for all $x$ satisfying some promise), and

\item[(ii)] $p\left(  x_{1,1},\ldots,x_{k,N}\right)  \in\left[  -1,1\right]  $
for all $x_{1,1},\ldots,x_{k,N}\in\left\{  -1,1\right\}  ^{kN}$.
\end{enumerate}

Recall that $\widetilde{\deg}\left(  f\right)  $, the \textquotedblleft
ordinary\textquotedblright\ approximate degree of $f$, is the minimum degree
of any polynomial $p:\mathbb{R}^{N}\rightarrow\mathbb{R}$\ such that
$\left\vert p\left(  x\right)  -f\left(  x\right)  \right\vert \leq\frac{1}%
{3}$\ for all $x$. \ Lemma \ref{bbcmwlem} of Beals et al.\ \cite{bbcmw}%
\ implies that $\widetilde{\deg}\left(  f\right)  \leq2\operatorname*{Q}%
\left(  f\right)  $\ for all $f$, where $\operatorname*{Q}\left(  f\right)
$\ is the bounded-error quantum query complexity of $f$.

Clearly $\widetilde{\deg}\left(  f\right)  \leq
\widetilde{\operatorname*{bmdeg}}\left(  f\right)  $\ for all $f$, by
identifying variables across the $k$ blocks. \ Also, Lemma \ref{blockbeals}%
\ implies that $\widetilde{\operatorname*{bmdeg}}\left(  f\right)
\leq2\operatorname*{Q}\left(  f\right)  $. \ Putting these facts together, we
find that $\widetilde{\operatorname*{bmdeg}}\left(  f\right)  $\ is a lower
bound on quantum query complexity that is \textit{at least} as good as
$\widetilde{\deg}\left(  f\right)  $, and \textit{might} sometimes be better.
\ We do not currently know whether there is any asymptotic separation between
$\widetilde{\deg}\left(  f\right)  $ and $\widetilde{\operatorname*{bmdeg}%
}\left(  f\right)  $, nor do we know whether there is an asymptotic separation
between $\widetilde{\operatorname*{bmdeg}}\left(  f\right)  $\ and
$\operatorname*{Q}\left(  f\right)  $. \ Note that Ambainis
\cite{ambainis:deg}\ exhibited a Boolean function $f$ such that
$\widetilde{\deg}\left(  f\right)  =O\left(  \operatorname*{Q}\left(
f\right)  ^{0.76}\right)  $. \ By contrast, we do not know any techniques for
upper-bounding $\widetilde{\operatorname*{bmdeg}}\left(  f\right)  $, that do
not \textit{also} upper-bound $\operatorname*{Q}\left(  f\right)  $.

\subsection{Preprocessing the Polynomial\label{PREPROCESS}}

We are now ready to prove Theorem \ref{estimatorthm}. \ Thus, suppose%
\[
p\left(  x_{1,1},\ldots,x_{k,N}\right)  =\sum_{i_{1},\ldots,i_{k}\in\left[
N\right]  }a_{i_{1},\ldots,i_{k}}x_{1,i_{1}}\cdots x_{k,i_{k}}%
\]
is a bounded block-multilinear polynomial of degree $k$. \ Then in our
estimation procedure, the first step is to preprocess $p$, in order to
\textquotedblleft balance\textquotedblright\ it, and get rid of any variables
that are \textquotedblleft too influential.\textquotedblright\ \ More
formally, set $\delta:=\varepsilon^{2}/N$. \ Then we wish to achieve the
following requirement:\ for every nonempty set $S\subseteq\left[  k\right]  $,%
\begin{equation}
\Lambda_{S}:=\sum_{(i_{j})_{j\in S}}\left(  \sum_{(i_{j})_{j\notin S}}%
a_{i_{1},\ldots,i_{k}}\right)  ^{2}\leq\delta. \label{eq:need}%
\end{equation}

The basic operation that we use to achieve this requirement is
\textit{variable-splitting}. \ The operation consists of taking a variable
$x_{j,l}$ and replacing it by $m$ variables, in the following way. \ We
introduce $m$ new variables $x_{j,l_{1}},\ldots,x_{j,l_{m}}$, and define
$p^{\prime}$ as the polynomial obtained by substituting $\frac{x_{j,l_{1}%
}+\cdots+x_{j,l_{m}}}{m}$ in the polynomial $p$ instead of $x_{j,l}$. \ We
refer to this as splitting $x_{j,l}$ into $m$ variables. \ Observe that
variable-splitting preserves the property that $p$ is bounded in $\left[
-1,1\right]  $ at all Boolean points---for, regardless of how we set
$x_{j,l_{1}},\ldots,x_{j,l_{m}}$, the value of $p^{\prime}$\ will simply equal
the value of $p$\ with $x_{j,l}$\ set to $\frac{x_{j,l_{1}}+\cdots+x_{j,l_{m}%
}}{m}$, which in turn is a convex combination of $p$\ with $x_{j,l}$\ set
to\ $-1$ and $p$\ with $x_{j,l}$\ set to\ $1$.

\begin{lemma}
\label{cl:main}Let $S\subseteq\left[  k\right]  $ be nonempty. \ Then there is
a sequence of variable-splittings that introduces at most $1/\delta$ new
variables, and that produces a polynomial $p^{\prime}$ that satisfies
$\Lambda_{S}\leq\delta$.
\end{lemma}

\begin{proof}
We start with the case $S=\left[  k\right]  $. \ Then we have to ensure
\begin{equation}
\sum_{i_{1},\ldots,i_{k}\in\left[  N\right]  }a_{i_{1},\ldots,i_{k}}^{2}%
\leq\delta, \label{eq:sum}%
\end{equation}
where $a_{i_{1},\ldots,i_{k}}^{2}$\ is the coefficient of $x_{1,i_{1}}\ldots
x_{k,i_{k}}$. \ Let
\[
V_{i}:=\sum_{i_{2},\ldots,i_{k}\in\left[  N\right]  }a_{i_{1},i_{2}%
,\ldots,i_{k}}^{2}.
\]
We now randomly set each $x_{j,i_{j}}$ for $j\geq2$, to be $1$\ or $-1$\ with
independent probability $1/2$. \ Let
\[
X_{i}:=\sum_{i_{2},\ldots,i_{k}\in\left[  N\right]  }a_{i_{1},i_{2}%
\ldots,i_{k}}x_{2,i_{2}}\cdots x_{k,i_{k}}.
\]
Then $\operatorname*{E}[X_{i}^{2}]=V_{i}$. \ By the concavity of the square
root function, this means $\operatorname*{E}[\left\vert X_{i}\right\vert
]\geq\sqrt{V_{i}}$. \ Hence%
\[
E[\left\vert X_{1}\right\vert +\cdots+\left\vert X_{N}\right\vert ]\geq
\sqrt{V_{1}}+\cdots+\sqrt{V_{N}}.
\]
If we set $x_{1,i}=1$ whenever $X_{i}\geq0$ and $x_{1,i}=-1$ otherwise, we
get
\[
p(x_{1,1},\ldots,x_{k,N})=\sum_{i=1}^{N}x_{1,i}X_{i}=\sum_{i=1}^{N}\left\vert
X_{i}\right\vert .
\]
Since $p(x_{1,1},\ldots,x_{k,N})$\ is bounded in $\left[  -1,1\right]  $\ at
all Boolean points, this means that
\[
\sqrt{V_{1}}+\cdots+\sqrt{V_{N}}\leq1.
\]
We now perform a sequence of variable-splittings. \ For each $i\in\left[
N\right]  $, let $m_{i}:=\left\lfloor \sqrt{V_{i}}/\delta\right\rfloor $, so
that%
\[
\delta m_{i}\leq\sqrt{V_{i}}<\delta\left(  m_{i}+1\right)  .
\]
Then we split $x_{1,i}$ into $m_{i}+1$ variables. \ This replaces each term
$a_{i_{1},\ldots,i_{k}}x_{1,i_{1}}\cdots x_{k,i_{k}}$ with $m_{i}+1$ terms
that each equal $\frac{1}{m_{i}+1}a_{i_{1},\ldots,i_{k}}x_{1,i_{1}}\cdots
x_{k,i_{k}}$. \ Therefore, this variable-splitting reduces $V_{i}$\ to
$V_{i}/\left(  m_{i}+1\right)  $, and decreases the sum (\ref{eq:sum}) by
$\frac{m_{i}}{m_{i}+1}V_{i}$.

After we have performed such variable-splittings for each $i$, the sum
(\ref{eq:sum}) becomes
\begin{align*}
\sum_{i=1}^{N}\frac{V_{i}}{m_{i}+1}  &  \leq\sum_{i=1}^{N}\frac{V_{i}}%
{\sqrt{V_{i}}/\delta}\\
&  =\delta\left(  \sqrt{V_{1}}+\cdots+\sqrt{V_{N}}\right) \\
&  \leq\delta.
\end{align*}
The number of new variables that get introduced equals
\[
\sum_{i=1}^{N}m_{i}\leq\sum_{i=1}^{N}\frac{\sqrt{V_{i}}}{\delta}\leq\frac
{1}{\delta}.
\]

The case $S\subset\left[  k\right]  $ reduces to the case $S=\left[  k\right]
$ in the following way. \ For typographical convenience, assume that
$S=\left[  \ell\right]  $ for some $\ell$. \ Then substituting $x_{i,j}=1$ for
$i>\ell$ transforms the polynomial $p(x_{1,1},\ldots,x_{k,N})$ into the
polynomial
\[
p^{\prime}(x_{1,1},\ldots,x_{\ell,N})=\sum_{i_{1},\ldots,i_{\ell}\in\left[
N\right]  }\bar{a}_{i_{1},\ldots,i_{\ell}}x_{1,i_{1}}\cdots x_{\ell,i_{\ell}}%
\]
where
\[
\bar{a}_{i_{1},\ldots,i_{\ell}}:=\sum_{i_{\ell+1},\ldots,i_{k}\in\left[
N\right]  }a_{i_{1},\ldots,i_{k}}.
\]
The statement of Lemma \ref{cl:main} now becomes
\[
\sum_{i_{1},\ldots,i_{\ell}\in\left[  N\right]  }\bar{a}_{i_{1},\ldots
,i_{\ell}}^{2}\leq\delta
\]
which can be achieved similarly to the previous case.
\end{proof}

Lemma \ref{cl:main} has the following consequence.

\begin{corollary}
\label{splitcor}There is a sequence of variable-splittings that introduces at
most $2^{k}/\delta$ new variables, and that produces a polynomial $p^{\prime}%
$\ that satisfies $\Lambda_{S}\leq\delta$\ for every nonempty subset
$S\subseteq\left[  k\right]  $.
\end{corollary}

\begin{proof}
We simply apply the procedure of Lemma \ref{cl:main} once for each nonempty
$S\subseteq\left[  k\right]  $, in any order. \ Since there are $2^{k}%
-1$\ possible choices for $S$, and since each iteration adds at most
$1/\delta$\ variables, the total number of added variables is at most
$2^{k}/\delta$. \ Furthermore, we claim that later iterations can never
\textquotedblleft undo\textquotedblright\ the effects of previous iterations.
\ This is because, if we consider how the quantity%
\[
\Lambda_{S}=\sum_{(i_{j})_{j\in S}}\left(  \sum_{(i_{j})_{j\notin S}}%
a_{i_{1},\ldots,i_{k}}\right)  ^{2}%
\]
is affected by variable-splittings applied to the variables in $R_{j}$, there
are only two possibilities: if $j\in S$\ then $\Lambda_{S}$\ can decrease,
while if $j\notin S$\ then $\Lambda_{S}$\ remains unchanged.
\end{proof}

We now apply Corollary \ref{splitcor} with the choice $\delta=\varepsilon
^{2}/N$. \ This introduces at most $2^{k}N/\varepsilon^{2}=O\left(
N/\varepsilon^{2}\right)  $\ new variables, and achieves $\Lambda_{S}%
\leq\varepsilon^{2}/N$\ for every $S$.

From now on, we will use $n$ to denote the \textquotedblleft
new\textquotedblright\ number of variables per block, which is a constant
factor greater than the \textquotedblleft old\textquotedblright\ number $N$.

\subsection{The Estimator\label{ESTIMATOR}}

Let
\[
b_{i_{1},\ldots,i_{k}}:=a_{i_{1},\ldots,i_{k}}x_{1,i_{1}}\cdots x_{k,i_{k}}.
\]
Then
\[
p(x_{1,1},\ldots,x_{k,n})=\sum_{i_{1},\ldots,i_{k}}b_{i_{1},\ldots,i_{k}}.
\]
We can estimate this sum in the following way. \ For each $i,j_{i}$
independently, let $y_{i,j_{i}}$ be a $\left\{  0,1\right\}  $-valued random
variable with $\Pr[y_{i,j_{i}}=1]=\frac{1}{n^{1/k}}$. \ We then take
\[
P:=b_{i_{1},\ldots,i_{k}}y_{1,i_{1}}\cdots y_{k,i_{k}}%
\]
as our estimator.

Clearly, this is an unbiased estimator of $p(x_{1,1},\ldots,x_{k,n})$, with
expectation%
\[
\operatorname*{E}[P]=\frac{p(x_{1},\ldots,x_{n})}{n}.
\]
The result we would like to prove is that $\operatorname*{Var}[P]=O(\delta
/n)$. \ If this is true, then performing $O\left(  1\right)  $\ repetitions of
$P$ allows us to estimate $p(x_{1,1},\ldots,x_{k,n})$ with precision
$\sqrt{\delta n}=\sqrt{\left(  \varepsilon^{2}/n\right)  \cdot n}=\varepsilon
$. \ This estimation can be carried out with $O(n^{1-1/k})$ queries because,
to calculate $P$, we only need the values of $x_{i,j}$ with $y_{i,j}=1$, and
the number of such variables is $O(n^{1-1/k})$, with a very high probability.
\ Note that%
\[
O(n^{1-1/k})=O\left(  \left(  \frac{N}{\delta}\right)  ^{1-1/k}\right)
=O\left(  \left(  \frac{N}{\varepsilon^{2}}\right)  ^{1-1/k}\right)  ,
\]
in terms of the number of variables $N$ of our original polynomial. \ Here the
big-$O$ hides a factor of $\exp\left(  k\right)  $.

\subsection{Warmup\label{WARMUP}}

As a warmup, consider the following simpler estimator. \ For each
$i_{1},\ldots,i_{k}$\ independently, let $y_{i_{1},\ldots,i_{k}}$ be a
$\left\{  0,1\right\}  $-valued random variable with
\[
\Pr[y_{i_{1},\ldots,i_{k}}=1]=\frac{1}{n}.
\]
Then let
\[
P^{\prime}:=\sum_{i_{1},\ldots,i_{k}\in\left[  n\right]  }b_{i_{1}%
,\ldots,i_{k}}y_{i_{1},\ldots,i_{k}}.
\]
Once again, $P^{\prime}$\ is clearly an unbiased estimator for $p$, with
expectation $\operatorname*{E}[P^{\prime}]=p/n$.

Let us start by proving that $\operatorname*{Var}[P^{\prime}]=O(\delta/n)$.
\ Let
\[
B=\sum_{i_{1},\ldots,i_{k}\in\left[  n\right]  }b_{i_{1},\ldots,i_{k}}^{2}.
\]
Then
\begin{align*}
\operatorname*{Var}[P^{\prime}]  &  =\sum_{i_{1},\ldots,i_{k}\in\left[
n\right]  }b_{i_{1},\ldots,i_{k}}^{2}\operatorname*{Var}[y_{i_{1},\ldots
,i_{k}}]\\
&  =\sum_{i_{1},\ldots,i_{k}\in\left[  n\right]  }b_{i_{1},\ldots,i_{k}}%
^{2}\left(  \frac{1}{n}-\frac{1}{n^{2}}\right) \\
&  \leq\frac{1}{n}\sum_{i_{1},\ldots,i_{k}\in\left[  n\right]  }%
b_{i_{1},\ldots,i_{k}}^{2}\\
&  =\frac{B}{n}.
\end{align*}
Taking $S=\left[  k\right]  $ in equation (\ref{eq:need}) implies that
$B\leq\delta$ and hence $\operatorname*{Var}[P^{\prime}]\leq\delta/n$.

\subsection{Second Estimator\label{SECOND}}

The variance of the original estimator $P$ is
\begin{align*}
\operatorname*{Var}[P]  &  =\sum_{i_{1},\ldots,i_{k}\in\left[  n\right]
}b_{i_{1},\ldots,i_{k}}^{2}\operatorname*{Var}\left[  y_{1,i_{1}}\cdots
y_{k,i_{k}}\right] \\
&  +\sum_{(i_{1},\ldots,i_{k})\neq(i_{1}^{\prime},\ldots,i_{k}^{\prime}%
)}b_{i_{1},\ldots,i_{k}}b_{i_{1}^{\prime},\ldots,i_{k}^{\prime}}%
\operatorname*{Cov}\left[  y_{1,i_{1}}\cdots y_{k,i_{k}},y_{1,i_{1}^{\prime}%
}\cdots y_{k,i_{k}^{\prime}}\right] \\
&  =\operatorname*{Var}[P^{\prime}]+\sum_{(i_{1},\ldots,i_{k})\neq
(i_{1}^{\prime},\ldots,i_{k}^{\prime})}b_{i_{1},\ldots,i_{k}}b_{i_{1}^{\prime
},\ldots,i_{k}^{\prime}}\operatorname*{Cov}\left[  y_{1,i_{1}}\cdots
y_{k,i_{k}},y_{1,i_{1}^{\prime}}\cdots y_{k,i_{k}^{\prime}}\right]  .
\end{align*}
If $i_{j}\neq i_{j}^{\prime}$ for all $j$, then $\prod_{j}y_{j,i_{j}}$ and
$\prod_{j}y_{j,i_{j}^{\prime}}$ are independent random variables and the
covariance between them is zero. \ If $i_{j}=i_{j}^{\prime}$ for $\ell$ values
of $j$, then
\begin{align*}
\operatorname*{Cov}\left[  \prod_{j}y_{j,i_{j}},\prod_{j}y_{j,i_{j}^{\prime}%
}\right]   &  =\Pr\left[  \prod_{j}y_{j,i_{j}}=\prod_{j}y_{j,i_{j}^{\prime}%
}=1\right]  -\Pr\left[  \prod_{j}y_{j,i_{j}}=1\right]  \Pr\left[  \prod
_{j}y_{j,i_{j}^{\prime}}=1\right] \\
&  =\frac{1}{\left(  n^{1/k}\right)  ^{2k-\ell}}-\left(  \frac{1}{\left(
n^{1/k}\right)  ^{k}}\right)  ^{2}\\
&  =\frac{1}{n^{2-\ell/k}}-\frac{1}{n^{2}}.
\end{align*}
Let $S_{\ell}$ consist of all pairs $(i_{1},\ldots,i_{k}),(i_{1}^{\prime
},\ldots,i_{k}^{\prime})$ such that $i_{j}=i_{j}^{\prime}$ for exactly $\ell$
values of $j$. \ Let $T_{\ell}$ be the multiset consisting of the elements of
$S_{\ell},\ldots,S_{k-1}$, with each element of $S_{\ell^{\prime}}$ occurring
${\binom{\ell^{\prime}}{\ell}}$ times. \ Then by inclusion-exclusion, we have
\[
S_{\ell}=T_{\ell}-{\binom{\ell+1}{\ell}}T_{\ell+1}+{\binom{\ell+2}{\ell}%
}T_{\ell+2}-\cdots
\]
where ${\binom{\ell^{\prime}}{\ell}}T_{\ell^{\prime}}$ denotes the union of
$\ell^{\prime}$ copies of $T_{\ell^{\prime}}$. \ Hence,
\begin{align}
\operatorname*{Var}[P]  &  =\operatorname*{Var}[P^{\prime}]+\sum_{\ell
=1}^{k-1}\left(  \frac{1}{n^{2-\ell/k}}-\frac{1}{n^{2}}\right)  \sum
_{(i_{1},\ldots,i_{k}),(i_{1}^{\prime},\ldots,i_{k}^{\prime})\in S_{\ell}%
}b_{i_{1},\ldots,i_{k}}b_{i_{1}^{\prime},\ldots,i_{k}^{\prime}}\nonumber\\
&  =\operatorname*{Var}[P^{\prime}]+\sum_{\ell=1}^{k-1}p_{\ell}\sum
_{(i_{1},\ldots,i_{k}),(i_{1}^{\prime},\ldots,i_{k}^{\prime})\in T_{\ell}%
}b_{i_{1},\ldots,i_{k}}b_{i_{1}^{\prime},\ldots,i_{k}^{\prime}}
\label{eq:nearly-final}%
\end{align}
where
\[
p_{\ell}:=\sum_{j=1}^{\ell}(-1)^{\ell-j}{\binom{\ell}{j}}\left(  \frac
{1}{n^{2-j/k}}-\frac{1}{n^{2}}\right)  .
\]
For large $n$, we have $p_{\ell}=(1\pm o(1))\frac{1}{n^{2-\ell/k}}$. \ To
complete the proof, we just need one more lemma.

\begin{lemma}
\label{lem:final}
\[
\sum_{(i_{1},\ldots,i_{k}),(i_{1}^{\prime},\ldots,i_{k}^{\prime})\in T_{\ell}%
}b_{i_{1},\ldots,i_{k}}b_{i_{1}^{\prime},\ldots,i_{k}^{\prime}}\leq
\delta{\binom{k}{\ell}}.
\]

\end{lemma}

\begin{proof}
Let $S\subseteq\left[  k\right]  $ with $\left\vert S\right\vert =\ell$. \ We
define $T_{S}$ as the set consisting of all pairs $(i_{1},\ldots,i_{k}%
),(i_{1}^{\prime},\ldots,i_{k}^{\prime})$ such that $i_{j}=i_{j}^{\prime}$ for
all $j\in S$ and $(i_{1},\ldots,i_{k})\neq(i_{1}^{\prime},\ldots,i_{k}%
^{\prime})$. \ Then
\[
T_{\ell}=\sum_{S~:~\left\vert S\right\vert =\ell}T_{S},
\]
and%
\[
\sum_{\left(  (i_{1},\ldots,i_{k}),(i_{1}^{\prime},\ldots,i_{k}^{\prime
})\right)  \in T_{\ell}}b_{i_{1},\ldots,i_{k}}b_{i_{1}^{\prime},\ldots
,i_{k}^{\prime}}=\sum_{S~:~\left\vert S\right\vert =\ell}\sum_{\left(
(i_{1},\ldots,i_{k}),(i_{1}^{\prime},\ldots,i_{k}^{\prime})\right)  \in T_{S}%
}b_{i_{1},\ldots,i_{k}}b_{i_{1}^{\prime},\ldots,i_{k}^{\prime}}.
\]
The lemma now follows by showing that, for each $S$, the inner sum is at most
$\delta$. \ To show that, we first add all pairs $\left(  (i_{1},\ldots
,i_{k}),(i_{1},\ldots,i_{k})\right)  $ to $T_{S}$. \ This may only increase
the sum because $a_{i_{1},\ldots,i_{k}}^{2}$ is always at least $0$. \ Then,
we group together all terms with the same values of $i_{j}=i_{j}^{\prime}$ for
$j\in S$. \ The sum of those is equal to
\[
\sum_{(i_{j},i_{j}^{\prime})_{j\notin S}}b_{i_{1},\ldots,i_{k}}b_{i_{1}%
^{\prime},\ldots,i_{k}^{\prime}}=\left(  \sum_{(i_{j})_{j\notin S}}%
b_{i_{1},\ldots,i_{k}}\right)  ^{2}.
\]
Because of (\ref{eq:need}), the sum of all such squares, over all $i_{j},j\in
S$ is at most $\delta$.
\end{proof}

Combining Lemma \ref{lem:final} with (\ref{eq:nearly-final}), we obtain%
\begin{align*}
\operatorname*{Var}[P]  &  =\operatorname*{Var}[P^{\prime}]+\sum_{\ell
=1}^{k-1}p_{\ell}\cdot O\left(  \delta\right) \\
&  =\operatorname*{Var}[P^{\prime}]+\sum_{\ell=1}^{k-1}O\left(  \frac{\delta
}{n^{2-\ell/k}}\right) \\
&  =\operatorname*{Var}[P^{\prime}]+O\left(  \frac{\delta}{n^{1+1/k}}\right)
\\
&  =O\left(  \frac{\delta}{n}\right)  .
\end{align*}

\section{\textsf{BQP}-Completeness\label{BQP}}

In this section, we prove that (an explicit version of) the $k$-fold
\textsc{Forrelation}\ problem, with $k=\operatorname*{poly}\left(  n\right)
$, is complete for the complexity class $\mathsf{P{}romiseBQP}$. \ More
generally, for any $k$, we will show how explicit $k$-fold
\textsc{Forrelation}\ captures the power of quantum circuits of depth
$O\left(  k\right)  $.

Recall that, in explicit $k$-fold \textsc{Forrelation}, we are given as input
$k$\ Boolean circuits $C_{1},\ldots,C_{k}$, which compute the Boolean
functions $f_{1},\ldots,f_{k}:\left\{  0,1\right\}  ^{n}\rightarrow\left\{
-1,1\right\}  $ respectively. \ The problem is to decide whether the
\textquotedblleft twisted sum\textquotedblright%
\[
\Phi_{f_{1},\ldots,f_{k}}:=\frac{1}{2^{\left(  k+1\right)  n/2}}\sum
_{x_{1},\ldots,x_{k}\in\left\{  0,1\right\}  ^{n}}f_{1}\left(  x_{1}\right)
\left(  -1\right)  ^{x_{1}\cdot x_{2}}f_{2}\left(  x_{2}\right)  \left(
-1\right)  ^{x_{2}\cdot x_{3}}\cdots\left(  -1\right)  ^{x_{k-1}\cdot x_{k}%
}f_{k}\left(  x_{k}\right)
\]
satisfies $\left\vert \Phi_{f_{1},\ldots,f_{k}}\right\vert \leq\frac{1}{100}$
or $\Phi_{f_{1},\ldots,f_{k}}\geq\frac{3}{5}$, promised that one of those is
the case. \ As we observed in Proposition \ref{inpromisebqp},\ this problem is
clearly \textit{in} $\mathsf{P{}romiseBQP}$, so our task reduces to showing
that it's $\mathsf{P{}romiseBQP}$-hard---i.e., that any quantum circuit can be
encoded into it.

For this task, it will suffice to consider an extremely restricted version of
$k$-fold \textsc{Forrelation}, in which each function $f_{i}$\ depends on at
most $3$ of its input bits.

We will appeal to a well-known result of Shi \cite{shi:gate}, who showed that
the gate set $\left\{  \operatorname*{H},\operatorname*{Toffoli}\right\}  $ is
already universal for quantum computation. \ Recall here that%
\[
\operatorname*{H}=\frac{1}{\sqrt{2}}\left(
\begin{array}
[c]{cc}%
1 & 1\\
1 & -1
\end{array}
\right)
\]
is the Hadamard gate, while the Toffoli gate is the $3$-qubit gate that maps
each basis vector $\left\vert x,y,z\right\rangle $ to $\left\vert x,y,z\oplus
xy\right\rangle $. \ The Toffoli gate is equivalent, under conjugating the
third qubit by Hadamards, to the controlled-controlled-sign or CCSIGN\ gate,
which maps each $\left\vert x,y,z\right\rangle $\ to $\left(  -1\right)
^{xyz}\left\vert x,y,z\right\rangle $. \ Thus, we deduce that the set
$\left\{  \operatorname*{H},\operatorname*{CCSIGN}\right\}  $ is also
universal for quantum computation.

In a bit more detail, given a quantum circuit $Q$ composed of Hadamard\ and
CCSIGN gates, acting on $n$ qubits, define%
\[
A_{Q}:=\left\langle 0\right\vert ^{\otimes n}Q\left\vert 0\right\rangle
^{\otimes n},
\]
so that $A_{Q}^{2}$\ is the probability that $Q$ returns the all-$0$ state to
itself. \ Then let \textsc{QSim}\ be the problem of deciding whether
$\left\vert A_{Q}\right\vert \leq\frac{1}{100}$\ or $A_{Q}\geq\frac{3}{5}$,
promised that one of those is the case.

\begin{lemma}
[follows from Shi \cite{shi:gate}]\textsc{QSim} is $\mathsf{P{}romiseBQP}%
$-complete.\label{qsimlem}
\end{lemma}

\begin{proof}
Besides what was said above, together with standard amplification, we just
need two further observations. \ First, by using uncomputing, we can modify
any quantum circuit so that it \textquotedblleft accepts\textquotedblright\ by
returning all its qubits to the initial state, $\left\vert 0\right\rangle
^{\otimes n}$, and \textquotedblleft rejects\textquotedblright\ by ending in
any state orthogonal to $\left\vert 0\right\rangle ^{\otimes n}$. \ Second, we
can handle the case that $A_{Q}$\ is negative by\ running both $Q$\ and $-Q$
(i.e., $Q$ with a $-1$\ global phase), and checking whether our \textsc{QSim}%
\ oracle returns $A_{Q}\geq\frac{3}{5}$\ for either of them.
\end{proof}

Now, the outline of a reduction from \textsc{QSim}\ to $k$-fold
\textsc{Forrelation}\ suggests itself almost immediately. \ Given an $n$-qubit
quantum circuit $Q$ over the basis $\left\{  \operatorname*{H}%
,\operatorname*{CCSIGN}\right\}  $, we want to construct Boolean functions
$f_{1},\ldots,f_{k}$\ with the property that $\Phi_{f_{1},\ldots,f_{k}}=A_{Q}%
$. \ To do so, we should exploit the fact that, as we have seen, $\Phi
_{f_{1},\ldots,f_{k}}$\ is a transition amplitude for a particular kind of
quantum circuit: namely, a circuit that consists of rounds of Hadamards
applied to all $n$ qubits, interleaved with diagonal matrices $U_{f_{i}}%
$\ that map each basis state $\left\vert x\right\rangle $\ to $f_{i}\left(
x\right)  \left\vert x\right\rangle $. Thus, we should use suitably-placed
$f_{i}$'s to simulate each of the CCSIGN\ gates in $Q$ (exploiting the fact
that CCSIGN is diagonal in the computational basis), while using the $\left(
-1\right)  ^{x_{i}\cdot x_{i+1}}$\ terms in the expression for $\Phi
_{f_{1},\ldots,f_{k}}$\ to simulate the Hadamard gates in $Q$.

However, there is a technical problem in implementing the above plan.
\ Namely, while $\Phi_{f_{1},\ldots,f_{k}}$\ \textit{will} equal the
transition amplitude $\left\langle 0\right\vert ^{\otimes n}Q^{\prime
}\left\vert 0\right\rangle ^{\otimes n}$, for some quantum circuit $Q^{\prime
}$ that consists of Hadamard and CCSIGN gates, the circuit $Q^{\prime}$\ will
contain $n$ Hadamard gates between every CCSIGN gate and the succeeding one,
\textit{whether we want Hadamards there or not}. \ This suggests that, in
order to encode an \textit{arbitrary} sequence of Hadamard and
$\operatorname*{CCSIGN}$ gates, we need some gadget that \textquotedblleft
cancels\textquotedblright\ unwanted Hadamard gates against each other, leaving
only the Hadamard gates that actually appear in the original circuit $Q$. \ Of
course, we can exploit the fact that $\operatorname*{H}^{2}$\ is the identity.
\ So for example, if we wanted to remove the $n$ Hadamard gates that
\textquotedblleft automatically appear\textquotedblright\ between $U_{f_{i-1}%
}$\ and $U_{f_{i}}$, then we could simply set $f_{i}$\ to be the constant $1$
function, so that $U_{f_{i}}$\ was the identity. \ Then every
$\operatorname*{H}$ between $U_{f_{i-1}}$\ and $U_{f_{i}}$\ would cancel with
a corresponding $\operatorname*{H}$ between $U_{f_{i}}$\ and $U_{f_{i+1}}$.
\ Alas, this still doesn't tell us how to cancel \textit{some}
$\operatorname*{H}$'s: that is, how to Hadamard certain desired qubits, but
not other qubits. \ We do this in the following theorem.

\begin{theorem}
\label{bqphard}Explicit $k$-fold \textsc{Forrelation}, for
$k=\operatorname*{poly}\left(  n\right)  $, is $\mathsf{P{}romiseBQP}$-hard.
\ (Moreover, the functions $f_{1},\ldots,f_{k}$\ produced by the reduction all
have the form $f_{i}\left(  x\right)  =\left(  -1\right)  ^{C\left(  x\right)
}$, where $C$ is a product of at most $3$ input bits.)
\end{theorem}

\begin{proof}
Given what we said above, the only additional ingredient we need is a gadget
that lets us Hadamard some desired \textit{subset} of the qubits,
$S\subset\left[  n\right]  $, and not the qubits outside $S$.%
\begin{figure}[ptb]%
\centering
\includegraphics[
trim=1.820609in 5.468394in 1.523933in 0.608960in,
height=0.7169in,
width=2.4483in
]%
{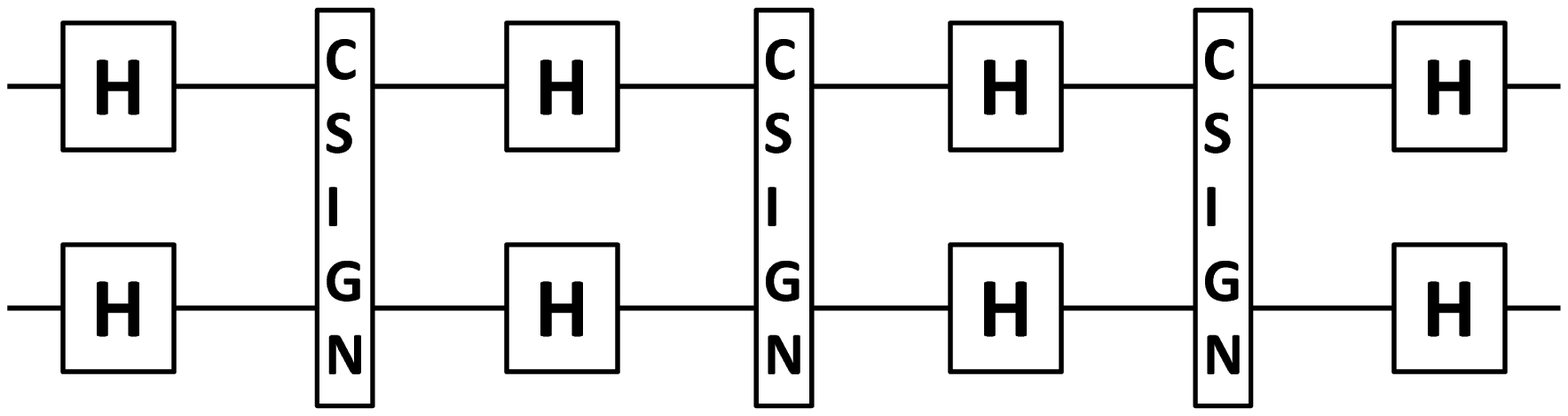}%
\caption{A $2$-qubit gadget for converting an even number of layers of
Hadamard gates into an odd number.}%
\label{gadget}%
\end{figure}

For simplicity, suppose that $\left\vert S\right\vert =2$, and let $a$ and $b$
be $S$'s elements. \ Our gadget, shown in Figure \ref{gadget}, consists of
three CSIGN gates (i.e., gates that map $\left\vert x,y\right\rangle $\ to
$\left(  -1\right)  ^{xy}\left\vert x,y\right\rangle $) on $a$ and $b$,
sandwiched between Hadamard gates. \ Note that we can implement a CSIGN on $a$
and $b$ as $U_{f_{i}}$,\ where $f_{i}\left(  z_{1},\ldots,z_{n}\right)
:=\left(  -1\right)  ^{z_{a}z_{b}}$. \ Meanwhile, the Hadamard gates are just
those that are automatically applied between each $U_{f_{i}}$\ and
$U_{f_{i+1}}$ in a quantum circuit for \textsc{Forrelation}.

To see why the gadget works, consider the following identity:%
\[
\left(  \frac{1}{2}\left(
\begin{array}
[c]{cccc}%
1 & 1 & 1 & 1\\
1 & -1 & 1 & -1\\
1 & 1 & -1 & -1\\
1 & -1 & -1 & 1
\end{array}
\right)  \left(
\begin{array}
[c]{cccc}%
1 & 0 & 0 & 0\\
0 & 1 & 0 & 0\\
0 & 0 & 1 & 0\\
0 & 0 & 0 & -1
\end{array}
\right)  \right)  ^{3}=\left(
\begin{array}
[c]{cccc}%
1 & 0 & 0 & 0\\
0 & 0 & 1 & 0\\
0 & 1 & 0 & 0\\
0 & 0 & 0 & 1
\end{array}
\right)  .
\]
In particular, if we let $\operatorname*{C}$\ stand for CSIGN,
$\operatorname*{H}^{\otimes2}$ for Hadamards on two qubits, and
$\operatorname*{S}$\ for the $2$-qubit SWAP gate, then%
\[
\operatorname*{H}\nolimits^{\otimes2}\operatorname*{C}\operatorname*{H}%
\nolimits^{\otimes2}\operatorname*{C}\operatorname*{H}\nolimits^{\otimes
2}\operatorname*{C}\operatorname*{H}\nolimits^{\otimes2}=\operatorname*{S}%
\operatorname*{H}\nolimits^{\otimes2}.
\]
Contrast this with what happens if we apply the $2$-qubit identity,
$\operatorname*{I}$, rather than $\operatorname*{C}$, in the inner layers:%
\[
\operatorname*{H}\nolimits^{\otimes2}\operatorname*{I}\operatorname*{H}%
\nolimits^{\otimes2}\operatorname*{I}\operatorname*{H}\nolimits^{\otimes
2}\operatorname*{I}\operatorname*{H}\nolimits^{\otimes2}=\operatorname*{I}.
\]
Thus, Hadamards get applied if $\operatorname*{C}$\ is chosen for the inner
layers, but \textit{not} if $\operatorname*{I}$\ is chosen. \ So this gadget
has the effect of Hadamarding $a$ and $b$, while not Hadamarding the other
qubits in the circuit. \ Now, the gadget also has the unintended side effect
of swapping $a$ and $b$. \ But since we know this is going to happen, we can
keep track of it by simply swapping the \textit{labels} of $a$ and $b$
whenever the gadget is applied.

To generalize to arbitrary subsets $S\subset\left[  n\right]  $: if
$\left\vert S\right\vert >2$ is even, then we simply partition $S$ into pairs,
and apply the $2$-qubit Hadamard gadget once in succession to each pair. \ If
$\left\vert S\right\vert $\ is odd, then the odd qubit in $S$ can be paired
with a \textquotedblleft dummy qubit,\textquotedblright\ which is introduced
into the circuit for this sole purpose.

Notice that each CSIGN gate is simulated by a single $f_{i}$, while each pair
of Hadamards is simulated by three $f_{i}$'s together with the Hadamards that
sandwich them. \ Thus, we can place a pair of Hadamards after a CSIGN gate, or
vice versa, with no difficulty. \ To place one CSIGN gate after another, or
one pair of Hadamards after another, we insert an $f_{i}=1$\ (i.e., a constant
$1$ function) in between them, in order to cancel the unwanted Hadamards.

Given a quantum circuit $Q$ on $n$ qubits, consisting of $m$ Hadamard and
CCSIGN gates, the end result of our reduction will be a list of Boolean
functions $f_{1},\ldots,f_{k}:\left\{  0,1\right\}  ^{n+1}\rightarrow\left\{
-1,1\right\}  $,\ with $k=O\left(  m\right)  $, such that $\Phi_{f_{1}%
,\ldots,f_{k}}=A_{Q}$. \ (The $n+1$\ comes from the addition of the dummy
qubit.) \ Furthermore, each $f_{i}\left(  z_{1},\ldots,z_{n}\right)  $\ in the
list will have the form $1$\ or $\left(  -1\right)  ^{z_{a}z_{b}}$ or $\left(
-1\right)  ^{z_{a}z_{b}z_{c}}$, so will be easy to specify using a Boolean circuit.
\end{proof}

As a side note, suppose we required the functions $f_{1},\ldots,f_{k}$\ to
depend on at most $2$ input bits, rather than $3$ bits. \ In that case, we
claim that $k$-fold \textsc{Forrelation}\ would be in $\mathsf{P}$. \ The
reason is just that in this case, our quantum circuit for \textsc{Forrelation}%
\ would be a\ stabilizer circuit, so the Gottesman-Knill Theorem would
apply.\footnote{Indeed, by a result of Aaronson and Gottesman \cite{ag},
$k$-fold \textsc{Forrelation}\ with this restriction is $\mathsf{\oplus L}%
$-complete.}

Examining the proof of Theorem \ref{bqphard}, we can derive a stronger
consequence. \ Define a \textit{depth-}$d$\textit{ quantum circuit} as one
where the gates are organized into $d$ sequential layers, with the gates
within each layer all commuting with one another.\footnote{Often, one further
requires that the gates within each layer act on disjoint sets of qubits.
\ But it will be convenient for us to drop that requirement.} \ Now, given a
depth-$d$ quantum circuit $Q$\ over the basis $\left\{  \operatorname*{H}%
,\operatorname*{CCSIGN}\right\}  $, let \textsc{QSim}$_{d}$\ be the problem of
deciding whether $A_{Q}:=\left\langle 0\right\vert ^{\otimes n}Q\left\vert
0\right\rangle ^{\otimes n}$\ satisfies $A_{Q}\geq\frac{1}{4}$ or $\left\vert
A_{Q}\right\vert \leq\frac{1}{100}$, promised that one of those is the case.
\ Then we have the following:

\begin{theorem}
\textsc{QSim}$_{d}$ is polynomial-time reducible to explicit $\left(
2d+1\right)  $-fold \textsc{Forrelation}. \ (Moreover, the functions
$f_{1},\ldots,f_{2d+1}$\ produced by the reduction all have the form
$f_{i}\left(  x\right)  =\left(  -1\right)  ^{p\left(  x\right)  }$, where
$p$\ is a degree-$3$ polynomial in the input bits.)\label{qsimd}
\end{theorem}

\begin{proof}
The only change we need to make to the proof of Theorem \ref{bqphard}\ is to
be a bit more frugal with $f_{i}$'s---using at most two $f_{i}$'s for each
layer of $Q$, rather than separate $f_{i}$'s for each gate.

In more detail, a given layer $L$ of $Q$ consists of Hadamard gates\ on some
subset\ of qubits $S\subseteq\left[  n\right]  $, as well as CCSIGN gates
(which might overlap each other) on some other subset of qubits $T\subseteq
\left[  n\right]  $\ satisfying $S\cap T=\varnothing$. \ Suppose we want to
simulate $L$ using the three functions $f_{i},f_{i+1},f_{i+2}$, together with
the Hadamards that sandwich them. \ Then we build up the functions as follows:
initially $f_{i}=f_{i+1}=f_{i+2}=1$. \ For each CCSIGN gate, acting on qubits
$a,b,c\in T$, we multiply $f_{i+1}$\ by $\left(  -1\right)  ^{z_{a}z_{b}z_{c}%
}$, leaving $f_{i}$\ and $f_{i+2}$\ unchanged. \ For each pair of
Hadamard\ gates,\ acting on qubits $a,b\in S$, we multiply $f_{i}$,\ $f_{i+1}%
$, and $f_{i+2}$ by $\left(  -1\right)  ^{z_{a}z_{b}}$. \ One can check that
the end result is%
\[
\operatorname*{H}\nolimits^{\otimes n}U_{f_{i+2}}\operatorname*{H}%
\nolimits^{\otimes n}U_{f_{i+1}}\operatorname*{H}\nolimits^{\otimes n}%
U_{f_{i}}\operatorname*{H}\nolimits^{\otimes n}=\sigma L,
\]
where $\sigma$\ represents a SWAP gate applied to each pair $a,b\in
S$\ (something that, as before, we can easily keep track of).

To separate two successive layers of the circuit, we could simply insert a
constant function, $f_{i}=1$. \ This would yield a $4d$-fold
\textsc{Forrelation}\ instance, $d$ being the number of layers. \ If we want
to decrease the number of $f_{i}$'s\ from $4d$\ to $2d+1$, then we can
eliminate each constant $f_{i}$, together with the Hadamard layers surrounding
it (which simply cancel each other out), and then merge $f_{i-1}$\ and
$f_{i+1}$\ into a single $f_{i}$\ by multiplying them: $f_{i}\left(  x\right)
=f_{i-1}\left(  x\right)  f_{i+1}\left(  x\right)  $, or $f_{i}\left(
x\right)  =\left(  -1\right)  ^{p\left(  x\right)  +q\left(  x\right)  }$\ if
$f_{i-1}\left(  x\right)  =\left(  -1\right)  ^{p\left(  x\right)  }$\ and
$f_{i+1}\left(  x\right)  =\left(  -1\right)  ^{q\left(  x\right)  }$.

Note that, at the end, each $f_{i}$\ is a degree-$3$ polynomial in its input
bits, and we again have $\Phi_{f_{1},\ldots,f_{k}}=A_{Q}$.
\end{proof}

So for example, we find that explicit $\log n$-fold \textsc{Forrelation}\ is a
complete promise problem for $\mathsf{P{}romiseBQNC}^{1}$: the class of
problems that captures what can be done using log-depth quantum circuits (and
which already contains \textsc{Factoring}, by a result of Cleve and Watrous
\cite{cw}).

\section{Appendix: Separations for Sampling and Relation
Problems\label{SAMPREL}}

Let \textsc{Fourier Sampling}\ be the following problem. \ Given oracle access
to a Boolean function $f:\left\{  0,1\right\}  ^{n}\rightarrow\left\{
-1,1\right\}  $, the task is to sample from a distribution $\emph{D}$\ over
$\left\{  0,1\right\}  ^{n}$\ such that $\left\Vert \emph{D}-\emph{D}%
_{f}\right\Vert \leq\varepsilon$, where $\emph{D}_{f}$\ is the distribution
defined by%
\[
\Pr_{\emph{D}_{f}}\left[  y\right]  =\hat{f}\left(  y\right)  ^{2}=\left(
\frac{1}{2^{n}}\sum_{x\in\left\{  0,1\right\}  ^{n}}f\left(  x\right)  \left(
-1\right)  ^{x\cdot y}\right)  ^{2}.
\]
It is clear that \textsc{Fourier Sampling}\ is solvable---indeed, with
$\varepsilon=0$---by a quantum algorithm that makes just a single query to
$f$. \ The algorithm consists of a round of Hadamard gates, then a query to
$f$, then another round of Hadamard gates, then a measurement in the
computational basis.

By contrast, we show in this appendix that any classical randomized algorithm
for \textsc{Fourier Sampling} requires $\Omega\left(  N/\log N\right)
$\ queries, where $N=2^{n}$ is the size of $f$'s truth table. \ In other
words, a much larger quantum versus classical separation can be achieved for
sampling problems than for decision problems.

\begin{theorem}
\label{samplb}Fix (say) $\varepsilon=0.01$. \ Then the randomized query
complexity of \textsc{Fourier Sampling}\ is $\Omega\left(  N/\log N\right)  $.
\end{theorem}

\begin{proof}
Let $A$\ be a classical algorithm, and let $\mathcal{E}_{f}=\left\{
p_{f,y}\right\}  _{y\in\left\{  0,1\right\}  ^{n}}$\ be the probability
distribution output by $A$ when given $f$ as an oracle. \ The success
condition is that, for all $f$,%
\[
\frac{1}{2}\sum_{y\in\left\{  0,1\right\}  ^{n}}\left\vert p_{f,y}-\hat
{f}\left(  y\right)  ^{2}\right\vert \leq\varepsilon.
\]
By an averaging argument, this implies that there exists a $y^{\ast}%
\in\left\{  0,1\right\}  ^{n}$\ such that%
\[
\operatorname*{E}_{f}\left[  \left\vert p_{f,y^{\ast}}-\hat{f}\left(  y^{\ast
}\right)  ^{2}\right\vert \right]  \leq\frac{2\varepsilon}{N}.
\]
So by Markov's inequality,%
\[
\left\vert p_{f,y^{\ast}}-\hat{f}\left(  y^{\ast}\right)  ^{2}\right\vert
\leq\frac{20\varepsilon}{N}%
\]
for at least a\ $9/10$\ fraction of $f$'s. \ Now assume by symmetry, and
without loss of generality, that $y^{\ast}=0^{n}$. \ Let $z_{i}:=\frac
{1+f\left(  x_{i}\right)  }{2}$ (where $x_{1},\ldots,x_{N}$ is a lexicographic
ordering of inputs), let $Z:=\left(  z_{1},\ldots,z_{N}\right)  $, and let%
\[
\left\vert Z\right\vert :=z_{1}+\cdots+z_{N}.
\]
Then $\hat{f}\left(  0^{n}\right)  =\left(  2\left\vert Z\right\vert
-N\right)  /N$. \ The question before us is how many $z_{i}$'s\ the algorithm
$A$ needs to query, in order to output $0^{n}$\ (or, as we'll say,
\textquotedblleft accept\textquotedblright) with a probability $p_{Z}%
:=p_{f,0^{n}}$\ that satisfies
\begin{equation}
\left\vert p_{Z}-\left(  \frac{2\left\vert Z\right\vert }{N}-1\right)
^{2}\right\vert \leq\frac{20\varepsilon}{N} \label{succ}%
\end{equation}
with probability at least $9/10$\ over $Z\in\left\{  0,1\right\}  ^{N}$.

Observe that, without loss of generality, $A$ just nonadaptively queries $t$
randomly-chosen inputs $z_{i_{1}},\ldots,z_{i_{t}}$, and then accepts with a
probability $q_{k}$\ that depends solely on $k:=z_{i_{1}}+\cdots+z_{i_{t}}$.
\ For, if $A$ did anything other than this, then by averaging over all
$N!$\ possible permutations of $Z$, we would obtain an algorithm of this
restricted form that made the same number of queries and that was just as
likely to satisfy (\ref{succ}). \ In particular, this means that the
probability $p_{\left\vert Z\right\vert }=p_{Z}$\ that $A$ accepts $Z$ depends
only on $\left\vert Z\right\vert $. \ Explicitly,%
\begin{equation}
p_{w}=\sum_{k=0}^{t}q_{k}r_{k,w}, \label{convcomb}%
\end{equation}
where%
\[
r_{k,w}=\frac{1}{2^{N}}\binom{t}{k}\binom{N-t}{w-k}%
\]
is the probability that $z_{i_{1}}+\cdots+z_{i_{t}}=k$\ conditioned on
$\left\vert Z\right\vert =w$.

Let%
\begin{align*}
U  &  :=\left\{  Z:\left\vert \left\vert Z\right\vert -\frac{N}{2}\right\vert
\leq\frac{\sqrt{N}}{4}\right\}  ,\\
V  &  :=\left\{  Z:\left\vert \left\vert Z\right\vert -\frac{N}{2}\right\vert
\in\left[  \frac{\sqrt{N}}{2},2\sqrt{N}\right]  \right\}  .
\end{align*}
Then note that for sufficiently large $N$,%
\begin{align*}
\Pr_{Z}\left[  Z\in U\right]   &  \geq\operatorname{erf}\left(  \frac
{1}{2\sqrt{2}}\right)  -o\left(  1\right)  >0.38,\\
\Pr_{Z}\left[  Z\in V\right]   &  \geq\operatorname{erf}\left(  2\sqrt
{2}\right)  -\operatorname{erf}\left(  \frac{1}{\sqrt{2}}\right)  -o\left(
1\right)  >0.31.
\end{align*}
This implies that, conditioned on $Z\in U$, we must have%
\begin{align*}
p_{Z}  &  \leq\left(  \frac{2\left\vert Z\right\vert }{N}-1\right)  ^{2}%
+\frac{20\varepsilon}{N}\\
&  \leq\left(  \frac{2\sqrt{N}/4}{N}\right)  ^{2}+\frac{20\varepsilon}{N}\\
&  =\frac{0.25+20\varepsilon}{N}\\
&  =\frac{0.45}{N}%
\end{align*}
with probability at least $1-\frac{0.1}{0.38}>2/3$\ over $Z$. \ Likewise,
conditioned on $Z\in V$, we must have%
\begin{align*}
p_{Z}  &  \geq\left(  \frac{2\left\vert Z\right\vert }{N}-1\right)  ^{2}%
-\frac{20\varepsilon}{N}\\
&  \geq\frac{1-20\varepsilon}{N}\\
&  =\frac{0.8}{N}%
\end{align*}
with probability at least $1-\frac{0.1}{0.31}>2/3$\ over $Z$.

Thus, it suffices to prove a lower bound for the following restricted problem:
for at least $2/3$ of strings $Z\in U$, accept with some probability
$p_{Z}\leq0.45/N$, while for at least $2/3$ of strings $Z\in V$, accept with
some probability $p_{Z}\geq0.8/N$. \ Indeed, let us assume without loss of
generality that if $Z$ is drawn from $U$ then $\left\vert Z\right\vert
=N/2$\ exactly, while if $Z$ is drawn from $V$ then $\left\vert Z\right\vert
=N/2+2\sqrt{N}$ exactly. \ This can only make the distinguishing task easier,
and therefore the lower bound stronger. \ Note that, because $p_{\left\vert
Z\right\vert }=p_{Z}$\ depends only on $\left\vert Z\right\vert $, any
algorithm that achieves $p_{Z}\leq0.45/N$\ for at least a $2/3$\ fraction of
$\left\vert Z\right\vert =N/2$\ actually achieves that for \textit{all}
$\left\vert Z\right\vert =N/2$, while any algorithm that achieves $p_{Z}%
\geq0.8/N$\ for at least a $2/3$\ fraction of $\left\vert Z\right\vert
=N/2+2\sqrt{N}$\ achieves that for all $\left\vert Z\right\vert =N/2+2\sqrt
{N}$.

Recall from (\ref{convcomb}) that $p_{\left\vert Z\right\vert }=p_{Z}$\ is a
linear combination of $r_{k,\left\vert Z\right\vert }$'s, which are the
probabilities for various numbers $k$\ of `1' bits to be observed among the
$t$\ bits queried, conditioned on $\left\vert Z\right\vert $ having the value
that it does. \ Moreover, the coefficients $q_{k}$\ in this linear combination
are all in $\left[  0,1\right]  $. \ We want to show that, if $t=o\left(
N/\log N\right)  $, then either $p_{N/2+2\sqrt{N}}=o\left(  1/N\right)  $\ or
else $p_{N/2+2\sqrt{N}}-p_{N/2}=o\left(  1/N\right)  $---either of which
suffices to show $A$'s failure.

We deduce this from two probabilistic claims. \ First, by a Chernoff bound,%
\begin{align*}
\sum_{k=t/2+c\sqrt{t}}^{t}r_{k,N/2+2\sqrt{N}}  &  =\Pr\left[  z_{i_{1}}%
+\cdots+z_{i_{t}}\geq\frac{t}{2}+c\sqrt{t}:\left\vert Z\right\vert =\frac
{N}{2}+2\sqrt{N}\right] \\
&  \leq\exp\left\{  -\frac{1}{3}\frac{\left(  \frac{t}{2}+c\sqrt{t}-\left(
\frac{t}{2}+\frac{2t}{\sqrt{N}}\right)  \right)  ^{2}}{\frac{t}{2}+\frac
{2t}{\sqrt{N}}}\right\} \\
&  \leq\exp\left\{  -\frac{1}{3t}\left(  c^{2}t-\frac{4ct^{3/2}}{\sqrt{N}%
}+\frac{4t^{2}}{N}\right)  \right\} \\
&  \leq\exp\left\{  -\frac{c^{2}}{3}+\frac{4c}{3}\sqrt{\frac{t}{N}}\right\} \\
&  \leq\exp\left\{  -\Omega\left(  c^{2}\right)  \right\}
\end{align*}
for large $c$. \ So in particular, if $c=\omega(\sqrt{\log N})$, then all the
events involving observing $k$\ `1' bits, for $k\geq t/2+c\sqrt{t}$, have
total probability $o\left(  1/N\right)  $. \ This means that, if $A$ worked,
then we could set $q_{k}=0$\ for all $k\geq t/2+c\sqrt{t}$\ without affecting
$A$'s success: we would still have $p_{N/2+2\sqrt{N}}-p_{N/2}=\Omega\left(
1/N\right)  $ and $p_{N/2}=O(1/N)$.

Thus, let us concentrate next on $r_{k,N/2+2\sqrt{N}}$\ and $r_{k,N/2}$\ for
$k\leq t/2+O(\sqrt{t\log N})$. \ Here, we look at their ratio:%
\begin{align*}
\frac{r_{k,N/2+2\sqrt{N}}}{r_{k,N/2}}  &  =\frac{\binom{N-t}{N/2+2\sqrt{N}-k}%
}{\binom{N-t}{N/2-k}}\\
&  =\frac{\left(  N/2-t+k-2\sqrt{N}+1\right)  \cdots\left(  N/2-t+k\right)
}{\left(  N/2-k+1\right)  \cdots\left(  N/2-k+2\sqrt{N}\right)  }\\
&  \leq\left(  1+\frac{2k-t-2\sqrt{N}}{N/2-k+1}\right)  ^{2\sqrt{N}}\\
&  \leq\left(  1+\frac{2k-t}{N/4}\right)  ^{2\sqrt{N}}\\
&  \leq\left(  1+O\left(  \frac{\sqrt{t\log N}}{N}\right)  \right)
^{2\sqrt{N}}\\
&  =\exp\left\{  O\left(  \sqrt{\frac{t\log N}{N}}\right)  \right\}  .
\end{align*}
Notice that, if $t=o\left(  N/\log N\right)  $, then the above ratio is
$1+o\left(  1\right)  $. \ This means that taking a nonnegative linear
combination of $r_{k,\left\vert Z\right\vert }$'s cannot possibly suffice to
achieve $p_{N/2}=O(1/N)$\ and $p_{N/2+2\sqrt{N}}-p_{N/2}=\Omega\left(
1/N\right)  $\ at the same time.
\end{proof}

We conjecture that Theorem \ref{samplb}\ is tight: that is, that there exists
a randomized algorithm for \textsc{Fourier Sampling}\ making $O\left(  N/\log
N\right)  $\ queries. \ More generally, we conjecture that \textit{any}
approximate sampling problem solvable with $1$ quantum query is also solvable
with $O\left(  N/\log N\right)  $\ classical randomized queries. \ Still more
generally, we conjecture that any approximate sampling problem solvable with
$k=O\left(  1\right)  $\ quantum queries is also solvable with $O(N/\left(
\log N\right)  ^{1/k})$\ classical randomized queries; and that this is tight,
being achieved by a $k$-fold generalization of \textsc{Fourier Sampling}. \ In
the $k$-fold generalization, we are given oracle access to $k$ Boolean
functions $f_{1},\ldots,f_{k}:\left\{  0,1\right\}  ^{n}\rightarrow\left\{
-1,1\right\}  $. \ The task is to sample from a distribution $\emph{D}$\ over
$\left\{  0,1\right\}  ^{n}$\ such that $\left\Vert \emph{D}-\emph{D}%
_{f_{1},\ldots,f_{k}}\right\Vert \leq\varepsilon$, where $\emph{D}%
_{f_{1},\ldots,f_{k}}$\ is the distribution defined by%
\[
\Pr_{\emph{D}_{f_{1},\ldots,f_{k}}}\left[  y\right]  =\left(  \frac
{1}{2^{n\left(  k+1\right)  /2}}\sum_{x_{1},\ldots,x_{k}\in\left\{
0,1\right\}  ^{n}}f_{1}\left(  x_{1}\right)  \left(  -1\right)  ^{x_{1}\cdot
x_{2}}f_{2}\left(  x_{2}\right)  \left(  -1\right)  ^{x_{2}\cdot x_{3}}\cdots
f_{k}\left(  x_{k}\right)  \left(  -1\right)  ^{x_{k}\cdot y}\right)  ^{2}.
\]

So far, we have discussed separations for approximate sampling problems. \ But
it is also possible to modify \textsc{Fourier Sampling}\ to produce a
\textit{relation} problem---that is, a problem of outputting any element of a
set $S$ of \textquotedblleft valid solutions\textquotedblright---with a large
quantum/classical separation. \ One way to do this would be to use the
construction of Aaronson \cite{aar:samp}, which, given any approximate
sampling problem, uses Kolmogorov complexity to produce a relation problem of
roughly equivalent difficulty. \ Unfortunately, that construction will blow up
the quantum query complexity from $1$\ to $O\left(  \log N\right)  $,
weakening the result. \ A more direct approach would be to consider the
following relation problem: given oracle access to a Boolean function
$f:\left\{  0,1\right\}  ^{n}\rightarrow\left\{  -1,1\right\}  $, output any
string $y\in\left\{  0,1\right\}  ^{n}$\ such that $\left\vert \hat{f}\left(
y\right)  \right\vert \geq c$. \ If we use the obvious Fourier sampling
algorithm, this problem is solvable with $1$ quantum query, with success
probability asymptotically equal to%
\[
\frac{2}{\sqrt{2\pi}}\int_{c}^{\infty}e^{-x^{2}/2}x^{2}dx.
\]
On the other hand, it is a plausible conjecture that any classical randomized
algorithm that makes $o\left(  N/\log N\right)  $\ queries to $f$, can solve
the relation problem with probability at most about%
\[
\frac{2}{\sqrt{2\pi}}\int_{c}^{\infty}e^{-x^{2}/2}dx.
\]
If (say) $c=1$, this would give us an $0.8$\ versus $0.317$\ gap in success
probabilities. \ That gap could be boosted further using amplification (which,
however, would increase the quantum query complexity, from $1$ to some larger constant).

\section{Appendix: Estimator for Arbitrary Bounded Quadratics\label{POLY}}

Assume that we have an arbitrary degree-$k$ polynomial
\[
p\left(  x_{1},\ldots,x_{N}\right)  =\sum_{I\subseteq\left[  N\right]
:\left\vert I\right\vert \leq k}a_{I}\prod_{i\in I}x_{i},
\]
with $p\left(  x\right)  \in\lbrack-1,1]$ whenever $x\in\left\{  -1,1\right\}
^{N}$. \ We would like to show that $p\left(  x\right)  $ can be estimated by
a randomized algorithm that makes $O\left(  N^{1-1/k}\right)  $ queries, using
a sampling procedure similar to what we used in Section \ref{SIM} for the
special case of block-multilinear polynomials. \ For our previous proof to
work, we need there to exist a sequence of variable-splittings that introduces
$O(N)$ new variables, and that transforms $p(x_{1},\ldots,x_{N})$ into a
polynomial%
\[
q(x_{1},\ldots,x_{M})=\sum_{I\subseteq\lbrack M]:|I|\leq k}b_{I}\prod_{i\in
I}x_{i}%
\]
that satisfies the following two requirements:

\begin{enumerate}
\item[(i)] $\sum_{I}b_{I}^{2}=O(\frac{1}{N})$;

\item[(ii)] for all $l\in\left[  k-1\right]  $, we have
\begin{equation}
\sum_{I,J:,I\neq J,|I\cap J|=l}b_{I}b_{J}=O\left(  \frac{1}{N^{l/k}}\right)  .
\label{eq:squares1}%
\end{equation}

\end{enumerate}

Requirement (i) is for the bound on the variance of the \textquotedblleft
warmup estimator,\textquotedblright\ in Section \ref{WARMUP}. \ Requirement
(ii) is for the \textquotedblleft real estimator,\textquotedblright\ in
Section \ref{SECOND}. \ Below, we will be able to prove requirement (i) for
any $k$, and requirement (ii) in the special case $k=2$.

\subsection{Fourier Basics\label{FBASICS}}

Given a real polynomial $p:\left\{  -1,1\right\}  ^{N}\rightarrow\mathbb{R}$,
we consider the following notions:%
\begin{align*}
\operatorname{Var}\left[  p\right]   &  :=\operatorname{E}\left[  \left(
p\left(  x\right)  -\operatorname{E}\left[  p\left(  x\right)  \right]
\right)  ^{2}\right]  ,\\
\operatorname{Inf}_{i}\left[  p\right]   &  :=\operatorname{E}\left[  \left(
p\left(  x^{i}\right)  -p\left(  x\right)  \right)  ^{2}\right]  ,\\
\left\Vert p\right\Vert _{\infty}  &  :=\max_{x\in\left\{  -1,1\right\}  ^{N}%
}\left\vert p\left(  x\right)  \right\vert
\end{align*}
(where $x^{i}$\ means $x$\ with the $i^{th}$\ bit flipped). \ Also, let
$\hat{p}\left(  S\right)  $ be the Fourier coefficient corresponding to the
subset $S\subseteq\left[  N\right]  $---or equivalently, the coefficient in
$p$\ of the monomial $\prod_{i\in S}x_{i}$.

Note that, since $\sum_{I}b_{I}^{2}=\operatorname{Var}[q]$, requirement (i) is
equivalent to $\operatorname{Var}[q]=O(\frac{1}{N})$.

From elementary Fourier analysis, we have the following useful lemma.

\begin{lemma}
\label{infbound}If $p:\left\{  -1,1\right\}  ^{N}\rightarrow\mathbb{R}$ is a
real polynomial of degree $k$, then%
\[
\sum_{i\in\left[  N\right]  }\operatorname{Inf}_{i}\left[  p\right]  \leq
k\operatorname{Var}\left[  p\right]  .
\]

\end{lemma}

\begin{proof}
We have%
\[
\operatorname{Inf}_{i}\left[  p\right]  =\sum_{S\ni i}\hat{p}\left(  S\right)
^{2},
\]
and hence%
\[
\sum_{i\in\left[  N\right]  }\operatorname{Inf}_{i}\left[  p\right]
=\sum_{\left\vert S\right\vert \leq k}\left\vert S\right\vert \hat{p}\left(
S\right)  ^{2}\leq k\sum_{\left\vert S\right\vert \leq k}\hat{p}\left(
S\right)  ^{2}=k\operatorname{Var}\left[  p\right]  .
\]

\end{proof}

\subsection{Requirement (i)\label{REQI}}

Our goal is to find variables $x_{i}$\ in $p$ with large influences (that is,
large values of $\operatorname{Inf}_{i}\left[  p\right]  $). \ To do so, we
will use the following result of Dinur et al.\ \cite[Theorem 3]{dfko}.

\begin{theorem}
[\cite{dfko}, Theorem 3]\label{dfkothm}There exists a constant $C$ for which
the following holds. \ Suppose $p:\left\{  -1,1\right\}  ^{N}\rightarrow
\mathbb{R}$ is a real polynomial of degree $k$, which satisfies
$\operatorname{Var}[p]=1$ and $\operatorname{Inf}_{i}\left[  p\right]  \leq
t^{2}C^{-k}$ for all $i\in\left[  N\right]  $. \ Then%
\[
\Pr_{x\in\left\{  -1,1\right\}  ^{N}}\left[  \left\vert p\left(  x\right)
\right\vert \geq t\right]  \geq\exp(-Ct^{2}k^{2}\log k).
\]

\end{theorem}

Theorem \ref{dfkothm} means, in particular, that if $\operatorname{Var}\left[
p\right]  =1$ and $\operatorname{Inf}_{i}\left[  p\right]  \leq t^{-2}C^{-k}$
for all $i\in\left[  N\right]  $, then $\left\Vert p\right\Vert _{\infty}\geq
t$: in other words, there \textit{exists} an $x\in\left\{  -1,1\right\}  ^{N}%
$\ such that $\left\vert p\left(  x\right)  \right\vert \geq t$. \ By
rescaling, we can turn this into the following.

\begin{corollary}
\label{cor:dfko}There exists a constant $C$ for which the following holds.
\ Suppose $p:\left\{  -1,1\right\}  ^{N}\rightarrow\mathbb{R}$ is a real
polynomial of degree $k$, which satisfies $\left\Vert p\right\Vert _{\infty
}\leq1$ and $\operatorname{Var}\left[  p\right]  \geq\frac{C}{N}$. \ Then
there exists an $i\in\left[  N\right]  $\ such that $\operatorname{Inf}%
_{i}\left[  p\right]  \geq\frac{1}{C^{k}}\operatorname{Var}\left[  p\right]
^{2}$.
\end{corollary}

\begin{proof}
We can reword Theorem \ref{dfkothm} as follows: if $\left\Vert p\right\Vert
_{\infty}\leq t$ and $\operatorname{Var}\left[  p\right]  =1$, then there
exists an $i\in\left[  N\right]  $\ such that $\operatorname{Inf}_{i}%
[p]\geq\frac{t^{2}}{C^{k}}$. \ Now suppose $\left\Vert p\right\Vert _{\infty
}\leq1$. \ Then define a new polynomial $q:=\frac{p}{\sqrt{\operatorname{Var}%
\left[  p\right]  }}$. \ We have $\operatorname{Var}[q]=1$ and $\left\Vert
q\right\Vert _{\infty}\leq\frac{1}{\sqrt{\operatorname{Var}\left[  p\right]
}}$, which implies that there exists an $i\in\left[  N\right]  $ such that
$\operatorname{Inf}_{i}\left[  q\right]  \geq\frac{\operatorname{Var}\left[
p\right]  }{C^{k}}$, or equivalently $\operatorname{Inf}_{i}\left[  p\right]
\geq\frac{\operatorname{Var}\left[  p\right]  ^{2}}{C^{k}}$.
\end{proof}

We now consider the following algorithm for variable-splitting. \ We start
with $p_{0}:=p$. \ We then repeat the following, for $j\in\{0,1,2,\ldots\}$:

\begin{enumerate}
\item[(1)] If $\operatorname{Var}\left[  p_{j}\right]  <\frac{C}{N}$, then
halt and output $p_{j}$.

\item[(2)] Otherwise, choose some variable $x_{i}$\ such that
$\operatorname{Inf}_{i}\left[  p\right]  \geq\frac{1}{C^{k}}\operatorname{Var}%
\left[  p\right]  ^{2}$\ (which is guaranteed to exist by Corollary
\ref{cor:dfko}). \ Let $p_{j+1}$ be the polynomial obtained from $p_{j}$\ by
splitting $x_{i}$ into two variables. \ (In other words, by defining new
variables $x_{i,1}$\ and $x_{i,2}$, and then replacing every occurrence of
$x_{i}$\ in $p_{j}$\ by $\frac{x_{i,1}+x_{i,2}}{2}$.)
\end{enumerate}

When the algorithm halts (say at step $J$), we must have $\operatorname{Var}%
\left[  p_{J}\right]  <\frac{C}{N}$. \ Furthermore, observe that for every
$j$,
\begin{align*}
\operatorname{Var}[p_{j+1}]  &  =\operatorname{Var}[p_{j}]-\frac
{\operatorname{Inf}_{i}\left[  p_{j}\right]  }{2}\\
&  \leq\operatorname{Var}\left[  p_{j}\right]  -\frac{\operatorname{Var}%
\left[  p_{j}\right]  ^{2}}{2C^{k}}.
\end{align*}
Solving this recurrence (together with the initial condition
$\operatorname{Var}\left[  p_{0}\right]  \leq1$) implies that the algorithm
can continue for at most $O\left(  N\right)  $\ steps until
$\operatorname{Var}\left[  p_{J}\right]  =O(1/N)$. \ Therefore, the algorithm
introduces at most $O\left(  N\right)  $\ new variables.

\subsection{Requirement (ii)\label{REQII}}

We now show how to satisfy requirement (ii) in the special case $k=2$. \ To do
so, we will need another result from Dinur et al.\ \cite{dfko}.

\begin{lemma}
[\cite{dfko}, Lemma 1.3]\label{dfkolem}There exists a constant $C$ such that
the following holds. \ Let $p:\left\{  -1,1\right\}  ^{N}\rightarrow
\mathbb{R}$ be a polynomial of degree $k$,\ and suppose that%
\[
\sum_{i\in\left[  N\right]  }\hat{p}\left(  \left\{  i\right\}  \right)
^{2}\geq1
\]
whereas $\left\vert \hat{p}\left(  \left\{  i\right\}  \right)  \right\vert
\leq\frac{1}{Ckt}$ for all $i\in\left[  N\right]  $. \ Then
\[
\Pr_{x\in\left\{  -1,1\right\}  ^{N}}\left[  \left\vert p\left(  x\right)
\right\vert \geq t\right]  \geq\exp(-Ct^{2}k^{2}).
\]

\end{lemma}

By rescaling $p$ (similarly to Corollary \ref{cor:dfko}), we can transform
Lemma \ref{dfkolem}\ into the following.

\begin{corollary}
\label{cor:dfko2}There exists a constant $C$ such that the following holds.
\ For any bounded real polynomial $p:\{-1,1\}^{N}\rightarrow\left[
-1,1\right]  $ of degree $k$ that satisfies%
\[
\sum_{i\in\left[  N\right]  }\hat{p}\left(  \left\{  i\right\}  \right)
^{2}=v,
\]
there exists an $i\in\left[  N\right]  $ such that $\left\vert \hat{p}\left(
\left\{  i\right\}  \right)  \right\vert \geq\frac{v}{Ck}$.
\end{corollary}

\begin{proof}
Let us define a new polynomial $q\left(  x\right)  :=\frac{p\left(  x\right)
}{\sqrt{v}}$. \ Then $q\left(  x\right)  \in\left[  -\frac{1}{\sqrt{v}}%
,\frac{1}{\sqrt{v}}\right]  $\ for all $x\in\left\{  -1,1\right\}  ^{N}$, and%
\[
\sum_{i\in\left[  N\right]  }\hat{q}\left(  \left\{  i\right\}  \right)
^{2}=1.
\]
So by Lemma \ref{dfkolem}, there must exist an index $i\in\left[  N\right]
$\ such that%
\[
\left\vert \hat{q}\left(  \left\{  i\right\}  \right)  \right\vert
>\frac{\sqrt{v}}{Ck}%
\]
and hence%
\[
\left\vert \hat{p}\left(  \left\{  i\right\}  \right)  \right\vert >\frac
{v}{Ck}.
\]

\end{proof}

We now specialize to the case $k=2$. \ Observe that, given a bounded quadratic
polynomial $p:\left\{  -1,1\right\}  ^{N}\rightarrow\left[  -1,1\right]  $
that we are trying to estimate, we can assume without loss of generality
that\ $p$\ is homogeneous: in other words, is a quadratic form. \ First of
all, if $p$ contains a degree-$0$ term (i.e., an additive constant) $c$, then
we can replace it by the degree-$2$ term $cx_{1}^{2}$. \ Next, suppose we
decompose $p$ as a sum of its quadratic part $p_{2}$ and its linear part
$p_{1}$:%
\[
p\left(  x\right)  =p_{2}\left(  x\right)  +p_{1}\left(  x\right)  .
\]
Then%
\[
p\left(  -x\right)  =p_{2}\left(  x\right)  -p_{1}\left(  x\right)
\]
is also bounded in $\left[  -1,1\right]  $\ for all $x\in\left\{
-1,1\right\}  ^{N}$. \ So%
\[
p_{1}\left(  x\right)  =\frac{p\left(  x\right)  -p\left(  -x\right)  }{2}%
\in\left[  -1,1\right]  ,
\]
and likewise $p_{2}\left(  x\right)  \in\left[  -1,1\right]  $, for all
$x\in\left\{  -1,1\right\}  ^{N}$. \ Hence we can simply estimate
$p_{2}\left(  x\right)  $\ and $p_{1}\left(  x\right)  $ separately and then
sum them. \ Furthermore, $p_{1}$\ has the form%
\[
p_{1}\left(  x\right)  =a_{1}x_{1}+\cdots+a_{n}x_{n}%
\]
for some coefficients satisfying $\left\vert a_{1}\right\vert +\cdots
+\left\vert a_{n}\right\vert \leq1$. \ By standard tail bounds, such a linear
form is easy to estimate to within $\pm\varepsilon$\ by querying only
$O\left(  1\right)  $\ variables $x_{i}$.

So our problem reduces to estimating the bounded quadratic form $p_{2}\left(
x\right)  $. \ To do this, the first step is to apply the variable-splitting
algorithm described in Section \ref{REQI}. \ This produces a \textit{new}
bounded quadratic form, which we denote $P$:%
\[
P(x_{1},\ldots,x_{N})=\sum_{i,j\in\left[  N\right]  }a_{i,j}x_{i}x_{j}.
\]
(Abusing notation, we continue to call the number of variables $N$, even
though $O\left(  N\right)  $\ new variables have been introduced.) \ We have
$\operatorname{Var}[P]=O(\frac{1}{N})$. \ Now, our goal is to perform another
sequence of variable-splittings, which should introduce $O(N)$ new variables
and achieve requirement (ii).

Observe that, to achieve requirement (ii), it suffices to ensure%
\begin{equation}
\sum_{i\in\left[  N\right]  }\left(  \sum_{i\neq j}a_{ij}\right)
^{2}=O\left(  \frac{1}{\sqrt{N}}\right)  . \label{eq:squares}%
\end{equation}
For if we achieve (\ref{eq:squares}), then we can obtain the original
requirement (ii) by removing all squares $a_{ij}^{2}$ from the left hand side
(which only decreases it) and dividing it by $2$.

Let%
\begin{align*}
\operatorname{I}_{i}\left[  P\right]   &  :=\left(  \sum_{i\neq j}%
a_{ij}\right)  ^{2},\\
\operatorname{V}\left[  P\right]   &  :=\sum_{i\in\lbrack N]}\operatorname{I}%
_{i}\left[  P\right]  .
\end{align*}
We show the following counterpart of Corollary \ref{cor:dfko}.

\begin{lemma}
\label{lem:dfko}There exists a constant $C$ for which the following holds.
\ If $\operatorname{V}\left[  P\right]  \geq\frac{C\log N}{N}$, then there
exists an $i\in\left[  N\right]  $ such that $\operatorname{I}_{i}\left[
P\right]  =\Omega(\operatorname{V}\left[  P\right]  ^{2})$.
\end{lemma}

\begin{proof}
We take $P(x_{1},\ldots,x_{N})$ and, for each $i\in\lbrack N]$, substitute
$x_{i}=1$ with probability $1/2$ (with choices for different $i$ made
independently). \ Let $q$ be the resulting polynomial.

The key claim is the following: let $x_{i}$ be a variable for which we do not
substitute $x_{i}=1$. \ Then
\begin{equation}
\left\vert \hat{q}\left(  \left\{  i\right\}  \right)  -\frac{1}{2}\sum_{i\neq
j}a_{ij}\right\vert =O\left(  \sqrt{\operatorname{Inf}_{i}[P]\log\frac
{1}{\epsilon}}\right)  \label{eq:bounds}%
\end{equation}
with probability at least $1-\epsilon$.

To prove the claim: since each term $x_{i}x_{j}$ is transformed into $x_{i}$
with probability $1/2$ by substituting $x_{i}=1$, the expectation of $\hat
{q}\left(  \left\{  i\right\}  \right)  $ is equal to $\frac{1}{2}\sum_{j\neq
i}a_{ij}$. \ Since the decision to substitute or not substitute $x_{i}=1$
changes the value of $\hat{q}\left(  \left\{  i\right\}  \right)  $ by the
amount $a_{ij}$, Azuma's inequality implies%
\[
\Pr\left[  \left\vert \hat{q}\left(  \left\{  i\right\}  \right)  -\frac{1}%
{2}\sum_{i\neq j}a_{ij}\right\vert \geq t\right]  \leq\exp\left(  -\frac
{t^{2}}{2\sum_{j\in\left[  N\right]  }a_{ij}^{2}}\right)  .
\]
Using $\operatorname{Inf}_{i}[P]=\sum_{j\in\left[  N\right]  }a_{ij}^{2}$ and
taking $t=\sqrt{2\operatorname{Inf}_{i}[P]\log\frac{1}{\epsilon}}$ completes
the proof.

We now show how the claim implies the lemma. \ Let $\epsilon=\frac{1}{2N}$.
\ Then we can make the substitutions so that (\ref{eq:bounds}) holds for all
$i\in\lbrack N]$. \ By substituting $\epsilon=\frac{1}{2N}$ into
(\ref{eq:bounds}), we have
\begin{equation}
\left\vert \hat{q}\left(  \left\{  i\right\}  \right)  \right\vert \in\left[
\frac{\sqrt{\operatorname{I}_{i}\left[  P\right]  }}{2}-O\left(
\sqrt{\operatorname{Inf}_{i}\left[  P\right]  \log N}\right)  ,\frac
{\sqrt{\operatorname{I}_{i}\left[  P\right]  }}{2}+O\left(  \sqrt
{\operatorname{Inf}_{i}\left[  P\right]  \log N}\right)  \right]
\label{eq:bounds1}%
\end{equation}
By squaring (\ref{eq:bounds1}) and summing over all $i\in\left[  N\right]  $,
we obtain
\begin{align*}
\left\vert \sum_{i\in\left[  N\right]  }\hat{q}\left(  \left\{  i\right\}
\right)  ^{2}-\frac{1}{4}\sum_{i\in\left[  N\right]  }\operatorname{I}%
_{i}\left[  P\right]  \right\vert  &  =O\left(  \sqrt{\log N}\right)
\sum_{i\in\left[  N\right]  }\sqrt{\operatorname{I}_{i}\left[  P\right]
\operatorname{Inf}_{i}\left[  P\right]  }+O\left(  \log N\right)  \sum
_{i\in\left[  N\right]  }\operatorname{Inf}_{i}\left[  P\right] \\
&  \leq O\left(  \sqrt{\log N}\right)  \sqrt{\sum_{i\in\left[  N\right]
}\operatorname{I}_{i}\left[  P\right]  \cdot\sum_{i\in\left[  N\right]
}\operatorname{Inf}_{i}\left[  P\right]  }+O\left(  \log N\right)  \sum
_{i\in\left[  N\right]  }\operatorname{Inf}_{i}\left[  P\right] \\
&  \leq O\left(  \sqrt{\log N}\right)  \sqrt{\operatorname{V}\left[  P\right]
\cdot2\operatorname{Var}\left[  P\right]  }+O\left(  \log N\right)
\cdot2\operatorname{Var}\left[  P\right] \\
&  =O\left(  \sqrt{\operatorname{V}\left[  P\right]  \frac{\log N}{N}}%
+\frac{\log N}{N}\right)  ,
\end{align*}
where the second line used Cauchy-Schwarz and the third used Lemma
\ref{infbound}.

Now, if $\operatorname{V}\left[  P\right]  \geq C\frac{\log N}{N}$ for a
sufficiently large constant $C$, then the above implies that%
\begin{align*}
\sum_{i\in\left[  N\right]  }\hat{q}\left(  \left\{  i\right\}  \right)  ^{2}
&  \geq\frac{1}{4}\operatorname{V}\left[  P\right]  -O\left(  \sqrt
{\operatorname{V}\left[  P\right]  \frac{\log N}{N}}+\frac{\log N}{N}\right)
\\
&  =\Omega\left(  \operatorname{V}\left[  P\right]  \right)  .
\end{align*}
By Corollary \ref{cor:dfko2}, this means that there exists an index
$i\in\left[  N\right]  $ such that $\left\vert \hat{q}\left(  \left\{
i\right\}  \right)  \right\vert =\Omega(\operatorname{V}\left[  P\right]  )$.
\ From equation (\ref{eq:bounds1}), we get that $\sqrt{\operatorname{I}%
_{i}\left[  P\right]  }=\Omega(\operatorname{V}\left[  P\right]  )$ and
$\operatorname{I}_{i}\left[  P\right]  =\Omega(\operatorname{V}\left[
P\right]  ^{2})$.
\end{proof}

Given Lemma \ref{lem:dfko}, we can use the same algorithm as in the previous
section. \ That is, we repeatedly choose a variable $i\in\left[  N\right]  $
that maximizes $\operatorname{I}_{i}\left[  P\right]  $ and split the variable
$x_{i}$. \ Let $p_{0},p_{1},\ldots$ be the resulting sequence of polynomials.
\ Initially, we have
\[
\operatorname{V}\left[  p_{0}\right]  =\sum_{i}\left(  \sum_{j\neq i}%
a_{ij}\right)  ^{2}\leq N\sum_{i}\sum_{j}a_{ij}^{2}=O(1),
\]
with the last equality following from $\operatorname{Var}[P]=O(\frac{1}{N})$,
which was achieved by the previous sequence of variable-splittings. \ We also
have
\[
\operatorname{V}[p_{j+1}]\leq\operatorname{V}[p_{j}]-\Omega(\operatorname{V}%
[p_{j}]^{2}),
\]
as long as $\operatorname{V}[p_{j}]\geq\frac{C\log N}{N}$. \ This means that
after $O(N/\log N)$ variable-splittings, we achieve $\operatorname{V}%
[p_{j}]<\frac{C\log N}{N}$, which is substantially stronger than the
requirement (\ref{eq:squares}) that we needed.

\section{\label{KFOLDLB}Appendix: Lower Bound for $k$-fold
\textsc{Forrelation}}

In this appendix, we use the machinery developed in Section \ref{GAP} to prove
a $\widetilde{\Omega}(\sqrt{N})$\ lower bound on the randomized query
complexity of $k$-fold \textsc{Forrelation}, for all $k\geq2$. \ Ideally, of
course, we would like to prove a lower bound that gets \textit{better} as $k$
gets larger, but even proving the same kind of lower bound that we had in the
$k=2$\ case will take some work.

In $k$-fold \textsc{Forrelation}, recall that we are given oracle access to
Boolean functions $f_{1},\ldots,f_{k}:\left\{  0,1\right\}  ^{n}%
\rightarrow\left\{  -1,1\right\}  $, and are interested in the quantity%
\[
\Phi_{f_{1},\ldots,f_{k}}:=\frac{1}{2^{\left(  k+1\right)  n/2}}\sum
_{x_{1},\ldots,x_{k}\in\left\{  0,1\right\}  ^{n}}f_{1}\left(  x_{1}\right)
\left(  -1\right)  ^{x_{1}\cdot x_{2}}f_{2}\left(  x_{2}\right)  \left(
-1\right)  ^{x_{2}\cdot x_{3}}\cdots\left(  -1\right)  ^{x_{k-1}\cdot x_{k}%
}f_{k}\left(  x_{k}\right)  .
\]
The problem is to decide whether $\left\vert \Phi_{f_{1},\ldots,f_{k}%
}\right\vert \leq\frac{1}{100}$\ or $\Phi_{f_{1},\ldots,f_{k}}\geq\frac{3}{5}%
$, promised that one of these is the case.

We will prove a lower bound for $k$-fold \textsc{Forrelation}\ by reducing the
\textsc{Gaussian Distinguishing}\ problem to it, and then applying Theorem
\ref{gdlb}---thereby illustrating the usefulness of formulating a general
lower bound for \textsc{Gaussian Distinguishing}.

\subsection{Concentration Inequalities\label{CONC}}

The first step is to prove some concentration inequalities for $k$-fold
\textsc{Forrelation}. \ In what follows, recall that $f_{\ell}^{\left(
x\right)  }\left(  x_{\ell}\right)  $ is defined as $f_{\ell}\left(  x_{\ell
}\right)  \left(  -1\right)  ^{x_{\ell}\cdot x}$, for all $\ell\in\left[
k\right]  $\ and $x_{\ell},x\in\left\{  0,1\right\}  ^{n}$.

\begin{lemma}
\label{phibound}Suppose $f_{1},\ldots,f_{k}:\left\{  0,1\right\}
^{n}\rightarrow\left\{  -1,1\right\}  $\ are chosen uniformly at random.
\ Then%
\[
\Pr_{f_{1},\ldots,f_{k}}\left[  \left\vert \Phi_{f_{1},\ldots,f_{k}%
}\right\vert \geq\frac{t}{\sqrt{N}}\right]  =O\left(  \frac{1}{t^{t}}\right)
.
\]

\end{lemma}

\begin{proof}
Imagine that $f_{1},\ldots,f_{k-1}$\ are fixed, so that we are considering
$\Phi_{f_{1},\ldots,f_{k}}$\ solely as a function of $f_{k}$. \ We have%
\[
\Phi_{f_{1},\ldots,f_{k}}=\frac{1}{\sqrt{N}}\sum_{x\in\left\{  0,1\right\}
^{n}}\alpha_{x}f_{k}\left(  x\right)
\]
where%
\[
\alpha_{x}=\Phi_{f_{1},\ldots,f_{k-2},f_{k-1}^{\left(  x\right)  }}=\frac
{1}{2^{kn/2}}\sum_{x_{1},\ldots,x_{k-1}\in\left\{  0,1\right\}  ^{n}}%
f_{1}\left(  x_{1}\right)  \left(  -1\right)  ^{x_{1}\cdot x_{2}}\cdots\left(
-1\right)  ^{x_{k-2}\cdot x_{k-1}}f_{k-1}\left(  x_{k-1}\right)  \left(
-1\right)  ^{x_{k-1}\cdot x}.
\]
Thus, by equation (\ref{ss}),%
\[
\sum_{x\in\left\{  0,1\right\}  ^{n}}\alpha_{x}^{2}=\sum_{x\in\left\{
0,1\right\}  ^{n}}\Phi_{f_{1},\ldots,f_{k-2},f_{k-1}^{\left(  x\right)  }}%
^{2}=1.
\]
So in particular, $\left\vert \alpha_{x}\right\vert \leq1$ for all $x$. \ We
now appeal to Bennett's inequality, which tells us that%
\[
\Pr_{f_{k}:\left\{  0,1\right\}  ^{n}\rightarrow\left\{  -1,1\right\}
}\left[  \left\vert \sum_{x\in\left\{  0,1\right\}  ^{n}}\alpha_{x}%
f_{k}\left(  x\right)  \right\vert \geq t\right]  \leq2\exp\left(  -h\left(
t\right)  \right)
\]
where%
\[
h\left(  t\right)  :=\left(  1+t\right)  \ln\left(  1+t\right)  -t.
\]

\end{proof}

The following is also useful.

\begin{lemma}
\label{balanced}Suppose $f_{1},\ldots,f_{k}:\left\{  0,1\right\}
^{n}\rightarrow\left\{  -1,1\right\}  $\ are chosen uniformly at random.
\ Then with probability $1-O\left(  1/N\right)  $, we have%
\[
\left\vert \sum_{z~:~z\cdot y=0}\Phi_{f_{1},\ldots,f_{k-1},f_{k}^{\left(
z\right)  }}^{2}-\frac{1}{2}\right\vert \leq\frac{\log^{5/2}N}{\sqrt{N}}%
\]
for all $y\in\left\{  0,1\right\}  ^{n}$.
\end{lemma}

\begin{proof}
By symmetry, we can assume without loss of generality that $y=10\cdots0$, so
that the sum is over all $z$ that start with $0$.

As in the proof of Lemma \ref{phibound}, imagine that $f_{1},\ldots,f_{k-1}%
$\ are fixed, so that we are considering%
\[
\Phi_{f_{1},\ldots,f_{k-1},f_{k}^{\left(  z\right)  }}=\frac{1}{\sqrt{N}}%
\sum_{x\in\left\{  0,1\right\}  ^{n}}\alpha_{x}f_{k}\left(  x\right)  \left(
-1\right)  ^{x\cdot z}%
\]
solely as a function of $f_{k}$. \ Then it is not hard to see that%
\[
\sum_{z~:~z\cdot y=0}\Phi_{f_{1},\ldots,f_{k-1},f_{k}^{\left(  z\right)  }%
}^{2}%
\]
is just the probability of measuring an output string that starts with $0$
when the \textsc{Forrelation}\ algorithm\ is run on $f_{1},\ldots,f_{k}$,
which also equals%
\begin{align*}
\sum_{x\in\left\{  0,1\right\}  ^{n-1}}\left(  \frac{\alpha_{0x}f_{k}\left(
0x\right)  +\alpha_{1x}f_{k}\left(  1x\right)  }{\sqrt{2}}\right)  ^{2}  &
=\frac{1}{2}\sum_{x\in\left\{  0,1\right\}  ^{n-1}}\left(  \alpha_{0x}%
f_{k}\left(  0x\right)  +\alpha_{1x}f_{k}\left(  1x\right)  \right)  ^{2}\\
&  =\frac{1}{2}\sum_{x\in\left\{  0,1\right\}  ^{n-1}}\left(  \alpha_{0x}%
^{2}+\alpha_{1x}^{2}+2\alpha_{0x}\alpha_{1x}f_{k}\left(  0x\right)
f_{k}\left(  1x\right)  \right) \\
&  =\frac{1}{2}+\sum_{x\in\left\{  0,1\right\}  ^{n-1}}\alpha_{0x}\alpha
_{1x}f_{k}\left(  0x\right)  f_{k}\left(  1x\right)  .
\end{align*}
Now, each $f_{k}\left(  0x\right)  f_{k}\left(  1x\right)  $\ is an
independent, uniform $\left\{  -1,1\right\}  $\ random variable.
\ Furthermore, by Lemma \ref{phibound}, with $1-o\left(  1/N\right)
$\ probability we have $\left\vert \alpha_{x}\right\vert \leq\frac{\log
N}{\sqrt{N}}$\ for all $x\in\left\{  0,1\right\}  ^{n}$, in which case%
\[
\left\vert \alpha_{0x}\alpha_{1x}\right\vert \leq\frac{\log^{2}N}{N}%
\]
for all $x\in\left\{  0,1\right\}  ^{n-1}$. \ By Hoeffding's inequality, it
follows that%
\begin{align*}
\Pr\left[  \left\vert \sum_{z~:~z\cdot y=0}\Phi_{f_{1},\ldots,f_{k-1}%
,f_{k}^{\left(  z\right)  }}^{2}-\frac{1}{2}\right\vert \geq\frac{t}{\sqrt{N}%
}\right]   &  \leq2\exp\left(  -\frac{2(t/\sqrt{N})^{2}}{\left(  N/2\right)
\left(  \frac{2\log^{2}N}{N}\right)  ^{2}}\right) \\
&  =2\exp\left(  -\frac{t^{2}}{4\log^{4}N}\right)  .
\end{align*}
Setting $t=C\log^{5/2}N$ for some constant $C$, this probability is at most
$1/N^{2}$. \ So by the union bound, the probability is at most $1/N$ when
summed over all $y$.
\end{proof}

\subsection{Continuous/Discrete Hybrid\label{HYBRID}}

Just like we did in the $k=2$\ case, it is convenient to define a continuous
analogue of the $k$-fold\ \textsc{Forrelation} problem---though in this case,
the problem will be a hybrid of continuous and discrete. \ In $k$-fold
\textsc{Real Forrelation}, we are given oracle access to functions
$f_{1},\ldots,f_{k-2}:\left\{  0,1\right\}  ^{n}\rightarrow\left\{
-1,1\right\}  $ as well as $f_{k-1},f_{k}:\left\{  0,1\right\}  ^{n}%
\rightarrow\mathbb{R}$. \ We are promised that each $f_{i}\left(
x_{i}\right)  $\ (for $i\in\left[  k-2\right]  $) is chosen uniformly and
independently from $\left\{  -1,1\right\}  $, and \textit{also} that one of
the following holds:

\begin{enumerate}
\item[(i)] Uniform measure $\mathcal{U}$: Each $f_{k-1}\left(  x_{k-1}\right)
$\ and $f_{k}\left(  x_{k}\right)  $\ is an independent $\mathcal{N}\left(
0,1\right)  $\ Gaussian.

\item[(ii)] Forrelated measure $\mathcal{F}$: Each $f_{k-1}\left(
x_{k-1}\right)  $\ is an independent $\mathcal{N}\left(  0,1\right)
$\ Gaussian, and each $f_{k}\left(  x_{k}\right)  $ is set equal to%
\[
\frac{1}{2^{\left(  k-1\right)  n/2}}\sum_{x_{1},\ldots,x_{k-1}\in\left\{
0,1\right\}  ^{n}}f_{1}\left(  x_{1}\right)  \left(  -1\right)  ^{x_{1}\cdot
x_{2}}\cdots\left(  -1\right)  ^{x_{k-2}\cdot x_{k-1}}f_{k-1}\left(
x_{k-1}\right)  \left(  -1\right)  ^{x_{k-1}\cdot x_{k}}.
\]

\end{enumerate}

The problem is to decide whether (i) or (ii) holds.

Here is another way to think about $k$-fold \textsc{Real Forrelation}: let%
\[
f\left(  x_{k-1}\right)  :=c_{x_{k-1}}f_{k-1}\left(  x_{k-1}\right)
,~~~~~~~~~~g\left(  x_{k}\right)  :=f_{k}\left(  x_{k}\right)  ,
\]
where%
\begin{align*}
c_{x_{k-1}}  &  =\sqrt{N}\Phi_{f_{1},\ldots,f_{k-3},f_{k-2}^{\left(
x_{k-1}\right)  }}\\
&  =\frac{1}{2^{\left(  k-2\right)  n/2}}\sum_{x_{1},\ldots,x_{k-2}\in\left\{
0,1\right\}  ^{n}}f_{1}\left(  x_{1}\right)  \left(  -1\right)  ^{x_{1}\cdot
x_{2}}\cdots f_{1}\left(  x_{k-2}\right)  \left(  -1\right)  ^{x_{k-2}\cdot
x_{k-1}}.
\end{align*}
Then%
\[
\Phi_{f_{1},\ldots,f_{k}}=\Phi_{f,g}=\frac{1}{2^{3n/2}}\sum_{x,y\in\left\{
0,1\right\}  ^{n}}f\left(  x\right)  \left(  -1\right)  ^{x\cdot y}g\left(
y\right)  .
\]
In other words, we can think of $k$-fold \textsc{Real Forrelation}\ as
equivalent to ordinary, $2$-fold \textsc{Real Forrelation}\ on the functions
$f_{k-1}\left(  x\right)  $\ and $f_{k}\left(  y\right)  $, except that each
$f_{k-1}\left(  x\right)  $\ is \textquotedblleft twisted\textquotedblright%
\ by a multiplicative factor $c_{x}$\ depending on $f_{1},\ldots,f_{k-2}$.
\ We have the following useful fact.

\begin{proposition}
\label{logub}With $1-o\left(  1/N\right)  $ probability over $f_{1}%
,\ldots,f_{k-2}$, we have $\left\vert c_{x}\right\vert \leq\log N$ for all
$x\in\left\{  0,1\right\}  ^{n}$.
\end{proposition}

\begin{proof}
By Lemma \ref{phibound},%
\[
\Pr_{f_{1},\ldots,f_{k-2}}\left[  \left\vert c_{x}\right\vert \geq\log
N\right]  =O\left(  \frac{1}{\left(  \log N\right)  ^{\log N}}\right)
=o\left(  \frac{1}{N^{2}}\right)  .
\]
The proposition now follows from the union bound.
\end{proof}

Note also that%
\[
\sum_{x\in\left\{  0,1\right\}  ^{n}}c_{x}^{2}=N.
\]
Next, we prove a $k$-fold analogue of Theorem \ref{2overpi}, showing that
$k$-fold \textsc{Real Forrelation}\ can be reduced to Boolean $k$%
-fold\ \textsc{Forrelation}.

\begin{theorem}
\label{2overpi2}Fix $f_{1},\ldots,f_{k-2}$\ (or equivalently, the multipliers
$c_{x}$). \ Suppose $\left\langle f,g\right\rangle =\left\langle f_{k-1}%
,f_{k}\right\rangle $ are drawn from the forrelated measure $\mathcal{F}$.
\ Define Boolean functions $F,G:\left\{  0,1\right\}  ^{n}\rightarrow\left\{
-1,1\right\}  $\ by $F\left(  x\right)  :=\operatorname*{sgn}\left(  f\left(
x\right)  \right)  $ and $G\left(  y\right)  :=\operatorname*{sgn}\left(
g\left(  y\right)  \right)  $. \ Then%
\[
\operatorname{E}\left[  \Phi_{F,G}\right]  =\frac{2}{\pi}\pm O\left(
\frac{\log^{3}N}{N}\right)  .
\]

\end{theorem}

\begin{proof}
By Proposition \ref{logub}, we can assume without loss of generality that
$\left\vert c_{x}\right\vert \leq\log N$ for all $x\in\left\{  0,1\right\}
^{n}$ (the times when this assumption fails can only change $\operatorname{E}%
\left[  \Phi_{F,G}\right]  $\ by $o\left(  1/N\right)  $).

By linearity of expectation, it suffices to calculate $\operatorname{E}\left[
c_{x}F\left(  x\right)  \left(  -1\right)  ^{x\cdot y}G\left(  y\right)
\right]  $\ for some specific $x,y$ pair. \ Let $v\in\mathbb{R}^{N}$ be a
vector of independent $\mathcal{N}\left(  0,1\right)  $\ Gaussians. \ Then we
can consider $\left\langle F,G\right\rangle $\ to have been generated as
follows:%
\begin{align*}
F\left(  x\right)   &  =\operatorname*{sgn}\left(  v_{x}\right)  ,\\
G\left(  y\right)   &  =\operatorname*{sgn}\left(  \sum_{x\in\left\{
0,1\right\}  ^{n}}c_{x}v_{x}\left(  -1\right)  ^{x\cdot y}\right)  .
\end{align*}
Let%
\[
Z:=\sum_{x^{\prime}\neq x}c_{x^{\prime}}v_{x^{\prime}}\left(  -1\right)
^{x^{\prime}\cdot y},
\]
and let $G^{\prime}\left(  y\right)  :=\operatorname*{sgn}\left(  Z\right)  $.
\ Then%
\[
\operatorname{E}\left[  c_{x}F\left(  x\right)  \left(  -1\right)  ^{x\cdot
y}G^{\prime}\left(  y\right)  \right]  =\operatorname{E}\left[  c_{x}%
\operatorname*{sgn}\left(  v_{x}\right)  \left(  -1\right)  ^{x\cdot
y}\operatorname*{sgn}\left(  Z\right)  \right]  =0,
\]
since $v_{x}$\ and $Z$ are independent Gaussians both with mean $0$. \ Note
that adding $c_{x}v_{x}\left(  -1\right)  ^{x\cdot y}$\ back to $Z$\ can only
flip $Z$ to having the same sign as $c_{x}\operatorname*{sgn}\left(
v_{x}\right)  \left(  -1\right)  ^{x\cdot y}$, not the opposite sign, and
hence can only increase $c_{x}F\left(  x\right)  \left(  -1\right)  ^{x\cdot
y}G\left(  y\right)  $. \ It follows that%
\[
\operatorname{E}\left[  c_{x}F\left(  x\right)  \left(  -1\right)  ^{x\cdot
y}G\left(  y\right)  \right]  =2\Pr\left[  G\left(  y\right)  \neq G^{\prime
}\left(  y\right)  \right]  .
\]

The event $G\left(  y\right)  \neq G^{\prime}\left(  y\right)  $\ occurs if
and only if the following two events both occur:%
\begin{align*}
\left\vert c_{x}v_{x}\right\vert  &  >\left\vert Z\right\vert ,\\
\operatorname*{sgn}\left(  c_{x}v_{x}\left(  -1\right)  ^{x\cdot y}\right)
&  \neq\operatorname*{sgn}\left(  Z\right)  .
\end{align*}
Since the distribution of $v_{x}$\ is symmetric about $0$, we can assume
without loss of generality that $c_{x}\left(  -1\right)  ^{x\cdot y}=1$.

Let $Z\left(  t\right)  $\ be the probability density function of $Z$. \ Then%
\[
\Pr\left[  \left\vert c_{x}v_{x}\right\vert >\left\vert Z\right\vert \text{
\ and \ }\operatorname*{sgn}\left(  v_{x}\right)  \neq\operatorname*{sgn}%
\left(  Z\right)  \right]  =2\int_{t=0}^{\infty}Z\left(  t\right)  \Pr\left[
c_{x}v_{x}>t\right]  dt.
\]

Now, $Z$ is a linear combination of $N-1$\ independent $\mathcal{N}\left(
0,1\right)  $\ Gaussians, with coefficients $\left\{  c_{x^{\prime}}\right\}
_{x^{\prime}\neq x}$. \ This means that $Z$\ has the $\mathcal{N}\left(
0,N-c_{x}^{2}\right)  $\ Gaussian distribution. \ Therefore%
\begin{align*}
2\int_{t=0}^{\infty}Z\left(  t\right)  \Pr\left[  c_{x}v_{x}>t\right]  dt  &
=\frac{2}{\sqrt{2\pi\left(  N-c_{x}^{2}\right)  }}\int_{t=0}^{\infty}%
\exp\left(  -\frac{t^{2}}{2\left(  N-c_{x}^{2}\right)  }\right)  \Pr\left[
c_{x}v_{x}>t\right]  dt\\
&  \leq\frac{2}{\sqrt{2\pi\left(  N-c_{x}^{2}\right)  }}\int_{t=0}^{\infty}%
\Pr\left[  c_{x}v_{x}>t\right]  dt\\
&  =\frac{2}{\sqrt{2\pi\left(  N-c_{x}^{2}\right)  }}\operatorname{E}\left[
\left\vert c_{x}v_{x}\right\vert \right] \\
&  =\frac{2\left\vert c_{x}\right\vert }{\pi\sqrt{N-c_{x}^{2}}}\\
&  \leq\frac{2\left\vert c_{x}\right\vert }{\pi\sqrt{N}}+O\left(
\frac{\left\vert c_{x}\right\vert ^{3}}{N^{3/2}}\right)  .
\end{align*}
In the other direction, for all $C>0$ we have%
\begin{align*}
2\int_{t=0}^{\infty}Z\left(  t\right)  \Pr\left[  c_{x}v_{x}>t\right]  dt  &
=\frac{2}{\sqrt{2\pi\left(  N-c_{x}^{2}\right)  }}\int_{t=0}^{\infty}%
\exp\left(  -\frac{t^{2}}{2\left(  N-c_{x}^{2}\right)  }\right)  \Pr\left[
c_{x}v_{x}>t\right]  dt\\
&  \geq\frac{2}{\sqrt{2\pi N}}\int_{t=0}^{C}\exp\left(  -\frac{t^{2}}{2\left(
N-c_{x}^{2}\right)  }\right)  \Pr\left[  c_{x}v_{x}>t\right]  dt\\
&  \geq\frac{2}{\sqrt{2\pi N}}\exp\left(  -\frac{C^{2}}{2\left(  N-c_{x}%
^{2}\right)  }\right)  \int_{t=0}^{C}\Pr\left[  c_{x}v_{x}>t\right]  dt\\
&  =\frac{2}{\sqrt{2\pi N}}\exp\left(  -\frac{C^{2}}{2\left(  N-c_{x}%
^{2}\right)  }\right)  \left(  \operatorname{E}\left[  \left\vert c_{x}%
v_{x}\right\vert \right]  -\frac{1}{c_{x}\sqrt{2\pi}}\int_{t=C}^{\infty
}te^{-t^{2}/\left(  2c_{x}^{2}\right)  }dt\right) \\
&  =\frac{2}{\sqrt{2\pi N}}\exp\left(  -\frac{C^{2}}{2\left(  N-c_{x}%
^{2}\right)  }\right)  \left\vert c_{x}\right\vert \left(  \sqrt{\frac{2}{\pi
}}-\frac{e^{-C^{2}/2}}{\sqrt{2\pi}}\right)  .
\end{align*}
If we set $C:=\sqrt{\log N}$, then using $\left\vert c_{x}\right\vert \leq\log
N$, the above is%
\[
\frac{2\left\vert c_{x}\right\vert }{\sqrt{2\pi N}}\left(  1-O\left(
\frac{\log N}{N}\right)  \right)  \left(  \sqrt{\frac{2}{\pi}}-\frac{1}%
{\sqrt{2\pi}N}\right)  \geq\frac{2\left\vert c_{x}\right\vert }{\pi\sqrt{N}%
}-O\left(  \frac{\log N}{N^{3/2}}\right)  .
\]

Therefore%
\begin{align*}
\operatorname{E}\left[  \Phi_{F,G}\right]   &  =\frac{1}{2^{3n/2}}\sum
_{x,y\in\left\{  0,1\right\}  ^{n}}\operatorname{E}\left[  c_{x}F\left(
x\right)  \left(  -1\right)  ^{x\cdot y}G\left(  y\right)  \right] \\
&  =\frac{1}{N^{3/2}}\sum_{x,y\in\left\{  0,1\right\}  ^{n}}\left(
\frac{2\left\vert c_{x}\right\vert }{\pi\sqrt{N}}+O\left(  \frac{\left\vert
c_{x}\right\vert ^{3}}{N^{3/2}}\right)  -O\left(  \frac{\log N}{N^{3/2}%
}\right)  \right) \\
&  =\frac{2}{\pi}+\frac{1}{N^{2}}\sum_{x\in\left\{  0,1\right\}  ^{n}%
}\left\vert c_{x}\right\vert ^{3}-O\left(  \frac{\log N}{N}\right)
\backslash\\
&  =\frac{2}{\pi}\pm O\left(  \frac{\log^{3}N}{N}\right)  .
\end{align*}

\end{proof}

By direct analogy to Corollary \ref{realfor}, Theorem \ref{2overpi2}\ implies
that there exists a reduction from $k$-fold \textsc{Real Forrelation}\ to
$k$-fold \textsc{Forrelation}.

\begin{corollary}
\label{realfor2}Suppose there exists a $T$-query algorithm that solves
$k$-fold \textsc{Forrelation}\ with bounded error. \ Then there also exists an
$O\left(  T\right)  $-query algorithm that solves $k$-fold \textsc{Real
Forrelation} with bounded error.
\end{corollary}

\subsection{Lower Bound}

Finally, we apply our lower bound for \textsc{Gaussian Distinguishing}\ to
obtain a lower bound on the randomized query complexity of $k$-fold
\textsc{Real Forrelation}, which almost matches what we obtained for the
$2$-fold case.

\begin{theorem}
$k$-fold \textsc{Real Forrelation} requires $\Omega(\sqrt{N}/\log^{7/2}%
N)$\ randomized queries.\label{kfoldlbthm}
\end{theorem}

\begin{proof}
The strategy is the following: we will give the values $f_{1}\left(
x_{1}\right)  ,\ldots,f_{k-2}\left(  x_{k-2}\right)  $, for all $x_{1}%
,\ldots,x_{k-2}\in\left\{  0,1\right\}  ^{n}$, away to the algorithm
\textquotedblleft free of charge.\textquotedblright\ \ This can only make our
lower bound stronger.

As we saw above, after we do this, $k$-fold \textsc{Real Forrelation} becomes
equivalent to ordinary $2$-fold \textsc{Real Forrelation}\ on the functions
$f\left(  x\right)  =c_{x}f_{k-1}\left(  x\right)  $ and $g\left(  y\right)
=f_{k}\left(  y\right)  $, where the $c_{x}$'s are known multipliers. \ This,
in turn, can be expressed as an instance of \textsc{Gaussian Distinguishing},
in which our set $\mathcal{V}$\ of test vectors consists of $\left\vert
1\right\rangle ,\ldots,\left\vert N\right\rangle $ along with the vectors
$\left\{  \left\vert \psi_{y}\right\rangle \right\}  _{y\in\left\{
0,1\right\}  ^{n}}$\ defined as follows:%
\[
\left\vert \psi_{y}\right\rangle :=\frac{1}{\sqrt{N}}\sum_{x\in\left\{
0,1\right\}  ^{n}}c_{x}\left(  -1\right)  ^{x\cdot y}\left\vert x\right\rangle
.
\]
Note that the $\left\vert \psi_{y}\right\rangle $'s are unit vectors, since
$\sum_{x}c_{x}^{2}=N$. \ As for inner products, with $1-O\left(  1/N\right)
$\ probability we have%
\[
\left\vert \left\langle x|\psi_{y}\right\rangle \right\vert =\frac{\left\vert
c_{x}\right\vert }{\sqrt{N}}\leq\frac{\log N}{\sqrt{N}}%
\]
for all $x,y$\ by Proposition \ref{logub}, and%
\[
\left\vert \left\langle \psi_{y}|\psi_{z}\right\rangle \right\vert =\frac
{1}{N}\left\vert \sum_{x\in\left\{  0,1\right\}  ^{n}}c_{x}^{2}\left(
-1\right)  ^{x\cdot\left(  y\oplus z\right)  }\right\vert =O\left(  \frac
{\log^{5/2}N}{\sqrt{N}}\right)
\]
for all $y,z$\ by Lemma \ref{balanced}. \ So, in summary,\ we have $\left\vert
\mathcal{V}\right\vert =2N$\ and $\left\vert \left\langle v|w\right\rangle
\right\vert \leq\varepsilon=O\left(  \frac{\log^{5/2}N}{\sqrt{N}}\right)  $
for all distinct $\left\vert v\right\rangle ,\left\vert w\right\rangle
\in\mathcal{V}$. \ By Theorem \ref{gdlb}, this implies that%
\[
\Omega\left(  \frac{1/\varepsilon}{\log(2N/\varepsilon)}\right)
=\Omega\left(  \frac{\sqrt{N}}{\log^{7/2}N}\right)
\]
queries are needed.
\end{proof}

\section{Appendix: Property Testing\label{PROP}}

In this appendix, we show that \textsc{Forrelation}\ can be recast as a
\textit{property-testing} problem: that is, as a problem of deciding whether
the functions $f,g:\left\{  0,1\right\}  ^{n}\rightarrow\left\{  -1,1\right\}
$\ satisfy a certain property, or are far in Hamming distance from any
functions satisfying that property. \ (For a recent survey of quantum
property-testing, see Montanaro and de Wolf \cite{mdw}.)

In particular, we will obtain a property of $N$-bit strings that

\begin{enumerate}
\item[(1)] can be quantumly $\varepsilon$-tested (with bounded error) using
only $O\left(  1/\varepsilon\right)  $ queries, but

\item[(2)] requires $\Omega(\frac{\sqrt{N}}{\log N})$\ queries to
$\varepsilon$-test classically,\ provided $\varepsilon$\ is a sufficiently
small constant.
\end{enumerate}

This is the largest quantum versus classical property-testing separation yet known.

Since our analysis works for $k$-fold \textsc{Forrelation} just as easily as
for $2$-fold, we will use the more general setting. \ Let $\mathcal{Y}$\ be
the set of all $k$-tuples of Boolean functions $\left\langle f_{1}%
,\ldots,f_{k}\right\rangle $\ such that $\Phi_{f_{1},\ldots,f_{k}}\leq\frac
{1}{100}$. \ Also, let $\mathcal{N}_{\varepsilon}$\ be the set of all
$k$-tuples $\left\langle g_{1},\ldots,g_{k}\right\rangle $\ that differ from
every $\left\langle f_{1},\ldots,f_{k}\right\rangle \in\mathcal{Y}$\ on at
least $\varepsilon\cdot k2^{n}$\ of the $k2^{n}$\ positions. \ Then we will be
interested in the problem of deciding whether $\left\langle f_{1},\ldots
,f_{k}\right\rangle \in\mathcal{Y}$\ or $\left\langle f_{1},\ldots
,f_{k}\right\rangle \in\mathcal{N}_{\varepsilon}$, promised that one of these
is the case.

Our goal is to show that both our quantum algorithm for \textsc{Forrelation},
\textit{and} our randomized lower bound for it, carry over to the
property-testing variant. \ There are two difficulties here. \ First, one
could imagine an $\left\langle f_{1},\ldots,f_{k}\right\rangle $\ that was far
from $\mathcal{Y}$\ in the Hamming distance sense, yet had an $\Phi
_{f_{1},\ldots,f_{k}}$\ value only \textit{slightly} greater than $\frac
{1}{100}$---in which case, we would need many repetitions of our quantum
algorithm to separate $\left\langle f_{1},\ldots,f_{k}\right\rangle $\ from
$\mathcal{Y}$. \ Second, one could imagine an $\left\langle f_{1},\ldots
,f_{k}\right\rangle $\ that was close to $\mathcal{Y}$\ in Hamming distance,
yet had (say) $\Phi_{f_{1},\ldots,f_{k}}\geq\frac{3}{5}$---in which case, the
known classical hardness of distinguishing $\left\vert \Phi_{f_{1}%
,\ldots,f_{k}}\right\vert \leq\frac{1}{100}$\ from $\Phi_{f_{1},\ldots,f_{k}%
}\geq\frac{3}{5}$\ might not imply anything about the hardness of
distinguishing $\left\langle f_{1},\ldots,f_{k}\right\rangle \in\mathcal{Y}%
$\ from $\left\langle f_{1},\ldots,f_{k}\right\rangle \in\mathcal{N}%
_{\varepsilon}$, causing the classical lower bound to
fail.\footnote{Furthermore, this worry is not farfetched: if $k\geq3$, then
there really \textit{are} cases where changing a single function value
can\ change $\Phi_{f_{1},\ldots,f_{k}}$\ dramatically. \ For example, let
$f_{1}$, $f_{2}$, and $f_{3}$ each be the identically-$1$ function. \ Then
$\Phi_{f_{1},f_{2},f_{3}}=1$. \ But if we simply change $f_{2}\left(
0^{n}\right)  $\ from $1$\ to $-1$, then $\Phi_{f_{1},f_{2},f_{3}}=-1$.}

Fortunately, we can deal with both difficulties. \ To start with the first:

\begin{lemma}
\label{proptest}Let $f_{1},\ldots,f_{k}:\left\{  0,1\right\}  ^{n}%
\rightarrow\left\{  -1,1\right\}  $\ be Boolean functions satisfying
$\Phi_{f_{1},\ldots,f_{k}}\geq0$. \ Then for all $\varepsilon>0$, there exist
functions $g_{1},\ldots,g_{k}$\ such that each $g_{i}$\ differs from\ $f_{i}%
$\ on at most $\varepsilon2^{n}$ coordinates, and $\Phi_{g_{1},\ldots,g_{k}%
}\leq\left(  1-\varepsilon\right)  ^{k}\Phi_{f_{1},\ldots,f_{k}}$.
\end{lemma}

\begin{proof}
We form each $g_{i}$\ by simply choosing a subset $S_{i}\subset\left\{
0,1\right\}  ^{n}$\ with $\left\vert S_{i}\right\vert =\varepsilon2^{n}%
$\ uniformly at random, then picking $g_{i}\left(  x\right)  $\ uniformly at
random if $x\in S_{i}$, or setting $g_{i}\left(  x\right)  :=f_{i}\left(
x\right)  $\ if $x\notin S_{i}$. \ By linearity of expectation,%
\[
\operatorname*{E}\left[  \Phi_{g_{1},\ldots,g_{k}}\right]  =\frac
{1}{2^{\left(  k+1\right)  n/2}}\sum_{x_{1},\ldots,x_{k}\in\left\{
0,1\right\}  ^{n}}\operatorname*{E}\left[  g_{1}\left(  x_{1}\right)  \left(
-1\right)  ^{x_{1}\cdot x_{2}}g_{2}\left(  x_{2}\right)  \left(  -1\right)
^{x_{2}\cdot x_{3}}\cdots\left(  -1\right)  ^{x_{k-1}\cdot x_{k}}g_{k}\left(
x_{k}\right)  \right]  .
\]
By symmetry, the expectation inside the sum is $0$ if we condition on
$x_{i}\in S_{i}$\ for any $i$. \ Conversely, if we condition on $x_{i}\notin
S_{i}$\ for all $i$, then the expectation is%
\[
\Lambda_{x_{1},\ldots,x_{k}}:=f_{1}\left(  x_{1}\right)  \left(  -1\right)
^{x_{1}\cdot x_{2}}f_{2}\left(  x_{2}\right)  \left(  -1\right)  ^{x_{2}\cdot
x_{3}}\cdots\left(  -1\right)  ^{x_{k-1}\cdot x_{k}}f_{k}\left(  x_{k}\right)
.
\]
Overall, then, the expectation is%
\[
\Lambda_{x_{1},\ldots,x_{k}}\cdot\prod_{i\in\left[  k\right]  }\Pr_{S_{i}%
}\left[  x_{i}\notin S_{i}\right]  =\Lambda_{x_{1},\ldots,x_{k}}\cdot\left(
1-\varepsilon\right)  ^{k},
\]
which yields%
\[
\operatorname*{E}\left[  \Phi_{g_{1},\ldots,g_{k}}\right]  =\left(
1-\varepsilon\right)  ^{k}\Phi_{f_{1},\ldots,f_{k}}.
\]
Clearly, then, there exists at least one choice of $g_{1},\ldots,g_{k}$\ such
that%
\[
\Phi_{g_{1},\ldots,g_{k}}\leq\left(  1-\varepsilon\right)  ^{k}\Phi
_{f_{1},\ldots,f_{k}}.
\]

\end{proof}

In contrapositive form, Lemma \ref{proptest}\ implies that, if a $k$-tuple of
functions $\left\langle f_{1},\ldots,f_{k}\right\rangle $\ has Hamming
distance at least $\varepsilon\cdot k2^{n}$\ from every tuple $\left\langle
g_{1},\ldots,g_{k}\right\rangle $\ such that $\Phi_{g_{1},\ldots,g_{k}}\leq c$
(for some $c\geq0$), then we must have%
\[
\Phi_{f_{1},\ldots,f_{k}}>\frac{c}{\left(  1-\varepsilon\right)  ^{k}%
}>c\left(  1+k\varepsilon\right)  .
\]
As a side note, we conjecture that Lemma \ref{proptest} also holds
\textquotedblleft in the other direction\textquotedblright---that is, that if
$\left\langle f_{1},\ldots,f_{k}\right\rangle $\ has Hamming distance at least
$\varepsilon\cdot k2^{n}$\ from every $\left\langle g_{1},\ldots
,g_{k}\right\rangle $\ such that $\Phi_{g_{1},\ldots,g_{k}}\geq c$ (for some
$c>0$), then $\Phi_{f_{1},\ldots,f_{k}}$\ must be significantly
\textit{smaller} than $c$---but we do not currently have a proof of that.

We now show how to deal with the second difficulty. \ Call a $k$-tuple\ of
Boolean functions $f_{1},\ldots,f_{k}:\left\{  0,1\right\}  ^{n}%
\rightarrow\left\{  -1,1\right\}  $\ \textit{good} if%
\[
\left\vert \Phi_{f_{1},\ldots,f_{i-1},f_{i}^{\left(  x\right)  }}\right\vert
\leq\frac{\log N}{\sqrt{N}}%
\]
for every $i\in\left[  k\right]  $\ and $x\in\left\{  0,1\right\}  ^{n}$; and
moreover, there exists a constant $C_{k}$\ such that, for every $i\in\left[
k\right]  $ and $t\in\left[  \log N\right]  $, the \textquotedblleft partial
sums\textquotedblright\ $\Phi_{f_{1},\ldots,f_{i-1},f_{i}^{\left(  x\right)
}}$\ satisfy the following property:%
\[
\Pr_{x\in\left\{  0,1\right\}  ^{n}}\left[  \left\vert \Phi_{f_{1}%
,\ldots,f_{i-1},f_{i}^{\left(  x\right)  }}\right\vert \geq\frac{t}{\sqrt{N}%
}\right]  \leq\frac{C_{k}}{t^{t/2}}.
\]
Then we have the following extension of Proposition \ref{logub}:

\begin{proposition}
\label{logub2}If $\left\langle f_{1},\ldots,f_{k}\right\rangle $\ is chosen
uniformly at random, then it is good with\ probability at least $1-\delta_{k}%
$, where $\delta_{k}$\ can be made arbitrarily small by increasing $C_{k}$.
\end{proposition}

\begin{proof}
By Lemma \ref{phibound}, for all $i\in\left[  k\right]  $\ and $x\in\left\{
0,1\right\}  ^{n}$ we have%
\[
\Pr_{f_{1},\ldots,f_{i}}\left[  \left\vert \Phi_{f_{1},\ldots,f_{i-1}%
,f_{i}^{\left(  x\right)  }}\right\vert \geq\frac{t}{\sqrt{N}}\right]
=O\left(  \frac{1}{t^{t}}\right)  .
\]
So by Markov's inequality,%
\[
\Pr_{f_{1},\ldots,f_{i}}\left[  \Pr_{x}\left[  \left\vert \Phi_{f_{1}%
,\ldots,f_{i-1},f_{i}^{\left(  x\right)  }}\right\vert \geq\frac{t}{\sqrt{N}%
}\right]  >\frac{C_{k}}{t^{t/2}}\right]  =O\left(  \frac{1}{C_{k}t^{t/2}%
}\right)  .
\]
So by the union bound,%
\begin{align*}
\Pr_{f_{1},\ldots,f_{i}}\left[  \exists i,t:\Pr_{x}\left[  \left\vert
\Phi_{f_{1},\ldots,f_{i-1},f_{i}^{\left(  x\right)  }}\right\vert \geq\frac
{t}{\sqrt{N}}\right]  >\frac{C_{k}}{t^{t/2}}\right]   &  =O\left(  \frac
{k}{C_{k}}\sum_{t\in\left[  \log N\right]  }\frac{1}{t^{t/2}}\right) \\
&  =O\left(  \frac{k}{C_{k}}\right)  .
\end{align*}

\end{proof}

Furthermore, good $k$-tuples behave as we want for property-testing purposes.

\begin{lemma}
\label{proptest2}Let $\left\langle f_{1},\ldots,f_{k}\right\rangle $\ be a
good $k$-tuple.\ \ Then for all modifications $g_{1},\ldots,g_{k}$\ such that
$\Pr_{x}\left[  f_{i}\left(  x\right)  \neq g_{i}\left(  x\right)  \right]
\leq\varepsilon$\ for all $i\in\left[  k\right]  $, we have%
\[
\left\vert \Phi_{f_{1},\ldots,f_{k}}-\Phi_{g_{1},\ldots,g_{k}}\right\vert
=O\left(  k\sqrt{\varepsilon}\log\frac{1}{\varepsilon}\right)  .
\]

\end{lemma}

\begin{proof}
Recall the \textquotedblleft standard\textquotedblright\ quantum algorithm for
$k$-fold \textsc{Forrelation}\ (the one shown in Figure \ref{kfold}). By
direct analogy to the hybrid argument of Bennett, Bernstein, Brassard, and
Vazirani \cite{bbbv}, we consider what happens if, in that algorithm, we
replace the $U_{f_{i}}$\ oracles by $U_{g_{i}}$\ oracles one by one---starting
with $U_{f_{k}}$, and working backwards towards $U_{f_{1}}$. \ Let%
\[
\left\vert \psi_{i}\right\rangle =\sum_{x\in\left\{  0,1\right\}  ^{n}}%
\Phi_{f_{1},\ldots,f_{i-2},f_{i-1}^{\left(  x\right)  }}\left\vert
x\right\rangle
\]
be the state of the quantum algorithm immediately before the $i^{th}$\ oracle
call. \ Then by quantum-mechanical linearity, the entire sequence of oracle
replacements can change the final amplitude of the all-$0$ state,
$\alpha_{0\cdots0}$, by at most%
\begin{align*}
\sum_{i\in\left[  k\right]  }\left\Vert U_{f_{i}}\left\vert \psi
_{i}\right\rangle -U_{g_{i}}\left\vert \psi_{i}\right\rangle \right\Vert  &
=\sum_{i\in\left[  k\right]  }\sqrt{\sum_{x~:~f_{i}\left(  x\right)  \neq
g_{i}\left(  x\right)  }\left(  2\Phi_{f_{1},\ldots,f_{i-2},f_{i-1}^{\left(
x\right)  }}\right)  ^{2}}\\
&  \leq2\sum_{i\in\left[  k\right]  }\sqrt{\varepsilon N\left(  \frac
{3\log1/\varepsilon}{\sqrt{N}}\right)  ^{2}+\sum_{t=3\log1/\varepsilon}^{\log
N}\Pr_{x}\left[  \left\vert \Phi_{f_{1},\ldots,f_{i-2},f_{i-1}^{\left(
x\right)  }}\right\vert \geq\frac{t}{\sqrt{N}}\right]  \left(  \frac
{t+1}{\sqrt{N}}\right)  ^{2}}\\
&  \leq2\sum_{i\in\left[  k\right]  }\sqrt{9\varepsilon\log^{2}\frac
{1}{\varepsilon}+\sum_{t=3\log1/\varepsilon}^{\log N}\frac{C_{k}}{t^{t/2}%
}\left(  \frac{t+1}{\sqrt{N}}\right)  ^{2}}\\
&  \leq2\sum_{i\in\left[  k\right]  }\sqrt{9\varepsilon\log^{2}\frac
{1}{\varepsilon}+C_{k}\cdot O\left(  \varepsilon\right)  }\\
&  =O\left(  k\sqrt{\varepsilon}\log\frac{1}{\varepsilon}\right)  .
\end{align*}
But since $\alpha_{0\cdots0}$\ is precisely equal to $\Phi_{f_{1},\ldots
,f_{k}}$, this means that $\Phi_{f_{1},\ldots,f_{k}}$\ can change by at most
$O\left(  k\sqrt{\varepsilon}\log\frac{1}{\varepsilon}\right)  $\ as well.
\end{proof}

In contrapositive form, Lemma \ref{proptest2}\ implies that if $\Phi
_{g_{1},\ldots,g_{k}}\geq\frac{3}{5}$, then for every good $k$-tuple
$\left\langle f_{1},\ldots,f_{k}\right\rangle $\ with $\left\vert \Phi
_{f_{1},\ldots,f_{k}}\right\vert \leq\frac{1}{100}$, there must be an
$i\in\left[  k\right]  $ such that $g_{i}$\ differs from $f_{i}$\ on
a\ $\Omega\left(  \frac{1}{k^{2}\log^{2}k}\right)  $\ fraction of points.

We now put everything together. \ Recall the property-testing problem, of
distinguishing $\mathcal{Y}$\ (the set of all $\left\langle f_{1},\ldots
,f_{k}\right\rangle $\ such that $\Phi_{f_{1},\ldots,f_{k}}\leq\frac{1}{100}$)
from $\mathcal{N}_{\varepsilon}$\ (the set of all $\left\langle g_{1}%
,\ldots,g_{k}\right\rangle $\ that are at least $\varepsilon$\ away from
$\mathcal{Y}$\ in the Hamming distance sense). \ Lemma \ref{proptest} implies
that, just by taking the quantum algorithm of Proposition \ref{inpromisebqp}%
\ and amplifying it a suitable number of times, we can solve this problem,
with error probability at most (say) $1/3$, using only $O\left(
1/\varepsilon\right)  $\ quantum queries.\footnote{Na\"{\i}ve repetition would
give $O\left(  1/\varepsilon^{2}\right)  $ queries, but we can get down to
$O\left(  1/\varepsilon\right)  $\ using amplitude amplification.}
\ Furthermore, if we only need to distinguish $\mathcal{Y}$\ from
$\mathcal{N}_{\varepsilon}$\ with $\Theta\left(  \varepsilon\right)  $\ bias,
then it suffices to make just $1$\ quantum query.

On the other hand, suppose we set $\varepsilon\leq\frac{C}{k^{3}\log^{2}k}$,
for some suitably small constant $C$. \ Then Lemma \ref{proptest2} implies
that, if we had a randomized algorithm to distinguish $\mathcal{Y}$\ from
$\mathcal{N}_{\varepsilon}$ with constant bias, then we could also use that
algorithm to distinguish the good tuples $\left\langle f_{1},\ldots
,f_{k}\right\rangle $\ with $\left\vert \Phi_{f_{1},\ldots,f_{k}}\right\vert
\leq\frac{1}{100}$\ from the tuples $\left\langle g_{1},\ldots,g_{k}%
\right\rangle $\ with $\Phi_{g_{1},\ldots,g_{k}}\geq\frac{3}{5}$. \ But such
an algorithm would solve the distributional version of $k$-fold
\textsc{Forrelation}, and we already showed (in Theorems \ref{thm1}\ and
\ref{kfoldlbthm} respectively) that any such algorithm requires $\Omega
(\frac{\sqrt{N}}{\log N})$\ queries for $k=2$\ or $\Omega(\frac{\sqrt{N}}%
{\log^{7/2}N})$\ queries for $k>2$.

\end{document}